\documentclass[notitlepage,11pt]{article}

\usepackage{amsfonts}
\usepackage{geometry}
\usepackage[onehalfspacing]{setspace}
\usepackage{graphicx}
\usepackage{amssymb}
\usepackage{amsthm,amsmath}
\usepackage[dvipsnames]{xcolor}
\usepackage {ulem}
\usepackage{enumerate}
\usepackage{enumitem}
\usepackage{dsfont}
\usepackage{natbib}
\usepackage{hyperref}
\usepackage{multirow}
\usepackage{tablefootnote}
\usepackage{threeparttable}
\usepackage{tabularx}
\usepackage{multicol}
\usepackage{pdfpages}
\usepackage{tikz}
\usepackage{pgfplots}
\usepackage{comment}
\usepackage{multibib}
\newcites{app}{Appendix References} 

\pgfplotsset{compat=1.17}
\usepackage{subcaption}

\hypersetup{colorlinks=true, linkcolor=red,          
    citecolor=blue,        
    filecolor=magenta,      
    urlcolor=cyan           
}

\usepackage[T1]{fontenc}

\usepackage{tgtermes}
\usepackage{cmbright}

\DeclareMathOperator*{\argmax}{arg\,max}

\newtheorem{theorem}{Theorem}
\newtheorem*{theorem*}{Theorem}

\newtheorem{algorithm}[theorem]{Algorithm}

\newtheorem{claim}{Claim}

\newtheorem{corollary}{Corollary}

\newtheorem{definition}{Definition}
\newtheorem{example}{Example}

\newtheorem{lemma}{Lemma}

\newtheorem{proposition}{Proposition}
\newtheorem{proposition*}{Proposition}
\newtheorem{observation}{Observation}
\newtheorem{remark}[theorem]{Remark}

\usepackage{algorithm}
\usepackage{algpseudocode} 
\algdef{SE}[IF]{IfNot}{EndIf}[1]{\algorithmicif\ \algorithmicnot\ #1\ \algorithmicthen}{\algorithmicend\ \algorithmicif}%

\usepackage{bm}
\usepackage{chngcntr,etoolbox}
\newcounter{resetdummycounter}
\newcommand{\resetcounterlist}[1]{%
  \renewcommand*{\do}[1]{\counterwithin*{##1}{resetdummycounter}}%
  \docsvlist{#1}}
\newcommand{\resetcounters}{\stepcounter{resetdummycounter}}
\resetcounterlist{definition, theorem, proposition, corollary}

\begin{document}

\title{Rawlsian Assignments\thanks{We are especially grateful to William Thomson, the Editor Scott Duke Kominers, and two anonymous
referees for their detailed comments and useful suggestions. 
We thank (in alphabetic order) Christian Basteck, Anna Bogomolnaia, Hector Cancela, Juan Dubra, Conal Duddy, Federico Echenique, Lars Ehlers, Di Feng, Marcelo Fernandez, Alan Frieze, Bettina Klaus, Thilo Klein, Vikram Manjunath, Hervé Moulin, Afshin Nikzad, Josué Ortega, Joaquin Paleo, Szilvia P\'apai, William Phan, Antonio Romero-Medina, Yuki Tamura, Alex Teytelboym, Juan Pablo Torres-Martinez, and Kemal Yildiz, as well as the conference/seminar participants at UNSL - IMASL, the 13th Conference on Economic Design, FEN - Universidad de Chile, the SEU 2022 meeting, the MATCH-UP 2022 workshop, 2022 Latin American Workshop in Economic Theory, 2022 Latin American Meeting of the Econometric Society (LAMES), the 2024 Matching in Practice workshop in Zürich, CEA - Universidad de Chile, Department of Economics FCS - UdelaR, Concordia University, and Universidad de Montevideo, for their comments and suggestions. Tom Demeulemeester gratefully acknowledges the support by Research Foundation -- Flanders under grant \mbox{11J8721N} and by the Swiss National Science Foundation (SNSF) through Project 100018$\_$212311. Financial support from ANII, FMV\_1\_2021\_1\_166576, is gratefully acknowledged. We also thank the members of the MTAV project (https://github.com/eze91/MTAV) for sharing the data with us, and Rodrigo Gonzalez for his excellent research assistance.}}
\author{Tom Demeulemeester\thanks{Department of Quantitative Economics, Maastricht University,  Tongersestraat 53, 6211 LM Maastricht, Netherlands.  (email: tom.demeulemeester@maastrichtuniversity.nl)} \ \ \ \ \ \ \ \ \ \ Juan S. Pereyra\thanks{Universidad de Montevideo, Prudencio de Pena 2440, Montevideo, Uruguay. (email:\ jspereyra@um.edu.uy) }}

\maketitle

\begin{abstract}
We study the assignment of indivisible goods to individuals without monetary transfers. Existing literature has mainly focused on efficiency and \textit{individually} fair assignments; consequently, egalitarian concerns have been overlooked. Drawing inspiration from the allocation of apartments in housing cooperatives—where families prioritize egalitarianism in assignments—we introduce the concept of Rawlsian assignment. We show that the Rawlsian assignment is unique and that the rule it induces is efficient and anonymous. The Rawlsian rule is not sd-strategy-proof, reflecting a fundamental incompatibility between worst-rank minimization and even weak incentive requirements. Nevertheless, it satisfies several robustness properties that limit other forms of manipulation.
 We illustrate our findings using cooperative housing preference data, showing significant improvements in egalitarian outcomes over both the probabilistic serial rule and the currently employed rule.
\bigskip

JEL\ classification: C78, D63.

\end{abstract}


\section{Introduction} 
We study the problem of allocating a set of indivisible goods to agents when monetary transfers cannot be used. The indivisibility of objects makes it 
generally impossible to achieve fairness from an ex-post perspective. As a remedy, we adopt an ex-ante perspective where each agent receives a random allocation, i.e., a lottery over the set of objects. Inspired by the Rawlsian concept of fairness, we propose a new solution concept, examine its properties, and evaluate its performance using real data.

Our primary motivation is the assignment of apartments in housing cooperatives. Housing cooperatives typically involve a group of families who participate in the construction of a building. Upon completion of the building, the problem of distributing the apartments among the families arises. The use of prices is, in general, not allowed by regulation, and the assignment process relies on the ordinal preferences of the families. In Uruguay, the regulation also states that apartments should be assigned randomly; however, participants have the option of using another procedure if they unanimously agree on it. For the Uruguayan housing cooperatives we study, the main concern of the members is the \textit{egalitarianism} of the final assignment. They aim to avoid unequal situations where, for example, some families receive their most preferred apartments while others rank their assignments very low. The spirit of this principle is not exclusive to Uruguay; it is present in many other cases as well (for more instances, we refer the reader to www.housinginternational.coop). Another example where a similar concern may exist is the assignment of faculty offices to professors in a university department. 

There have been many important contributions in the literature regarding the efficiency and envy-freeness of the assignments. However, little attention has been paid to \textit{egalitarian} concerns. In this paper, we introduce a new concept to address this issue and examine its properties. Furthermore, using data from several housing cooperatives, we demonstrate that the solution we propose is more egalitarian from a Rawlsian perspective compared to the currently implemented one. Thus, the Rawlsian solution we recommend improves upon the outcomes achieved by the current rule.

We adopt Rawls' concept of justice and assess an assignment based on the well-being of the worst-off individuals. The rank of an object on an agent's preferences is the position of the object in her preferences. Thus, a higher rank indicates lower satisfaction (for example, the most preferred object has rank $1$). 
The analysis begins by identifying, for a given assignment, who the worst-off agents are. If we were considering deterministic assignments, we could examine the rank of the assigned object. The worst-off agents would be those associated with the highest rank. Subsequently, we could select assignments that minimize the highest rank (i.e., maximize the \textit{satisfaction} of the worst-off agents). In general, there are multiple assignments that satisfy this criterion, allowing for recursive application of the same criterion, as proposed by \cite{sen2017collective}.

When we expand the analysis to include random assignments, defining the worst-off agents becomes more challenging. If agents' cardinal utilities were known, one could consider the expected utility of each agent and identify the agent with the lowest expected utility (after some normalization). However, in our problem, we only have access to agents' ordinal preferences. So instead, for each agent, we determine the rank of her least preferred object among those she receives with positive probability. We then focus on the agents for whom this rank is the highest. Among these agents, the worst-off individuals are those who receive the object with the highest probability.

To compare two assignments, we first consider the worst-off agent, that is, the one who receives her least preferred object with the highest probability. If the probabilities are the same under both assignments, we turn to the agent who receives her least preferred object with the second-highest probability. If all the probabilities associated with the least preferred object of each agent are the same, we conduct the same comparison with the two least preferred objects of each agent. If, at any point, the probability in the first assignment is strictly higher than the corresponding probability in the second assignment, we say that the first is Rawlsian-dominated by the second. An assignment is considered Rawlsian if it is not Rawlsian-dominated by any other assignment.

We first show that, in any problem, there always exists a \textbf{unique} Rawlsian assignment. Roughly, if there were two Rawlsian assignments, one would consider the average assignment. In this assignment, agents who receive (with positive probability) the objects associated with the highest rank would be assigned them with a lower probability. 

We extend preferences over objects to preferences over lotteries using first-order stochastic dominance (sd). In words, an agent sd-prefers a lottery if it guarantees her a weakly higher probability of receiving her most preferred
object, a weakly higher probability of receiving her two most preferred objects, and so on. An assignment sd-dominates another assignment if for every agent, the lottery she receives in the first assignment is sd-preferred to the lottery in the second assignment. 
We prove that the Rawlsian rule is \textbf{sd-efficient} and \textbf{anonymous}: no assignment sd-dominates the Rawlsian assignment, and agents’ allocations do not depend on their names. 

We also study the incentive properties of the Rawlsian rule. Although the Rawlsian rule is not sd-strategy-proof, we show that this limitation is not specific to our construction: there is a fundamental incompatibility between worst-rank minimization and even weak strategy-proofness requirements. At the same time, the Rawlsian rule satisfies several robustness properties that are closely connected to its objective. It satisfies group lower invariance: groups of agents cannot affect any agent's probability of receiving low-ranked objects by changing only the relative order of higher-ranked objects. It also satisfies consistency and non-bossiness. These properties distinguish the Rawlsian rule from standard alternatives such as probabilistic serial and random serial dictatorship, as well as from other rules with an egalitarian objective, such as the positive and prudent equality rules of \cite{duddy2025egalitarian}.

The construction of our concept resembles the welfarist definition of the classic probabilistic serial rule \citep{bogomolnaia2001new}, proposed by \cite{bogomolnaia2015random,bogomolnaia2018most}. The author describes the probabilistic serial (PS) rule in terms of the cumulative probabilities of each agent being assigned her top $k$ objects. She shows that the PS rule is the unique rule that lexicographically maximizes the cumulative probabilities for all agent-preference pairs $(i,k)$. The Rawlsian criterion is based on a similar procedure but starts from the least preferred objects and minimizes their probability. Moreover, the Rawlsian rule lexicographically minimizes the cumulative probabilities preference by preference, instead of lexicographically minimizing the entire vector of cumulative probabilities.

Next, we study how to compute the Rawlsian assignment for a given problem. We introduce an algorithm to compute it in polynomial time by solving at most $ \frac{n^3+n^2}{2}$ linear programs, where $n$ is the number of agents.

We illustrate our analysis using a set of preferences from housing cooperatives in Uruguay. We compare our solution with the outcomes of the PS rule, and the rule that is currently used, called the \textbf{MTAV} from the Spanish \textit{Mejor Tecnología de Asignación de Viviendas}
(``Best Housing Assignment Technology''), proposed by \cite{prino2016optimal}. Roughly, the MTAV runs as follows. First, it selects all deterministic assignments that minimize the maximum rank. 
Second, among these assignments, it considers those that maximize the sum of the families' utilities, assuming that the utility of the apartment ranked in position $k$ is $n-k$ for each family. Finally, if multiple assignments are selected, one is randomly chosen (see Appendix \ref{MTAV} for a formal description). 

The general finding is that the Rawlsian solution significantly improves the assignment of the worst-off families. For instance, under the PS assignment, there is at least one family who receives their least preferred apartment with positive probability in all but two out of 24 cooperatives. In contrast, with the Rawlsian assignment, only three cooperatives have a family that receives their least preferred apartment with positive probability. Although there are both \textbf{sd-winners} and \textbf{sd-losers} when transitioning from the PS rule to the Rawlsian rule---that is, some families receive a Rawlsian lottery that sd-dominates their PS lottery, while others receive a PS lottery that sd-dominates their Rawlsian lottery---the number of families in the former group is larger in all evaluated instances.
Regarding the MTAV, by construction, its maximum rank (the rank of the assigned apartment in the preferences of the worst-off family) coincides with that of the Rawlsian rule. However, there are cases where the Rawlsian rule assigns fewer families to their least preferred apartment.

Another measure to compare the assignments is the expected number of families assigned apartments ranked in the first $k$ positions, for each $k \in \{1,\ldots,n\}$. For a fixed cooperative and a fixed random assignment, this expectation is taken over the lottery induced by the
assignment rule; equivalently, it is the sum of the assignment probabilities with which families receive apartments ranked in the first $k$ positions. The Rawlsian rule assigns a lower expected number of families to their least preferred apartments compared to the PS rule. It also assigns a lower expected number of families to their top choices, especially their first choice. The MTAV falls between these two distributions. To illustrate this, Figure \ref{fig:intro} displays the rank distribution graph for a cooperative with 28 families.
Therefore, our solution provides an alternative to both the PS and MTAV that significantly improves upon them from an egalitarian perspective.

Finally, we complement the empirical findings by analyzing the maximum rank of the Rawlsian rule in large markets. We consider markets of size $n$, where agents' preferences are drawn i.i.d.\ from a uniform distribution, and study the limit of the expected maximum rank as $n$ tends to infinity. Although the average maximum rank of the Rawlsian rule tends to infinity, we show that it grows at a slow rate (it is upper bounded by $\lfloor\ln(n)\rfloor$ plus a constant). In a market of size \(n = 1000\), for example, our result implies that the expected rank of the least preferred assigned object is at most~9. For the PS rule, however, simulations suggest that the expected maximum rank in a random market of size~$n$ is very close to $n$.


\begin{figure}
  \centering
  \includegraphics[width=0.6\linewidth]{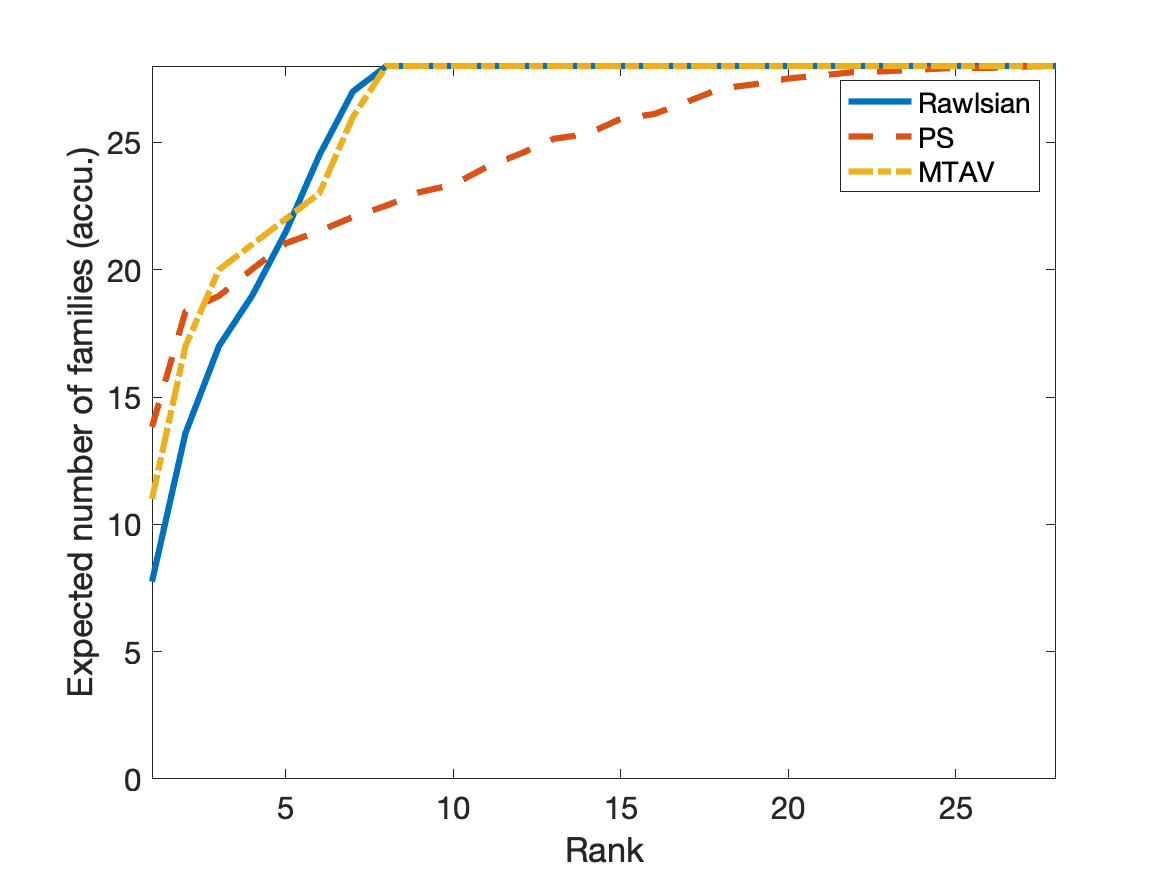}
  \caption{Cumulative rank distribution for a housing cooperative with 28 families.}
  \label{fig:intro}
\end{figure}

In the next section, we place our contributions within the related literature. Section \ref{def} presents the model and definitions. In Section \ref{rawls}, we define the concept of a Rawlsian assignment. Section \ref{results} contains our main results and relates the Rawlsian assignment to other concepts in the literature. In Section \ref{algorithm}, we describe an algorithm to find the Rawlsian assignment for a given problem, and in Section \ref{app} we use it to illustrate the results with data from different housing cooperatives. Section \ref{section_large} includes the analysis for large markets. We conclude in Section \ref{conclu}. All the omitted proofs are in the Appendix and the Online Appendix.

\section{Related Literature}

Since the introduction of the assignment problem by \cite{hylland1979efficient}, there have been numerous significant contributions related to sd-efficiency, sd-envy-freeness, and sd-strategyproofness. Notable among these is the random serial dictatorship (or random priority) rule introduced by \cite{zhou1990conjecture} and \cite{abdulkadirouglu1998random}. While this rule is sd-strategyproof (it is immune to individual preferences manipulation), its outcome may be stochastically dominated with respect to individual preferences, thus lacking sd-efficiency. In response, \cite{bogomolnaia2001new} proposed the PS rule, which is sd-efficient but not sd-strategyproof. Moreover, the PS rule is sd-envy-free: the lottery each agent receives stochastically dominates with respect to individual preferences the lottery of every other agent. The random serial dictatorship rule is not sd-envy-free. In this paper, we introduce a Rawlsian criterion which induces a well-defined rule that is sd-efficient, as is PS, while satisfying an egalitarian requirement. 

The Rawlsian idea of justice, as described by \cite{rawls2020theory}, has found applications in two-sided matching markets. For instance, \cite{masarani1989existence}, \cite{romero2005equitable}, and \cite{kuvalekar2024fair} explore the compatibility of this concept of justice with stability in a marriage market. They adopt the Rawlsian criterion to select a stable matching that treats both sides of the market ``symmetrically.'' In a school choice model, \cite{galichon2023stable}  show that stability often comes at the cost of extreme forms of inequality. The assignment problem we study differs from the literature mentioned above. There are agents on one side and objects on the other, objects do not possess priorities, and the notion of stability does not apply. \cite{klaus2010smith} adapt the Rawlsian criterion to the roommate problem and demonstrate its compatibility with stability. Another way to compare the inequality of different assignments is by applying the Lorenz order. This method has been used in our framework by various studies, such as \cite{pycia2017incentive} and \cite{harless2018learning}.

In the same framework as ours, \cite{afacan2024rawlsian} study a Rawlsian notion of fairness but consider \textbf{only} deterministic assignments. They define an assignment to be more Rawlsian than another if the assignment ranking of the worst-off agent is lower in the former than in the latter; and in the case of equality, the size of the worst-off agents group is smaller in the former. For the case of deterministic assignments, our notion of Rawlsian assignment is a refinement of their notion. However, expanding the analysis to probabilistic assignments allows us to show, in contrast to \cite{afacan2024rawlsian}, the existence of a unique Rawlsian assignment. \cite{ortega2023cost} compare different rules in school choice in terms of their expected maximum rank, where the expectation is taken over the random draw of markets in their model. They show that the deferred acceptance and top trading cycles rules perform poorly on this measure, being outperformed by the rule that minimizes the sum of ranks for students. \cite{ortega2024unimprovable} define the Rawlsian inequality of an assignment as the ratio between the maximum rank of the assignment and that of the Rawlsian assignment. The authors compare the deferred acceptance rule and the efficiency-adjusted deferred acceptance rule of \cite{kesten2010school} in terms of their Rawlsian inequality.

The study by \cite{duddy2025egalitarian} is the closest to our paper. Duddy introduces an alternative notion of egalitarianism, which he calls even-handedness. He argues that an agent is a least advantaged agent at a given assignment if her cumulative probabilities of being assigned to her top $k$ objects are weakly lower than the cumulative probabilities of every other agent, for all values of $k$. Loosely speaking, an assignment $x$ is even-handed if there does not exist another assignment~$y$ such that a least advantaged agent at $y$ is strictly better off (in the stochastic dominance sense) at $y$ than at $x$.

Note that a given profile may admit multiple even-handed assignments. \cite{duddy2025egalitarian} proposes several rules that select a unique even-handed and sd-efficient assignment for every profile, including the ``positive equality'' rule and the ``prudent equality'' rule. 
The positive equality rule uses the following order to compare assignments. 
It starts by considering, for each assignment, the probability with which each agent receives her most preferred object, and takes the lowest probability. The rule selects the assignment for which the lowest probability is the highest. If the two lowest probabilities are the same, it moves on to the agent who receives her two most preferred objects with the lowest total probability, and selects the assignment for which the lowest total probability is the highest, and so on. The positive equality rule chooses the maximum assignment based on this order. On the other hand, the prudent rule performs a similar comparison but starts by considering the agent who receives, at each assignment, her least preferred object with highest probability. Then, it selects the assignment for which the highest probability is the lowest. \textit{When the two assignments coincide, it then moves to the agent who receives her least and second-least preferred objects with the highest total probability, and so on.} 

In contrast, the Rawlsian assignment begins by considering the agent who receives her least preferred object with the largest probability and aims to minimize this probability (same as the prudent equality rule). However, in the case of ties, instead of considering the least and second-least preferred objects, it moves to the agent who receives her least preferred object with the second-highest probability. Intuitively, our proposal focuses on the worst-off agent defined as the agent who receives her least preferred object with the largest probability, and then moves to the second worst-off agent, defined as the agent who receives her least preferred object with the second-highest probability.
The Rawlsian assignment does not satisfy Duddy's even-handedness criterion, and there are assignments that satisfy even-handedness that are not Rawlsian. Thus, the two concepts are independent. For a more detailed discussion, we refer to Appendix \ref{conal}. 

Also related is the paper by \cite{feizi2022distributive} who introduces a notion of inequality in random assignments. The key element of his definition is the number of agents who envy other agents in a representation of the assignment as a lottery over deterministic assignments. The concept aims to minimize this envy, which is different from our approach. 

The notion of Rawlsian assignments is also related to the downward lexicographic extension of stochastic dominance \citep{cho2016equivalence, cho2018probabilistic}. Unlike \cite{cho2018probabilistic}, who compares two lotteries from an individual perspective, we use a similar criterion but apply it to the whole assignment. Also, \cite{cho2016equivalence} show that upward lexicographic (ul) efficiency and sd-efficiency are equivalent. A Rawlsian assignment is sd-efficient (Proposition \ref{sd-eff}), and thus, also ul-efficient, but not the other way around.
Finally, other papers have also used the rank distribution to define different concepts. \cite{featherstone2020rank} introduces rank efficiency. An assignment is rank efficient if its rank distribution cannot be stochastically dominated. We prove in the Online Appendix (Section \ref{rank}) that there is no relation between Rawlsian and rank efficient assignments.

\section{Primitives and definitions}\label{def}
Let $I=\{1,\ldots,n\}$ be the set of agents and $O$ the set of objects (with $\vert O \vert=n$). Each agent $i \in I$ has (strict) preferences over the set of objects, denoted by $\succeq_i$. 
A preference profile is denoted by $\succeq =(\succeq_i)_{i \in I}$. Sometimes we represent agent $i$'s preferences $\succeq_i$ as an $n$-dimensional vector $r_i(\succeq_i)\in\{1,\ldots,n\}^n$, where $r_{io}(\succeq_i)=k$ means that object $o$ is ranked $k$-th by agent $i$ at $\succeq_i$. We denote $r_{io}(\succeq_i)$ by $r_{io}$ when agent $i$'s preferences are clear from the context. The assumption of strict and complete preferences fits our main motivation: families in the cooperatives have to rank all the apartments, and ties in preferences are not allowed. 
An assignment problem is a tuple $(I,O,\succeq)$. We fix the sets of agents and objects, and define a problem as a preference profile $\succeq$.

A solution to an assignment problem is a (random) assignment $x=(x_i)_{i \in I}$, where each $x_i$ is a probability distribution over $O$, and for every object $o \in O$, $\sum_{i \in I} x_{io}=1$. We interpret $x_{io} \in [0,1]$ as the probability with which agent $i$ is allocated object $o$. An assignment is deterministic if all of its entries are either zero or one. The Birkhoff-von Neumann theorem \citep{birkhoff1946tres, vonNeumann1953} states that any random assignment can be written as a convex combination of deterministic assignments.

We usually describe an assignment by a matrix where the rows are indexed by the agents and the columns are indexed by the objects, and for each $i$, row $i$ is the lottery agent $i$ receives. Let $X$ be the set of assignments (or equivalently, the set of $n\times n $ bi-stochastic matrices, that is, matrices of non-negative real numbers in which the entries in each row sum to $1$, and the entries in each column sum to $1$). 
A \textbf{rule} is a function $\phi$ which maps every problem to an assignment: for every $\succeq$, $\phi(\succeq) \in X$. We will refer to the lottery assigned to agent $i$ -- given by the $i$-th row of an assignment -- as her allocation. We denote by $\phi_i(\succeq)$ the allocation of agent $i$ by rule $\phi$ in problem $\succeq$.

We define the \textbf{maximum rank} of an assignment as the worst rank of the objects to which the agents are assigned with strictly positive probability. Formally, the maximum rank of an assignment $x\in X$ equals \[\max_{(i,o)\in I\times O}\{r_{io}|x_{io}>0\}.\]

An assignment has minimal maximum rank if there does not exist another assignment with a strictly lower maximum rank. A rule minimizes maximum rank if it outputs an assignment with minimal maximum rank.

We will use the following concept of efficiency for random assignments due to \cite{bogomolnaia2001new}.

\begin{definition}{\ }
\label{def:sd}
\begin{enumerate}
    \item An allocation $x_i$ for agent $i$ \textbf{first-order stochastically} dominates (sd-dominates) another allocation $x_i'$ if, for each $o \in O$:

$$
\sum_{o' : o' \succeq_i o} x_{io'} \geq \sum_{o' : o' \succeq_i o} x'_{io'}. 
$$

In this case, we use the notation $x_i \succeq_i^{sd} x'_i$.

\item An assignment $x$ \textbf{stochastically dominates} another assignment $x'$ if for every agent $i$, $x_i \succeq_i^{sd} x'_i$, and $x \neq x'$. 

\item An assignment is \textbf{sd-efficient} if no other assignment stochastically dominates it. 
\end{enumerate}
\end{definition}

The standard notion of fairness requires that each agent should find her allocation at least as desirable as anyone else's allocation. 

\begin{definition}
Given a problem $\succeq$, an assignment $x$ is \textbf{sd-envy-free} for $\succeq$ if for all $i,j \in I$, we have $x_i \succeq_i^{sd}  x_j$.
\end{definition}

If instead we require that there is no other agent's allocation that sd-dominates the agent's allocation, we have a weaker notion, called \textbf{weak sd-envy-freeness}.
\medskip

The following requirement is that assignments should be independent of the agents' names.

\begin{definition}
Given a problem $\succeq=(\succeq_i)_{i \in I}$ and a permutation $\pi:I\rightarrow I$, a rule $\phi$ satisfies \textbf{anonymity} if for each $i \in I$:

$$\phi_i((\succeq_{\pi(i)})_{i\in I})=\phi_{\pi(i)}((\succeq_{i})_{i \in I}).$$

\end{definition}

Anonymity implies a weaker notion of fairness called \textbf{equal treatment of equals}, according to which agents with the same preferences should obtain the same allocation.
\medskip



Finally, we consider the requirement of immunity to misrepresentations of individual preferences. 

\begin{definition}
A rule $\phi$ is \textbf{sd-strategyproof} if at any preference profile $\succeq$ no agent benefits by misreporting her preferences: for each $\succeq$, for each $i \in I$ , and for each $\succeq'_i$: 

$$
\phi_i(\succeq) \ \ \succeq_i^{sd} \ \ \phi_i(\succeq'_i,\succeq_{-i}).
$$

\end{definition}

As with sd-envy-freeness, if in the last definition we require that there be no manipulation such that its outcome sd-dominates the allocation the individual gets from truth-telling, we get a weaker notion, called \textbf{weak sd-strategyproofness.} It is weaker because it allows the allocation the agent gets when she manipulates to not be comparable to her allocation under truth-telling.

We will compare our proposal with the PS rule due to \cite{bogomolnaia2001new} which is defined as follows. Given a problem, think of each object as an infinitely divisible good with unit supply. \textbf{Step 1}: All agents start by consuming probabilities of receiving their most preferred object at the same unit speed. Proceed to the next step when an object is completely exhausted. \textbf{Step $s \geq 2$}: All agents consume probabilities of receiving their remaining most preferred object at the same speed. Proceed to the next step when an object is completely exhausted. 
The procedure terminates when each agent has consumed exactly one total unit of objects (i.e., at time one). The allocation of an agent $i$ is given by the amount of each object she has consumed until the algorithm terminates.

\section{Rawlsian assignments}\label{rawls}

In this section, we define our main concept of Rawlsian assignments. Given an assignment $x$, we denote by $b^x_i (k)$ the total probability with which agent $i$ gets objects with a rank between $n$ and $k$ in her preferences. That is:

$$
b^x_i (k)= \sum_{o \in O} \mathds{1} \{r_{io} \geq k \} x_{io}.
$$

In particular, $b^x_i(n)$ is the probability with which agent $i$ receives her least preferred object. By definition, $b^x_i(1)=1$. We denote by $b^x_i$ the vector of the cumulative probabilities from the least to the most preferred object: $b^x_i =(b^x_i(n),b^x_i(n-1),\ldots, b^x_i(1)= 1)$. Note that this is the opposite direction from the cumulative probability representation used in the definition of stochastic dominance (Definition \ref{def:sd}), where probabilities are accumulated from the most to the least preferred object:
\[
\sum_{o\in O} \mathds{1}\{r_{io}=1\}x_{io},\quad
\sum_{o\in O} \mathds{1}\{r_{io}\leq 2\}x_{io},\quad \ldots,\quad
\sum_{o\in O} \mathds{1}\{r_{io}\leq n\}x_{io}=1.
\]


Given an assignment $x$, and the vectors $(b^x_i)_{i \in I}$, we define the vector $B^x \in [0,1]^{n^2}$ as follows.

\begin{enumerate}
    \item The first elements $(B^x_1, \ldots, B^x_n)$ are the elements $(b^x_1(n), \ldots, b^x_n(n))$ listed in a non-increasing order. 
    
    \item Elements  $(B^x_{n+1}, \ldots, B^x_{2n})$ are the elements $(b^x_1(n-1), \ldots, b^x_n(n-1))$ listed in a non-increasing order. 
    
    \item In general, elements $(B^x_{(k-1)n+1}, \ldots, B^x_{kn})$ for $k=1,\ldots, n$, are the elements $(b^x_1(n-k+1), \ldots, b^x_n(n-k+1))$ listed in a non-increasing order. 
\end{enumerate}

Vector $B^x$ describes the cumulative distribution of probabilities induced by assignment $x$, from the least preferred object of each agent, to her most preferred object. 
The first $n$ entries of the vector are the probabilities with which each agent receives her least preferred object. And, in particular, the first entry is the highest among the probabilities with which agents receive their least preferred object. 

The following example illustrates the definition.

\begin{example}\label{ex1}
Let $I=\{1,2,3\}$, $O=\{a,b,c\}$, and preferences as follows:

\[
\begin{array}{c|ccc}
\succeq_1 & a & b & c \\ 
\succeq_2 & a & b & c \\ 
\succeq_3 & b & c & a \\ 
\end{array}
\]


Consider the assignment:
\begin{equation*}
x = 
\begin{pmatrix}
\frac{1}{2}& \frac{1}{2} & 0 \\
\frac{1}{2} & \frac{1}{2}  & 0  \\
0 &  0 & 1
\end{pmatrix}.
\end{equation*}

Recall that the entries of $b_i^x$ are ordered from the least preferred object to the most preferred object. Then: $b^x_1=b^x_2=\left(0,\frac{1}{2},1\right)$, $b^x_3=(0,1,1)$, and $B^{x}=\left(0,0,0,1,\frac{1}{2},\frac{1}{2},1,1,1\right).\hfill\qed$

\end{example}

To compare two assignments, $x$ and $y$, we first compare the highest probability with which an agent receives her least preferred object, that is, the first entry of vectors $B^x$ and $B^y$. If the probability is the same under both assignments, we compare the second-highest probability with which an agent receives her least preferred object. If all probabilities associated with the least preferred object of each agent are equal, we conduct the same comparison for the probabilities with which agents receive the least and second-least preferred objects. If at some point, an entry of $B^x$ is lower than the corresponding entry of $B^y$, we say that $x$ Rawlsian-dominates $y$. Formally, we compare the vectors $B^x$ and $B^y$ lexicographically. This is the idea of the following definition. 

\begin{definition}
Given two assignments $x$ and $y$, $x$ \textbf{Rawlsian-dominates} $y$ ($x$ R-dominates  $y$) if there is $j\in \{1,\ldots,n^2\}$ such that $B^{x}_j<B^y_j$, and for all $i<j$, $B^{x}_i=B^{y}_i$.
\end{definition}

The Rawlsian comparison is lexicographic, but it differs from the standard leximin order applied to the collection of all numbers $\{b_i^x(k):i\in I,\ k=1,\ldots,n\}$. The leximin order compares two vectors by first arranging the entries of each vector in non-decreasing order, then comparing the smallest entries. If the smallest entry is larger in one of the two vectors, then that vector lexicographically dominates the other; if the smallest entries coincide, it compares the second-smallest entries, and so on. In contrast, our Rawlsian comparison first groups the entries $\{b_i^x(k)\}$ by rank cutoff. We compare all agents' probabilities of receiving their least preferred object before considering the probabilities of receiving one of their two least preferred objects, and so on. Within each rank cutoff, agents'
probabilities are ordered from highest to lowest. Thus, the ordering first
prioritizes the rank cutoff and then, within each cutoff, the worst-off agents.

This construction is also different from the characterization of the PS rule in \cite{bogomolnaia2015random,bogomolnaia2018most}. That characterization is based on cumulative probabilities from the top of the preference ranking and lexicographically maximizes the resulting list of agent-rank cumulative probabilities (i.e., maximizes the smallest entry, then the second-smallest, etc.). In our notation, PS is the assignment that lexicographically maximizes the collection of numbers $\{1-b_i^x(k):i\in I,\ k=1,\ldots,n\}$ over all feasible assignments in $x\in X$. By contrast, our Rawlsian comparison uses cumulative probabilities from the bottom of the ranking and lexicographically minimizes them cutoff by cutoff.

The following is the key concept of our analysis.  

\begin{definition}
\label{def:rawlsian}
An assignment $x$ is \textbf{Rawlsian} if it is not Rawlsian-dominated (R-dominated) by any other assignment. \end{definition}


Example \ref{ex2} shows a Rawlsian assignment for the problem in Example \ref{ex1}.

\begin{example}\label{ex2}

In the problem of Example \ref{ex1}, consider the assignment:

\begin{equation*}
y = 
\begin{pmatrix}
1 & 0 & 0 \\
0 & 0  & 1  \\
0 &  1 & 0
\end{pmatrix},
\end{equation*}
\noindent
to which is associated the vector $B^y=(1,0,0,1,0,0,1,1,1)$. Clearly, $x$ R-dominates $y$. Moreover, we claim that  $x$ is a Rawlsian assignment. First, note that if another assignment $w$ is such that $w_{3c}<1$, then we should have $w_{1c}>0$ or $w_{2c}>0$. This new assignment would be R-dominated by $x$. So, if an assignment, say $z$, R-dominates $x$, it should satisfy $z_{1b}<\frac{1}{2}$ or $z_{2b}<\frac{1}{2}$. In the first case, $z_{2b}>\frac{1}{2}$ should hold, and $z$ would be R-dominated by $x$. In the second case, $z_{1b}>\frac{1}{2}$ should hold, and $z$ would be R-dominated by $x$. As a result, no assignment R-dominates $x$. \hfill$\qed$
\end{example}

To better understand the intuition behind the Rawlsian criterion, we will find a Rawlsian assignment for some particular preference profiles. First, if all agents' top choices are different, a Rawlsian assignment assigns each agent their first choice with probability one. Alternatively, if all agents' preferences are the same, each agent is assigned each object with probability $\frac{1}{n}$. Lastly, an agent is assigned her least preferred object if, and only if, all agents have the same least preferred object.

Given an assignment $x$, we defined the vector $B^x$ by considering, in the first place, the probabilities associated with the least preferred object of each agent, that is, the object of rank $n$. Next, we consider the two least preferred objects, the objects of rank $n-1$ or higher. And so on, and so forth. The order $(n,n-1,\ldots,2)$ comes from our definition of the worst-off agents.\footnote{Because the probability with which each agent is assigned to her $n$ least preferred objects is always one, we do not include $1$ in this order.} However, other orders are also possible. For example, we might consider the other ``extreme'' of the Rawlsian assignment: first consider the probabilities with which each agent is assigned to her $(n-1)$ least preferred objects (rank $2$ or higher), then her $(n-2)$ least preferred objects (rank $3$ or higher), etc. The associated order would be $(2,3,\ldots,n)$.

In Appendix~\ref{family}, we consider a family of assignments obtained by changing the order in which the cumulative-rank levels are considered. For example, the Rawlsian assignment corresponds to the order $(n,n-1,\ldots, 2)$. We show that many of the results in Section \ref{results} continue to hold, including uniqueness, sd-efficiency, and anonymity. 
Additionally, a modified version of the algorithm in Section~\ref{algorithm} can be used to compute the outcomes selected by each order in this family. We also show in the Online Appendix, Section \ref{frac_Boston}, that the other ``extreme'' of the Rawlsian assignment in this family, the one associated with the order $(2,3,\ldots,n)$, differs from the fractional Boston assignment introduced by \cite{chen2023fractional}, despite the similarities between the two assignments.




\section{Results}\label{results}
In this section, we study the theoretical properties of Rawlsian assignments. We show that every profile admits a unique Rawlsian assignment, and that this assignment is sd-efficient and anonymous. Furthermore, we explore fundamental incompatibilities between the egalitarian objective of minimizing the maximum rank and properties such as envy-freeness and strategyproofness. In response, we identify several properties that ensure robustness against certain types of manipulations, such as consistency, non-bossiness, and group lower invariance.

\subsection{Uniqueness, efficiency \& anonymity}

Every problem admits a Rawlsian assignment.
In principle, there might exist multiple Rawlsian assignments but, in fact, this is never the case. Suppose otherwise, and let $x$ and $y$ be two Rawlsian assignments. Then, the assignment $\frac{1}{2}x+\frac{1}{2}y$ R-dominates $x$ and $y$. Thus, as the following proposition states, every problem admits a unique Rawlsian assignment. It is worth noting that randomization is key for the uniqueness. If we restrict the analysis to deterministic assignments, there are problems with multiple Rawlsian assignments.

\begin{proposition}\label{prop1}
Each problem has a unique Rawlsian assignment.
\end{proposition}

Proposition \ref{prop1} allows us to define the Rawlsian rule as the rule that assigns to each problem its unique Rawlsian assignment. 

A minimum requirement of fairness is anonymity: agents' allocations cannot depend on their names, that is, if agents’ names are permuted, their allocations should be permuted in the same way. As we state in the next proposition, the Rawlsian rule satisfies this property.

\begin{proposition}\label{prop2}
The Rawlsian rule satisfies anonymity.
\end{proposition}

We now analyze the efficiency of the Rawlsian rule. In contrast to many economic environments where there is a tension between egalitarianism and efficiency, these concepts are compatible in our framework. 
Given two assignments, $x$ and $y$, we define the following order to compare the associated vectors $B^x$ and $B^y$.

\begin{definition}
    Let $x$ and $y$ be two assignments. We say that $B^x \triangleleft B^y$ if $b_i^x(k) \leq b_i^y(k)$ for every $i,k=1,\ldots,n$, and $b_i^x(k) < b_i^y(k)$ for some $i,k=1,\ldots,n$. 
\end{definition}

That is, every entry of the vector $B^x$ is lower than or equal to the corresponding entry of the vector $B^y$. 

The following observation shows that stochastic dominance can be rewritten in terms of the lower-tail cumulative probabilities $b_i^x(k)$, which follows immediately from the definitions.

\begin{observation}\label{obs1}
Let $x$ and $y$ be two assignments. Then $x$ stochastically dominates $y$
if and only if $B^x \triangleleft B^y$.
\end{observation}

\begin{proof}
For each agent $i$ and each rank cutoff $k\geq 2$,
\[
\sum_{o\in O:\, r_{io}\leq k-1} x_{io}
=
1-b_i^x(k).
\]
Thus, $x_i$ sd-dominates $y_i$ if and only if
\[
b_i^x(k)\leq b_i^y(k)
\qquad
\text{for every } k=2,\ldots,n.
\]
If $x\neq y$, at least one of these inequalities must be strict. Hence
$x$ stochastically dominates $y$ if and only if $B^x \triangleleft B^y$.
\end{proof}

Consider an assignment $y$ that is not sd-efficient. Then there exists another assignment $x$ that sd-dominates $y$. By Observation \ref{obs1}, $B^x\triangleleft B^y$, and therefore $x$ R-dominates $y$. Hence $y$ cannot be the Rawlsian assignment.
This proves the following proposition.

\begin{proposition}\label{sd-eff}
The Rawlsian rule is sd-efficient.
\end{proposition}

The Birkhoff-von Neumann theorem states that any assignment can be represented as a lottery over the set of deterministic assignments. The Rawlsian assignment is sd-efficient; thus, \textbf{every} such representation will only contain ex-post efficient deterministic assignments \citep{aziz2015ex}. Then, we have the following corollary.

\begin{corollary}
Every representation of the Rawlsian assignment as a lottery over the set of deterministic assignments uses only ex-post efficient assignments with minimal maximum rank.
\end{corollary}

There is another concept of efficiency, called ``rank efficiency'', which was proposed by \cite{featherstone2020rank}. In the Online Appendix (Section \ref{rank}), we give its formal definition and show that there is no relation between Rawlsian and rank efficient assignments.

In the Online Appendix, Section \ref{cardinal}, we discuss how the Rawlsian assignment relates to cardinal representations of agents' ordinal preferences. We show that the sd-efficiency of the Rawlsian rule implies that there exist cardinal utilities, consistent with the agents' ordinal preferences, for which the Rawlsian rule maximizes the expected utility of the worst-off agent. However, the converse does not hold. An assignment that maximizes the expected utility of the worst-off agent for some cardinal representation need not be Rawlsian, because the cardinal maximin objective depends on the particular utility gaps between the different ranks, whereas the Rawlsian assignment depends only on ordinal information. 
Moreover, we show that no rule satisfying equal treatment of equals maximizes the expected utility of the worst-off agent \textit{for all} cardinal utilities that are consistent with the agents' ordinal preferences.

\subsection{Envy-freeness}
\label{subsec:envy_free}
We show that envy-freeness is incompatible with the egalitarian objective of minimizing the maximum rank, i.e., the worst rank to which any agent is assigned with positive probability. 

\begin{proposition}
\label{prop:sd_envyImposs}
    For $n\geq 3$, no rule satisfying sd-envy-freeness minimizes the maximum rank.
\end{proposition}
\begin{proof}
    Consider the following profile $\succeq$:
    \[
\begin{array}{c|ccc}
\succeq_1 & a & b & c \\ 
\succeq_2 & b & a & c \\ 
\succeq_3 & a & c & b \\ 
\end{array}.
\]
Any rule $\phi$ that minimizes the maximum rank will satisfy $\phi_{3c}(\succeq) = 1$. This implies that $\phi_{3a}(\succeq) = 0$. As a result, either $\phi_{1a}(\succeq) > 0$ or $\phi_{2a}(\succeq) > 0$ must hold. Assume, without loss of generality, that $\phi_{1a}(\succeq) > 0$. Then, the allocation $\phi_3(\succeq)$ of agent 3 does not stochastically dominate the allocation $\phi_1(\succeq)$ of agent 1. Hence, $\phi$ violates sd-envy-freeness.
\end{proof}

It follows immediately that the Rawlsian rule violates sd-envy-freeness. Note that the Rawlsian rule also violates weak sd-envy-freeness for the profile in the proof of Proposition \ref{prop:sd_envyImposs}.\footnote{In fact, the profile $\succeq$ in Proposition \ref{prop:sd_envyImposs} can also be used to prove that, for $n\geq 3$, no rule satisfying \textit{weak} sd-envy-freeness\textit{ and sd-pairwise-efficiency} minimizes the maximum rank. An assignment $x$ is sd-pairwise-efficient if no pair of agents $i,j\in I$ exists such that $x_i \succeq^{sd}_j x_j, x_j \succeq^{sd}_i x_i$, and $x_i\neq x_j$. At profile $\succeq$ in the proof of Proposition~\ref{prop:sd_envyImposs}, sd-pairwise-efficiency would require that $\phi_{1a}(\succeq) =1$ and $\phi_{2b}(\succeq) =1$. Hence, $\phi_1(\succeq) \succeq_3^{sd} \phi_3(\succeq)$, violating weak sd-envy-freeness.}

Several incompatibility results between egalitarianism, envy-freeness and efficiency exist in the literature. \citet[][Theorem 1]{duddy2025egalitarian}, for example, showed that there is no rule for $n\geq4$ satisfying sd-efficiency, sd-envy-freeness and even-handedness. In the context of fair division, \cite{brams2017maximin}, for example, propose necessary and sufficient conditions for the existence of an envy-free allocation of indivisible goods among two players that minimizes the worst rank of the received goods.

\subsection{Strategyproofness and max-rank minimization}
Because the Rawlsian rule is sd-efficient and satisfies equal treatment of equals (as a consequence of being anonymous), it follows immediately from \cite{bogomolnaia2001new}'s impossibility result that the Rawlsian rule violates sd-strategyproofness when $n\geq 4$. 

In fact, we show that even weak notions of sd-strategyproofness are fundamentally incompatible with the egalitarian objective of minimizing the maximum rank. To this end, consider the following relaxation of sd-strategyproofness which was first introduced by \cite{hashimoto2014two} as weak invariance, and used by \cite{mennle2021partial} in their characterization of sd-strategyproofness.

A rule is \textbf{upper invariant} if swapping the order of two consecutively ranked objects in the preferences does not affect the assignment probabilities of more preferred objects. 
\begin{definition}
    \label{def:upper_invar}
     A rule $\phi$ is \textbf{upper invariant} if, for every agent $i\in I$, all preference profiles $(\succeq_i,\succeq_{-i})$, and all misreports $\succeq_i'$ that are obtained by swapping two consecutively ranked objects $a$ and $b$ at $\succeq_i$, i.e., $a\succeq_i b$ but $b\succeq_i'a$, it holds that $\phi_{ij}(\succeq_i,\succeq_{-i})=\phi_{ij}(\succeq'_i,\succeq_{-i})$ for all objects $j$ for which $j\succeq_i a$.
\end{definition}

The following impossibility result shows that no rule that minimizes the maximum rank can satisfy upper invariance. Proposition \ref{prop:SPimposs} strengthens the impossibility result by \citet[][Proposition 5]{afacan2024rawlsian}. They showed that in a more general model, where the number of agents and objects can differ, and which allows for an outside option, there is no \textit{deterministic} rule satisfying \textit{strategyproofness} that minimizes the maximum rank when $n\geq 4$. 

\begin{proposition}
\label{prop:SPimposs}
    For $n\geq 3$, no rule satisfying upper invariance minimizes the maximum rank.
\end{proposition}
\begin{proof}
    Consider the profile $\succeq$.
    \[
\begin{array}{c|ccc}
\succeq_1 & a & b & c  \\ 
\succeq_2 & a & b & c  \\ 
\succeq_3 & a & c & b \\ 
\end{array}
\]
Any rule $\phi$ that minimizes the maximum rank, will enforce that $\phi_{3c}(\succeq) = 1$. Hence, $\phi_{3a}(\succeq)=0$. Now, consider the alternative profile $(\succeq'_3, \succeq_{-3})$ in which agent 3 swaps the positions of objects $b$ and~$c$.
    \[
\begin{array}{c|ccc}
\succeq_1 & a & b & c  \\ 
\succeq_2 & a & b & c  \\ 
\succeq'_3 & a & b & c \\ 
\end{array}
\]
By upper invariance, it must still hold that $\phi_{3a}(\succeq'_3, \succeq_{-3}) = 0$. Hence, either $\phi_{1a}(\succeq'_3, \succeq_{-3}) > 0$ or $\phi_{2a}(\succeq'_3, \succeq_{-3}) > 0$. Assume, without loss of generality, that $\phi_{1a}(\succeq'_3, \succeq_{-3}) > 0$. Consider the alternative profile $(\succeq''_1, \succeq_2, \succeq'_3)$ in which agent 1 swaps the positions of objects $b$ and $c$.
\[
\begin{array}{c|ccc}
\succeq''_1 & a & c & b  \\ 
\succeq_2 & a & b & c  \\ 
\succeq'_3 & a & b & c \\ 
\end{array}
\]
By upper invariance, it must hold that $\phi_{1a}(\succeq''_1, \succeq_2, \succeq'_3) > 0$. This contradicts that $\phi$ minimizes the maximum rank, as every rule that minimizes the maximum rank has $\phi_{1c}(\succeq''_1, \succeq_2, \succeq'_3) = 1$.

In general, for $n$ agents, we can construct a similar example where $n-1$ agents submit the same preferences, and one agent ranks the first $n-2$ objects identically to the other agents, while swapping the order of the last two objects.
\end{proof}


It follows immediately that the Rawlsian rule violates upper invariance.


Additionally, we evaluated the Rawlsian rule on other relaxations of sd-strategyproofness that are not logically implied by upper invariance. In the Online Appendix, Section \ref{weak_sdSP}, we provide an example in which the Rawlsian rule violates weak sd-strategyproofness. 
Lastly, in Appendix~\ref{SP_axioms_Mennle_Seuken} we also show that the Rawlsian rule satisfies the lower invariance axiom in the decomposition of sd-strategyproofness by \cite{mennle2021partial}, but violates the axiom of swap monotonicity.

\subsection{Robustness properties of the Rawlsian rule}\label{robust}

Although all rules that minimize the maximum rank are manipulable, the Rawlsian rule satisfies several robustness and invariance properties that distinguish it from other rules with the same egalitarian objective. We first show that it satisfies group lower invariance: even a group of agents cannot change \textit{any} agent's assignment probabilities for lower-ranked objects by reshuffling only the order of their higher-ranked objects. We then show that the Rawlsian rule also satisfies consistency and non-bossiness.

\subsubsection{Group lower invariance}

Decision-makers interested in minimizing the maximum rank naturally care about the assignment probabilities of lower-ranked objects. Consider the following property: whenever a coalition of agents reshuffles only their top-$k$ objects, the assignment probabilities of objects ranked strictly below \(k\) remain unchanged for all agents.

To formalize what it means for a coalition to reshuffle their top-$k$ objects, we introduce the following notion. 

\begin{definition}
    Given a $k\in \{1,\dots,n\}$, two profiles $\succeq$ and $\succeq'$ are \textbf{top-$\mathbf{k}$-shuffled} if:
    \begin{itemize}
        \item for every agent $i\in I$, and for all objects $o, o' \in O$ for which $r_{io}(\succeq_i) \leq k$ and $r_{io'}(\succeq_i) > k$, it holds that $o \succeq_i o'$ if and only if $o \succeq_i'o'$,
        \item for all agents $i\in I$, and for all objects $o, o' \in O$ for which $r_{io}(\succeq_i) > k$ and $r_{io'}(\succeq_i) > k$, it holds that $o \succeq_i o'$ if and only if $o \succeq_i'o'$.    
    \end{itemize}
\end{definition}

Note that two top-$k$-shuffled profiles may differ in the ranking of the objects placed in positions $1, \ldots, k$.

We now define \textbf{group lower invariance}. It requires that, if agents only reshuffle objects within their top-$k$ positions, then the assignment probabilities of objects ranked worse than position $k$ remain unchanged for all agents.

\begin{definition}
A rule $\phi$ is \textbf{group lower invariant} if, for every $k\in\{1,\ldots,n\}$ and every pair of top-$k$-shuffled profiles $\succeq$ and $\succeq'$, it holds that for every agent $i\in I$ and every object $o\in O$ with $r_{io}(\succeq_i)>k$,
\[
\phi_{io}(\succeq)=\phi_{io}(\succeq').
\]
\end{definition}

Group lower invariance is substantially stronger than lower invariance by \cite{mennle2021partial} (Definition \ref{def:lower_invar} in Appendix \ref{SP_axioms_Mennle_Seuken}), as it allows for joint manipulations by groups of agents, does not restrict attention to adjacent swaps, and requires the assignment probabilities of objects ranked below the manipulated positions to remain unchanged \textit{for all agents}.

We have the following result.

\begin{proposition}
\label{prop:group_lower_invar}
    The Rawlsian rule satisfies group lower invariance.
\end{proposition}

The proof is included in Section \ref{app:group_lower_invar}. 
This result shows that, under the Rawlsian rule, an agent can change the assignment probabilities of lower-ranked objects only by misreporting the relative order of those lower-ranked objects. In Appendix \ref{app:group_lower_invar_other_rules}, we show that PS and RSD, as well as the positive and prudent equality rules by \cite{duddy2025egalitarian} all violate group lower invariance.

\subsubsection{Consistency and non-bossiness}


We next establish two robustness properties of the Rawlsian rule: consistency and non-bossiness. These properties are important because they capture forms of robustness that remain meaningful even when full strategyproofness is impossible. Consistency ensures that the rule behaves stably when an agent leaves with her assignment probabilities: the assignments of the remaining agents are not affected except through the corresponding reduction in object capacities. We refer the reader to Appendix \ref{app:cons_non_boss} for a formal definition of consistency. Non-bossiness rules out a different type of strategic manipulation, namely situations in which an agent can change the allocations of others without changing her own allocation \citep{satterthwaite1981strategy, bade2016fairness}. In other words, a rule satisfies non-bossiness if, whenever an agent changes her preference report without affecting her own allocation, the allocations of all other agents also remain unchanged.



\begin{definition}
\label{def:non-bossy}
   A rule $\phi$ satisfies \textbf{non-bossiness} if, for each profile $\succeq$, for each agent $i\in I$, and for each $\succeq_i'$, it holds that
    $$\phi_i(\succeq) = \phi_i(\succeq'_i, \succeq_{-i}) \Rightarrow \phi(\succeq) = \phi(\succeq'_i, \succeq_{-i}).$$
\end{definition}

Together, these properties show that the Rawlsian rule satisfies natural robustness requirements that complement its egalitarian objective.

\begin{proposition}
\label{prop:cons_non_boss}
    The Rawlsian rule satisfies consistency and non-bossiness.
\end{proposition}

The proof of Proposition \ref{prop:cons_non_boss} is included in Appendix \ref{app:cons_non_boss}.

PS satisfies consistency and therefore also non-bossiness \citep{han2016consistency, heo2025agent}. RSD violates consistency for $n\geq 4$ \citep[][Proposition 2]{han2016consistency}, but satisfies non-bossiness \citep{bade2016fairness}. The positive and prudent equality rules by \cite{duddy2025egalitarian} violate both consistency and non-bossiness (see Appendix \ref{app:cons_non_boss}).

\subsubsection{Non-bossiness and Rawlsian leximin reasoning}
We conclude this section by explaining why non-bossiness is a key notion when applying Rawlsian leximin reasoning to our setting. The main idea is that any rule that can be interpreted as selecting cardinal leximin maximizers satisfies a welfare version of non-bossiness: if an agent changes her report but receives the same allocation, then the expected utilities of all other agents must remain unchanged.

Let $(u_i)_{i\in I}$ be a profile of cardinal utilities, where
$u_i=(u_{io})_{o\in O}$ and $u_{io}$ denotes agent $i$'s utility from object $o$. Given an assignment $x\in X$, the expected utility of agent $i$ is
\[
U_i(x)=\sum_{o\in O}x_{io}u_{io}.
\]
Let
\[
U^\uparrow(x)=\big(U_{(1)}(x),\ldots,U_{(n)}(x)\big)
\]
be the vector of expected utilities ordered from lowest to highest. 
We say that $x^*\in X$ is a \textit{leximin maximizer} if, for every $x\in X$,
either
\[
U^\uparrow(x^*)=U^\uparrow(x),
\]
or there exists \(k\in\{1,\ldots,n\}\) such that
\[
U_{(k)}(x^*)>U_{(k)}(x),
\]
and
\[
U_{(\ell)}(x^*)=U_{(\ell)}(x)
\quad \text{for all } \ell<k.
\]

When a rule selects one of the leximin maximizers for every problem, we say that the rule \textit{lexicographically maximizes the expected utilities of the worst-off agents}. Although there may be several leximin-maximizing assignments, all of them induce the same leximin expected-utility vector, and, moreover, \textit{the same expected utility for each agent}.\footnote{This follows from convexity of the feasible assignment set and linearity of expected utilities: if two leximin maximizers induced different expected utility vectors with the same ordered vector, their average would strictly improve the leximin vector.}

Consider a rule $\phi$ and a profile $\succeq$ such that the following two conditions hold:
\begin{enumerate}[label = (\roman*)]
    \item there exists an agent $i$ and preferences $\succeq'_i$ such that $\phi_i(\succeq) = \phi_i(\succeq'_i, \succeq_{-i})$, and
    \item there exists an agent $j$ who is strictly better off after agent $i$'s manipulation at $\succeq$, i.e., $ \phi_j(\succeq'_i, \succeq_{-i}) \succ^{sd}_j \phi_j(\succeq)$.
\end{enumerate}

The two conditions imply that the rule violates non-bossiness. Indeed, standard non-bossiness (Definition \ref{def:non-bossy}) only requires that some other agent's allocation changes, whereas condition (ii) requires the stronger statement that some other agent is strictly better off in the stochastic dominance sense. At both profiles $\succeq$ and $(\succeq_i',\succeq_{-i})$, agent $i$ receives the same allocation. Hence, the same residual shares of the objects remain to be divided among the same set of non-manipulating agents, whose preferences are unchanged. Therefore, if $\phi$ were to select cardinal leximin maximizers at both profiles, the non-manipulating agents should receive the same expected utilities before and after the change in agent $i$'s report. This contradicts condition (ii), because a strict stochastic dominance improvement for agent $j$ implies that agent~$j$'s expected utility is strictly higher under every strict cardinal utility representation consistent with $\succeq_j$. 
In other words, if a rule satisfies the two conditions stated above, then it cannot be interpreted, at these two profiles, as lexicographically maximizing the expected utilities of the worst-off agents for any strict cardinal utilities consistent with the reported ordinal preferences.

The examples in Appendix~\ref{app:cons_non_boss} in which the positive and prudent equality rules of \citet{duddy2025egalitarian} violate non-bossiness all satisfy condition~(ii). Indeed, the manipulating agent's own allocation is unchanged, while another agent's allocation improves in the stochastic dominance sense. Thus, these rules fail a robustness requirement that is naturally associated with the Rawlsian idea of lexicographically maximizing the expected utilities of the worst-off agents.

\section{Computing the Rawlsian assignment}\label{algorithm}
We now define an algorithm to compute the Rawlsian assignment in polynomial time. The proposed algorithm is an extension of the algorithm by \cite{airiau2023portioning}, which computes a leximin distribution of a divisible object, to our setting with multiple objects. The pseudo-code is shown in Algorithm~\ref{alg:Rawls} and an implementation of the code is available online (\url{https://github.com/DemeulemeesterT/Rawlsian-assignments}).

Intuitively, the algorithm constructs the vector $B^x$ in the same order in which the Rawlsian criterion compares assignments. It begins with the probabilities with which agents receive their least preferred object, then moves to the probabilities with which agents receive one of their two least preferred objects, and so on.

Consider first the initial step, where $k=n$ in the description of Algorithm \ref{alg:Rawls}. For each agent~$i$, the term $b_i^x(n)$ is the probability with which agent $i$ receives her least preferred object. The algorithm first solves a linear program, which we call LP 1, that finds the smallest number $b^*$ such that there exists a feasible assignment $x$ with
\[
b_i^x(n)\leq b^*
\qquad \text{for every } i\in I.
\]


Once $b^*$ has been found, the algorithm identifies which agents must receive their least preferred object with probability exactly $b^*$. To do this, it solves a second type of linear program, which we call LP 2. Fix an agent $j\in I$. LP 2 asks whether it is possible to reduce agent $j$'s probability of receiving her least preferred object strictly below $b^*$ while keeping every agent's probability of receiving her least preferred object weakly below $b^*$. Formally, LP 2 maximizes $\varepsilon$ subject to feasibility and to the constraints
\[
b_{j}^x(n)\leq b^*-\varepsilon
\]
and
\[
b_i^x(n)\leq b^*
\qquad \text{for every } i\in I.
\]
If the optimal value is $\varepsilon^*=0$, then the probability of agent $j$ cannot be moved strictly below $b^*$ without increasing some other agent's probability above $b^*$. Hence, the algorithm fixes
\[
b_{j}^x(n)=b^*
\]
and records that this coordinate of $B^x$ has been determined by adding agent $j$ to $I_n$, which is the subset of agents for which the probability of being assigned their $n$-th choice has been fixed by the algorithm.

After fixing all agents who must attain $b^*$, the algorithm repeats the same two-LP procedure for the remaining agents, now keeping the previously fixed values unchanged. In this way, it determines the largest value of $b_i^x(n)$, then the second-largest value, and so on, until the probabilities $b_i^x(n)$ have been fixed for all agents. If at some point LP 1 gives $b^*=0$, then all remaining agents can be assigned zero probability of receiving their least preferred object.

The algorithm then proceeds to $k=n-1$. At this stage, it applies the same procedure to the probabilities $b_i^x(n-1)$, that is, the probabilities with which agents receive one of their two least preferred objects, while keeping all previously fixed values $b_i^x(n)$ unchanged. The algorithm continues recursively for $k=n-2,\ldots,1$. When the procedure terminates, the assignment returned by the last linear program is the Rawlsian assignment.
 
 It can be checked that Algorithm~\ref{alg:Rawls} needs to solve at most $ n^2\cdot\frac{n+1}{2}$ linear programs to find the Rawlsian assignment. Note that by changing the order in which all agent-preference pairs are considered, a similar algorithmic procedure can be used to compute the probabilistic assignments of alternative rules that depend on lexicographically optimizing a vector of cumulative assignment probabilities. This is the case, for example, of the generalized Rawlsian rules (Appendix \ref{family}), the positive and prudent equality rules by \cite{duddy2025egalitarian} as described in Appendix \ref{conal}, or PS as characterized by \cite{bogomolnaia2015random}.

\begin{algorithm}[tb!]
 \caption{Computing the Rawlsian assignment}
    \label{alg:Rawls}
    $I_k \gets \varnothing$ for $k \in\{1,\ldots,n\}$.\\
    $b^x_i(k)$ is fixed once $i\in I$ is added to $I_k$.
    \begin{algorithmic}
    \For{$k \in \{n, \ldots, 1\}$} 
    \While{$I_k \neq I$}
    \State \parbox[t]{\dimexpr\textwidth-3\leftmargin-\labelsep-\labelwidth}{%
    \textbf{LP 1:} Find the minimum value of $b^*$ such that there exists an assignment $x\in[0,1]^{n^2}$ satisfying\strut}\vspace{-0.4cm}
    \State \begin{align*}
        \sum_{o\in O}x_{io} &= 1 &\forall \; i \in I\\[-3pt]
        \sum_{i \in I} x_{io} &= 1 &\forall \; o \in O\\[-3pt]
        \sum_{o \in O} \mathds{1} \{r_{io} \geq k \} x_{io} &\leq b^* &\forall \; i \in I\setminus I_k\\[-3pt]
        \sum_{o \in O} \mathds{1} \{r_{io} \geq t \} x_{io} &= b^x_i(t) &\forall \; i \in I_t: t\geq k
    \end{align*}
    \If{$b^* = 0$} { set $b_i^x(k)=0$ for every $i\in I\setminus I_k$} 
    
    {\ \ \ \ \ \ \ \ \ \ \ \ \ \ \ \  \ \ \ \ \ \ \ \ $I_k \gets I$}
    \Else
    \For{$j \in I\setminus I_k$}
    \State \parbox[t]{\dimexpr\textwidth-7\leftmargin-\labelsep-\labelwidth}{%
    \textbf{LP 2:} Find the maximum value of $\epsilon^*$ such that there exists an assignment $x\in[0,1]^{n^2}$ satisfying\strut}\vspace{-0.4cm}
    \State \begin{align*}
        \sum_{o\in O}x_{io} &= 1 &\forall \; i \in I\\[-3pt]
        \sum_{i \in I} x_{io} &= 1 &\forall \; o \in O\\[-3pt]
        \sum_{o \in O} \mathds{1} \{r_{jo} \geq k \} x_{jo} &\leq b^* - \epsilon^*&\\[-3pt]
        \sum_{o \in O} \mathds{1} \{r_{io} \geq k \} x_{io} &\leq b^* &\forall \; i \in I\setminus I_k\\[-3pt]
        \sum_{o \in O} \mathds{1} \{r_{io} \geq t \} x_{io} &= b^x_i(t) &\forall \; i \in I_t: t\geq k
    \end{align*}
    \If{$\epsilon^* = 0$} {add $j$ to $I_k$, and set $b^x_{j}(k) = b^*$.}
    \EndIf
    \EndFor
    \EndIf
    \EndWhile
    \EndFor\\
    \Return{The solution $x\in[0,1]^{n^2}$ from the last solved LP.}
    \end{algorithmic}
\end{algorithm}

\section{Empirical Application}\label{app}
Our main motivation is the assignment of apartments in housing cooperatives. A housing cooperative is formed by a group of families who join to construct a building. 
Once it is finished, the apartments are to be distributed. Prices are not used, so the situation fits the assignment problem described above.

Before the rule currently in use, cooperatives assigned apartments randomly. Preferences were not considered and the assignment was sampled randomly from the set of all assignments. The assignment was in general not efficient, so it was decided to change the rule. A group of researchers from the Engineering School of the University of Uruguay, proposed a new rule, called the \textbf{MTAV}, which was finally adopted. More information about the current rule can be found in \cite{prino2016optimal} and \cite{Paleo2021}. 

One of the main concerns of the families that participate in the cooperatives is the \textit{distribution} of the final assignment. They do not want \textit{inegalitarian} distributions in terms of the ranking of the assigned apartment in each family's preferences. They want to avoid a situation where, for example, a family gets their most preferred unit, while others get apartments ranked very low in their preferences. The MTAV explicitly addresses this concern (we define it in Appendix \ref{MTAV}). 

In this section, we use the data of families' preferences from $24$ cooperatives to compare the outcomes of the Rawlsian, PS and MTAV rules. In practice, the assignment is organized according to the number of rooms in the apartments. Therefore, all the apartments we consider as part of a cooperative have the same number of rooms. The sizes of the cooperatives range from $4$ to $42$ families (with an average size of 17). As a proxy for the correlation of preferences, we consider the cumulative number of different apartments ranked in the first, second, third, and fourth position in the preferences. The result shows that preferences are not highly correlated. We should note that, since our data consist of ordinal rankings, our efficiency analysis is necessarily ordinal: the Rawlsian assignment is sd-efficient, but we do not observe cardinal utilities and therefore cannot measure cardinal welfare losses.
 Section \ref{charac} of the Online Appendix presents detailed information for each cooperative.

It is worth noting that there are only two cooperatives, $C_4$ and $C_{10}$, where the Rawlsian and the PS rules coincide. These are the smallest cooperatives (each consisting of 4 families). In addition to these 8 families, there is only one family that receives the same allocation under the two rules. Thus, overall, 9 families out of 408 receive the same lottery over apartments.

We should mention that the MTAV is not strategyproof \citep{Paleo2021}. Nonetheless, we take preferences submitted by the families at face value. We are not aware of manipulations by the families, and in general given the information held by the families, it is very difficult to 
manipulate the rule profitably.

\subsection{Comparison with PS rule: maximum rank}\label{max_Rank}
The Rawlsian rule is sd-efficient and is designed to improve the welfare of the worst-off families. A first way to measure this improvement is what we call the maximum rank. Given an assignment, we consider for each agent the rank of the least preferred object received with positive probability. Then, we take the maximum rank across all families. 
In Table \ref{maxrank} we look at the maximum rank of the Rawlsian and PS rules. In contrast to what might be expected (based on the correlation of preferences), in all but two cooperatives, the PS rule assigns at least one family their least preferred apartment with positive probability. Under the Rawlsian rule, this happens only in three cooperatives (those with a small number of families). For the rest of the cooperatives, 
the average maximum rank of the Rawlsian rule as a percentage of the length of families' preferences is 48\%.

\begin{table}[htbp]
     \begin{center}
     \caption{Size and maximum rank of each cooperative for the Rawlsian and PS assignment.}   \label{maxrank}
     \begin{tabular}{lcccc}
     \hline
     Coop. & $n$ & Max.\ Rawls & Max.\ Rawls (\%) & Max.\ PS \\ \hline
     $C_1$ & 26 & 13 & 50 & 26 \\
     $C_2$ & 18 & 12 & 67 & 18 \\
     $C_3$ & 4 & 2 & 50 & 4 \\
     $C_4$ & 4 & 3 & 75 & 3 \\
     $C_5$ & 28 & 8 & 29 & 28 \\
     $C_6$ & 8 & 3 & 38 & 8 \\
     $C_7$ & 29 & 8 & 28 & 29 \\
     $C_8$ & 12 & 7 & 58 & 12 \\
     $C_9$ & 15 & 6 & 40 & 14 \\
     $C_{10}$ & 4 & 4 & 100 & 4 \\
     $C_{11}$ & 11 & 5 & 45 & 11 \\
     $C_{12}$ & 16 & 6 & 38 & 16 \\
     $C_{13}$ & 39 & 14 & 36 & 39 \\
     $C_{14}$ & 42 & 33 & 79 & 42 \\
     $C_{15}$ & 14 & 9 & 64 & 14 \\
     $C_{16}$ & 6 & 6 & 100 & 6 \\
     $C_{17}$ & 9 & 3 & 33 & 9 \\
     $C_{18}$ & 15 & 8 & 53 & 15 \\
     $C_{19}$ & 9 & 9 & 100 & 9 \\
     $C_{20}$ & 20 & 10 & 50 & 20 \\
     $C_{21}$ & 24 & 7 & 29 & 24 \\
     $C_{22}$ & 7 & 2 & 29 & 7 \\
     $C_{23}$ & 40 & 11 & 28 & 40 \\
     $C_{24}$ & 8 & 7 & 88 & 8 \\
     \hline
     \end{tabular}
     \end{center}
   \footnotesize{  \textit{Notes}:} 
\footnotesize{Coop.\ stands for cooperative, each denoted as $C_i$ for $i=1,\ldots,24$. Size is the number of families in each cooperative. For Max.\ Rawls (Max.\ PS) we compute the rank of the least preferred object assigned with positive probability by the Rawlsian (PS) rule for each family, and then we take the maximum among all families. Max.\ Rawls (\%) expresses Max.\ Rawls as a percentage of the length of families' preferences (or, equivalently, the size of the cooperative).}
\end{table}

We complement this comparison with two additional exercises. First, we conduct an i.i.d.\ bootstrap analysis within each cooperative: for each cooperative with $n$ families, we generate artificial profiles by drawing $n$ preference lists independently with replacement from the submitted rankings. For each bootstrap profile, we compute the maximum rank under the Rawlsian rule and under PS. The results are shown in Table \ref{tab:bootstrap_max_rank} in the Online Appendix, Section \ref{OA:bootstrap}. The same qualitative pattern emerges: the Rawlsian rule continues to deliver a substantially lower maximum rank than PS.


Second, we consider a simulation exercise in which we vary the correlation of preferences. For each cooperative, we construct a central ranking of apartments using the average observed rank of each apartment. We then generate artificial preference profiles around this central ranking, where a parameter $\theta$ controls how concentrated preferences are around it. When $\theta=0$, preferences are uniformly random; as $\theta$ increases, agents' preferences become increasingly aligned with the central ranking and therefore more correlated. Table~\ref{tab:correlation_simulation_normalized} reports the normalized maximum rank under the Rawlsian rule and PS for different values of $\theta$. The Rawlsian rule improves the worst assignment relative to PS for all values of $\theta$, although the size of the improvement decreases as preferences become more correlated. Details are provided in Section \ref{OA:bootstrap} of the Online Appendix, including the normalized maximum rank under different levels of preference correlation for each of the cooperatives.

\begin{table}[htbp]
\begin{center}
\caption{Normalized maximum rank under different levels of preference correlation}
\label{tab:correlation_simulation_normalized}
\begin{threeparttable}
\begin{tabular}{lccc}
\hline
 $\theta$ &  Rawlsian$/n$ &  PS$/n$ &  $(PS-R)/n$ \\
\hline
     0.00 &          0.29 &    0.96 &        0.67 \\
     0.05 &          0.31 &    0.96 &        0.66 \\
     0.10 &          0.37 &    0.97 &        0.60 \\
     0.20 &          0.51 &    0.97 &        0.47 \\
     0.40 &          0.67 &    0.98 &        0.31 \\
     0.80 &          0.81 &    0.99 &        0.18 \\
     1.60 &          0.90 &    1.00 &        0.09 \\
     3.20 &          0.97 &    1.00 &        0.03 \\
\hline
\end{tabular}
\end{threeparttable}
\end{center}
\footnotesize{  \textit{Notes}:} 
\footnotesize{The table reports averages across simulated preference profiles and cooperatives. For each cooperative, we construct a central ranking using average observed ranks. For each value of $\theta$, we generate artificial preference profiles using a distribution centered at this ranking. Larger values of $\theta$ generate preferences that are more concentrated around the central ranking. Rawlsian$/n$ and PS$/n$ report the average maximum rank under each rule divided by the number of agents in the cooperative. The last column reports the normalized difference.}
\end{table}

The comparison of maximum ranks shows that the Rawlsian assignment assigns fewer families their least preferred apartments. Now we look at the intensive margin, that is, the probabilities with which families are assigned their least preferred objects. It could be that even when the PS rule assigns families apartments that are ranked very low, this occurs with very small probability. To investigate this, fix a cooperative and an assignment $x$. The expected number of families assigned apartments ranked in position $k \in \{1,\ldots,n\}$ is
\[
\sum_{i\in I}\sum_{o\in O:\, r_{io}=k} x_{io}.
\]
This expectation is taken with respect to the lottery represented by the random assignment $x$, conditional on the observed  preference profile of the cooperative. Tables~\ref{dist1}--\ref{dist6} in the Online Appendix present the results for all the cooperatives considering the Rawlsian and PS rules. Not only is the maximum rank higher under PS than under the Rawlsian assignment for each cooperative, but also the cumulative probability of being assigned their least preferred apartments is substantially higher. 
For example, for cooperative $C_{13}$ the Rawlsian rule assigns all families apartments ranked 14th or better (out of 39 apartments), while the PS rule assigns (in expectation) 8 families apartments ranked 15th or worse. 

The general picture regarding the expected number of families that are assigned apartments with rank $k$ is as follows. The Rawlsian rule assigns a lower number of families their least preferred apartments compared to the PS rule. But, at the same time, it also assigns a lower number of families their top choices, and especially their first choice. As an illustration, Figure \ref{fig1} reports the cumulative rank distributions for two relatively large cooperatives, C7 and C23. These examples are not used as evidence on their own; rather, they illustrate the pattern documented in the Online Appendix, Section \ref{dist_agents}. Relative to the PS rule, the Rawlsian assignment places less mass on the very top ranks but substantially reduces the mass assigned to low-ranked apartments. The corresponding figures for all cooperatives are reported in the Online Appendix, Section \ref{cum_figs}. 

\begin{figure}[htp]
\centering
\begin{subfigure}{.5\textwidth}
  \centering
  \includegraphics[width=1.1\linewidth]{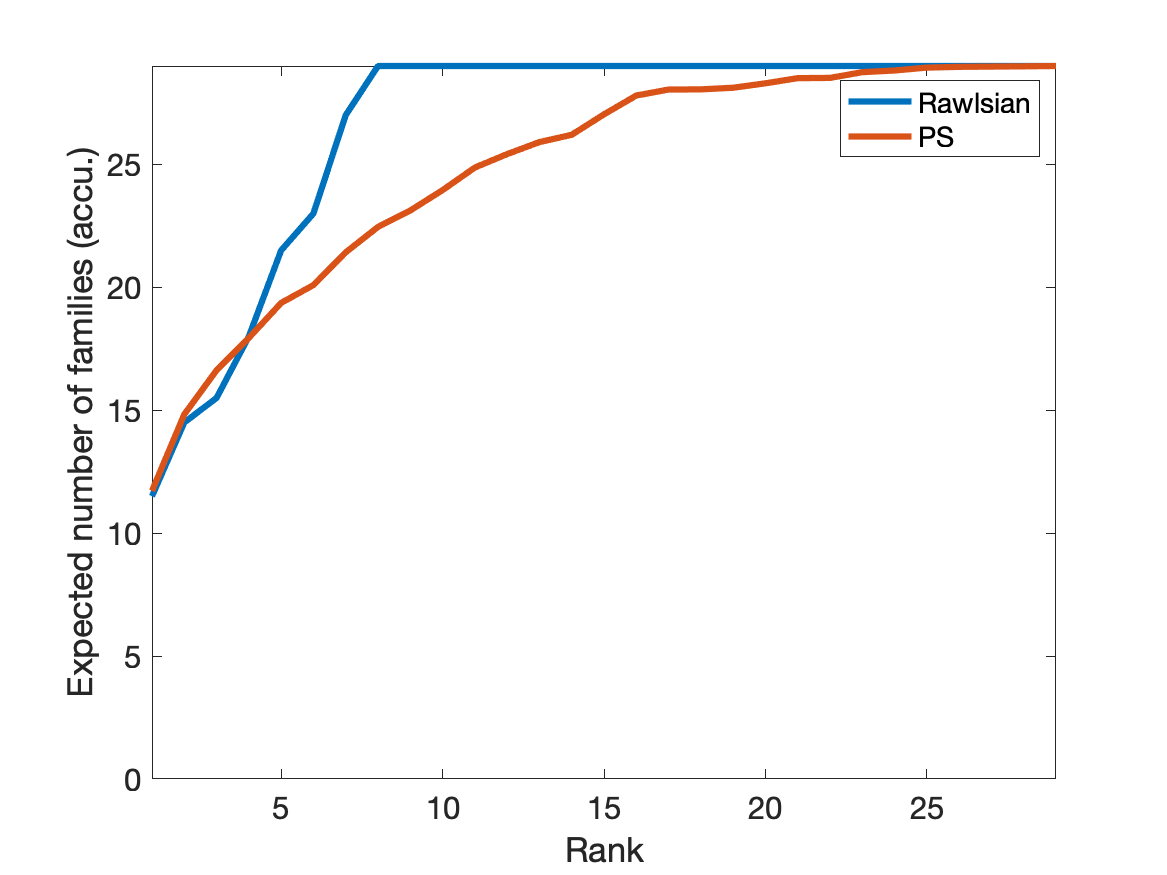}
  \caption{Cooperative $C_7$}
  \label{fig:dist}
\end{subfigure}%
\begin{subfigure}{.5\textwidth}
  \centering
  \includegraphics[width=1.1\linewidth]{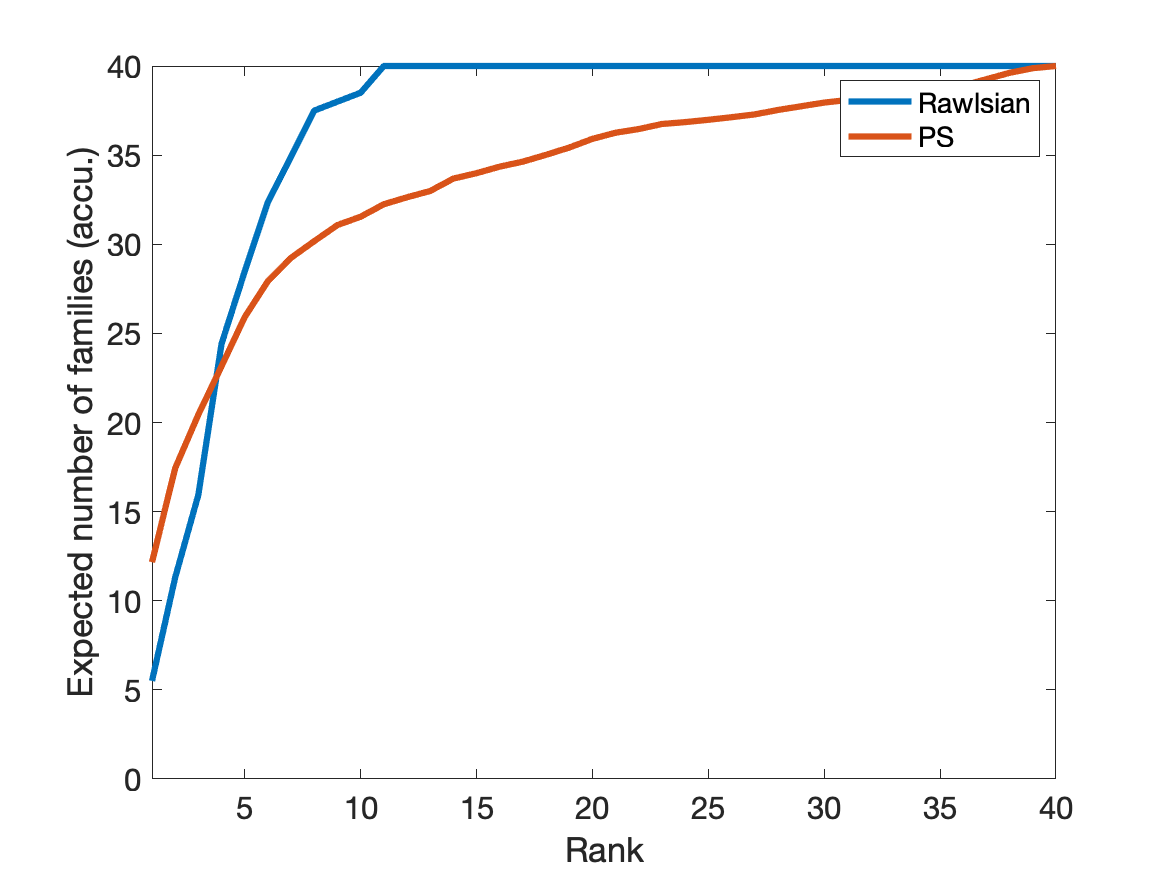}
  \caption{Cooperative $C_{23}$}
  \label{fig:sub2}
\end{subfigure}


\caption{Cumulative distribution function of the expected number of families that are assigned apartments with rank $k$  by the Rawlsian and PS rules.\label{fig1}}
\end{figure}

\subsection{Comparison with the PS rule: individual preferences over assignments} \label{subsec:result_envy}
Both the Rawlsian and the PS rules are sd-efficient. Therefore, it is never the case that all families prefer one assignment to the other. 
In this section we compare, for each cooperative, the number of families that prefer the assignment of one rule over the other. We say that a family prefers the Rawlsian (PS) assignment to the PS (Rawlsian) assignment if the lottery the family receives in the first assignment sd-dominates the lottery in the second assignment. Obviously, there are families for whom neither assignment sd-dominates the other. Table \ref{comp} presents the results.     
For every cooperative, more families prefer their Rawlsian assignment to their PS assignment. Moreover, the average percentage of families that prefer their Rawlsian (PS) assignment over the PS (Rawlsian) assignment is 35\% (9\%).\footnote{One may wonder if the Rawlsian allocation is always more popular than the PS allocation (that is, if the number of agents who prefer the former is always higher than the number of agents who prefer the latter). The example in Remark \ref{Rem_rank_eff} of the Online Appendix, Section \ref{rank}, shows that this is not the case. Indeed, in this problem, two agents prefer the PS assignment, one agent  prefers the Rawlsian assignment, one agent obtains the same allocation, and one agent does not prefer either assignment based on stochastic dominance.}

In addition, we also report in Table \ref{comp} the magnitude of these comparisons using expected ranks. For each agent $i$, we define the expected rank of agent $i$'s assignment under rule $m\in\{R,PS\}$ as

\[
ER_i^m = \sum_{o\in O} r_{io}x_{io}^m
\]where, as before, $r_{io}$ is the rank of object $o$ in $i$'s preference list.

The average gain from the Rawlsian rule, reported in Table \ref{comp}, is computed among agents whose Rawlsian assignment strictly sd-dominates their PS assignment:
\[
\text{Gain R}
=
\frac{1}{|\{i:x_i^R\succ_i^{sd}x_i^{PS}\}|}
\sum_{i:x_i^R\succ_i^{sd}x_i^{PS}}
\left(
ER_i^{PS}-ER_i^R
\right).
\]
Analogously, the average gain from PS is computed among agents whose PS assignment strictly sd-dominates their Rawlsian assignment:
\[
\text{Gain PS}
=
\frac{1}{|\{i:x_i^{PS}\succ_i^{sd}x_i^R\}|}
\sum_{i:x_i^{PS}\succ_i^{sd}x_i^R}
\left(
ER_i^R-ER_i^{PS}
\right).
\]
Both measures are expressed in rank positions, and therefore provide an interpretable measure of the size of the gains and losses. The gain columns indicate that, conditional on a strict sd-improvement, the average gains are comparable in magnitude across rules. Pooling across cooperatives, families who are strictly better off under the Rawlsian rule gain on average about $3.28$ rank positions in expected-rank terms, while families who are strictly better off under PS gain about $3.68$ rank positions. The main difference is therefore the incidence of improvements: the Rawlsian rule strictly sd-improves the assignment of many more families than PS does.

\begin{table}[htbp]
\centering
\setlength{\tabcolsep}{3.5pt}
\renewcommand{\arraystretch}{1.05}
\caption{Stochastic dominance comparison between the Rawlsian rule and PS}
\label{comp}
\begin{tabular}{lrrrrr}
\hline
Coop. & $n$ & Prefer $R$ & Gain $R$ & Prefer $PS$ & Gain $PS$ \\
\hline
$C_{1}$  & 26 & 8  & 4.73  & 1 & 2.57  \\ 
$C_{2}$  & 18 & 6  & 4.95  & 1 & 6.29  \\
$C_{3}$  & 4  & 3  & 0.44  & 1 & 0.67  \\
$C_{4}$  & 4  & 0  & --    & 0 & --    \\
$C_{5}$  & 28 & 10 & 4.14  & 1 & 6.48  \\
$C_{6}$  & 8  & 4  & 1.22  & 1 & 0.83  \\
$C_{7}$  & 29 & 12 & 4.67  & 2 & 3.18  \\
$C_{8}$  & 12 & 6  & 1.96  & 1 & 4.04  \\
$C_{9}$  & 15 & 4  & 1.80  & 1 & 1.88  \\
$C_{10}$ & 4  & 0  & --    & 0 & --    \\
$C_{11}$ & 11 & 5  & 1.10  & 2 & 0.67  \\
$C_{12}$ & 16 & 8  & 1.74  & 1 & 0.50  \\
$C_{13}$ & 39 & 9  & 7.09  & 1 & 3.46  \\
$C_{14}$ & 42 & 7  & 10.63 & 1 & 21.96 \\
$C_{15}$ & 14 & 3  & 3.54  & 1 & 5.00  \\
$C_{16}$ & 6  & 3  & 0.66  & 2 & 0.53  \\
$C_{17}$ & 9  & 4  & 0.77  & 1 & 0.89  \\
$C_{18}$ & 15 & 5  & 2.25  & 1 & 3.76  \\
$C_{19}$ & 9  & 5  & 1.72  & 1 & 3.49  \\
$C_{20}$ & 20 & 2  & 5.82  & 1 & 6.77  \\
$C_{21}$ & 24 & 9  & 3.63  & 1 & 1.19  \\
$C_{22}$ & 7  & 4  & 0.92  & 1 & 0.50  \\
$C_{23}$ & 40 & 6  & 6.83  & 2 & 2.59  \\
$C_{24}$ & 8  & 2  & 1.48  & 0 & --    \\
\hline
\end{tabular}

\vspace{0.4em}
\begin{minipage}{0.92\textwidth}
\footnotesize
\textit{Notes}: The table compares the Rawlsian rule, denoted by $R$, and probabilistic serial, denoted by PS. The columns Prefer $R$ and Prefer $PS$ report the number of agents for whom the corresponding strict stochastic dominance relation holds. Gain $R$ reports the average reduction in expected rank, among agents whose Rawlsian allocation strictly sd-dominates their PS allocation. Gain $PS$ is defined analogously for agents whose PS allocation strictly sd-dominates their Rawlsian allocation.
\end{minipage}
\end{table}

\subsection{Comparison with the PS rule: sd-envy-freeness}
In this section we turn to the analysis of sd-envy-freeness. The PS assignment is sd-envy-free, so no family envies another family under its outcome. As observed before, the Rawlsian rule is not sd-envy-free, so some families may experience envy. 

We show in Table \ref{envy}, for each cooperative, the number of families (and the percentage over all the families) whose allocation does not sd-dominate the allocation of some other family. Also, among those families that experience envy, we show the average number of families that are sd-envied. There is only one cooperative where no family experiences envy. Among the rest, the percentage of families with sd-envy ranges from 22\% to 94\%, with an average of 58\%. 
If we consider weak sd-envy, that is, a family has weak sd-envy over another family if the allocation of this last family sd-dominates the allocation of the first one, the percentage of families with weak sd-envy ranges from 11\% to 67\%, with an average of 36\%.

 \begin{table}[htbp]
     \begin{center}
     \caption{Sd-envy in the Rawlsian assignment.}\label{envy}
     \begin{tabular}{lrrrr}
     \hline
     Coop. & Size & Envy Fam.\ & Envy Fam.\ (\%) & Avg. Envied Fam.\ \\ \hline
     $C_1$ & 26 & 20 & 77 & 6 \\
     $C_2$ & 18 & 17 & 94 & 6 \\
     $C_3$ & 4 &  1 & 25 & 2 \\
     $C_4$ & 4 & 0 & 0 & 0 \\
     $C_5$ & 28 & 18 & 64 & 5 \\
     $C_6$ & 8 & 4 & 50 & 1 \\
     $C_7$ & 29 & 18 & 62 & 5 \\
     $C_8$ & 12 & 6 & 50 & 3 \\
     $C_9$ & 15 & 10 & 67 & 4 \\
     $C_{10}$ & 4 & 0 & 0 & 0 \\
     $C_{11}$ & 11 & 5 & 45 & 2 \\
     $C_{12}$ & 16 & 10 & 63 & 3 \\
     $C_{13}$ & 39 & 33 & 85 & 7 \\
     $C_{14}$ & 42 & 40 & 95 & 12 \\
     $C_{15}$ & 14 & 11 & 79 & 5 \\
     $C_{16}$ & 6 & 2 & 33 & 3 \\
     $C_{17}$ & 9 & 2 & 22 & 3 \\
     $C_{18}$ & 15 & 10 & 67 & 5 \\
     $C_{19}$ & 9 & 4 & 44 & 6 \\
     $C_{20}$ & 20 & 18 & 90 & 7 \\
     $C_{21}$ & 24 & 15 & 63 & 5 \\
     $C_{22}$ & 7 & 3 & 43 & 1 \\
     $C_{23}$ & 40 & 36 & 90 & 6 \\
     $C_{24}$ & 8 & 7 & 88 & 3 \\
     \hline
     \end{tabular}
     \end{center}
    \footnotesize{ \textit{Notes}:} 
\footnotesize{Envy Fam.\ is the number of families with sd-envy, Envy Fam.\ (\%) is the percentage of families with sd-envy, and Avg.\ Envied Fam.\ is the average of the number of sd-envied families. }
     \end{table}

\subsection{Comparison with MTAV}\label{MTAV:comp}
Tables~\ref{dist1}--\ref{dist6} in Section~\ref{dist_agents} of the Online Appendix report the expected number of families assigned apartments with rank $k$ under each of the three rules. By construction, the maximum ranks of the Rawlsian rule and MTAV always coincide.
 However, it is interesting to note that there are cases where the Rawlsian rule assigns fewer families their least preferred apartment (among those that are received with positive probability). For example, consider $C_1$: under the Rawlsian rule only one family receives their apartment of maximum rank (which is 13, out of 26 apartments), while under MTAV two families receive their 13th choice. The same is true for cooperatives $C_2$, $C_5$, $C_9$, $C_{13}$, $C_{15}$, $C_{17}$, $C_{21}$, and $C_{23}$. In general, the outcome of the MTAV is located between the two other rules. Indeed, the Rawlsian rule outperforms the MTAV under the Rawlsian criterion, but not in terms of the expected number of families assigned their top choice. The PS rule outperforms the MTAV rule in terms of the expected number of families assigned their top choice, but not under the Rawlsian criterion.

\section{Analysis for large markets}\label{section_large}

In the empirical application, we show that the average maximum rank of the Rawlsian assignment (excluding those cooperatives with less than 10 families) is around 46\% of the market size (Section \ref{max_Rank}). In this section, we build on this finding by analyzing the maximum rank of the Rawlsian rule in large markets. Specifically, we consider markets of size $n$, where agents' preferences are drawn i.i.d.\ from a uniform distribution, and study the limit as $n$ tends to infinity. Although the average maximum rank of the Rawlsian rule tends to infinity, we show that it grows at a slow rate. In particular, the expected maximum rank for the Rawlsian rule is upper bounded by $\lfloor\ln(n)\rfloor$ plus a constant. For instance, in a market of size \(n = 1000\), our result implies that the expected rank of the least preferred assigned object is at most~9.
 
Let $\mathcal{P}_n$ denote the set of all possible preference profiles in a market of size $n$. We define the random variable $\succeq\in\mathcal{P}_n$ that selects one of the preference profiles in $\mathcal{P}_n$ uniformly at random. Then, we denote the maximum rank of the assignment found by rule $\phi$ at $\succeq$ as:
$$
r^\phi_{max}(\succeq)=\max_{(i,o) \in I\times O} \{r_{io}(\succeq)|\phi_{io}(\succeq)>0\}.
$$

The expected maximum rank of the rule $\phi$ is then defined as:


\begin{equation}
    \mathbb{E}(r^{\phi}_{max})=\sum_{\succeq \in \mathcal{P}_n} \frac{1}{\vert \mathcal{P}_n \vert} r^\phi_{max}(\succeq).
\end{equation}

First, we show that the expected maximum rank of the Rawlsian assignment goes to infinity when the market grows large.
\begin{proposition}\label{prop:LB_Rawls}
    $\lim_{n\to\infty}\mathbb{E}(r^{Rawls}_{max}) = +\infty$.
\end{proposition}

Next, we provide an upper bound on the growth rate of the Rawlsian rule's maximum rank.

    

    

\begin{proposition}\label{large}
     Consider the Rawlsian rule. Then, when $n\rightarrow \infty$:
     \begin{equation}
         \mathbb{E}_n(r^{Rawls}_{max}) \leq \lfloor\ln(n)\rfloor + \sum_{k=-1}^{+\infty}\left(1-e^{-2e^{-k}}\right) \approx \lfloor\ln(n)\rfloor + 2.77026.
     \end{equation}
\end{proposition}

The proof of Proposition~\ref{prop:LB_Rawls} can be found in Section~\ref{proofs_large2} of the Online Appendix, and the proof of Proposition~\ref{large} can be found in Appendix~\ref{large_proof}.
 We are not aware of any theoretical result on the maximum rank of the PS rule. \cite{ortega2023cost} showed that the rank-minimizing rule, which minimizes the expected rank of agents, has an expected maximum rank of $\log_2(n)$ in large random markets. Interestingly, our result implies that the expected maximum rank of the Rawlsian rule in random markets is asymptotically bounded above by a fraction $\ln(2) \approx 0.69$ of the expected maximum rank of the rank-minimizing rule.

To show the tightness of the bound in Proposition~\ref{large}, we have randomly sampled $1,000$ preference profiles for each market of size $n\in\{3,\ldots,59\}$. The left panel of Figure~\ref{fig:random} shows that the empirical analysis in Section~\ref{max_Rank} regarding the difference in the maximum rank between the Rawlsian and PS rules seems to hold in general. The probability with which the PS rule assigns an agent to their last choice with positive probability is close to one. As the size of the market grows and preferences are uniformly i.i.d., the difference between the maximum rank of the Rawlsian and the PS rules tends to infinity on average. Therefore, the egalitarian advantage of the Rawlsian rule over PS also holds in large markets.

The right panel of Figure~\ref{fig:random} shows that the upper bound of Proposition~\ref{large} is reasonably close to the observed maximum rank of the Rawlsian rule. Interestingly, the true maximum rank of the Rawlsian rule seems to be approximately equal to $\ln(n)+1$.

\begin{figure}[htp]
\centering
\includegraphics[width=1\linewidth]{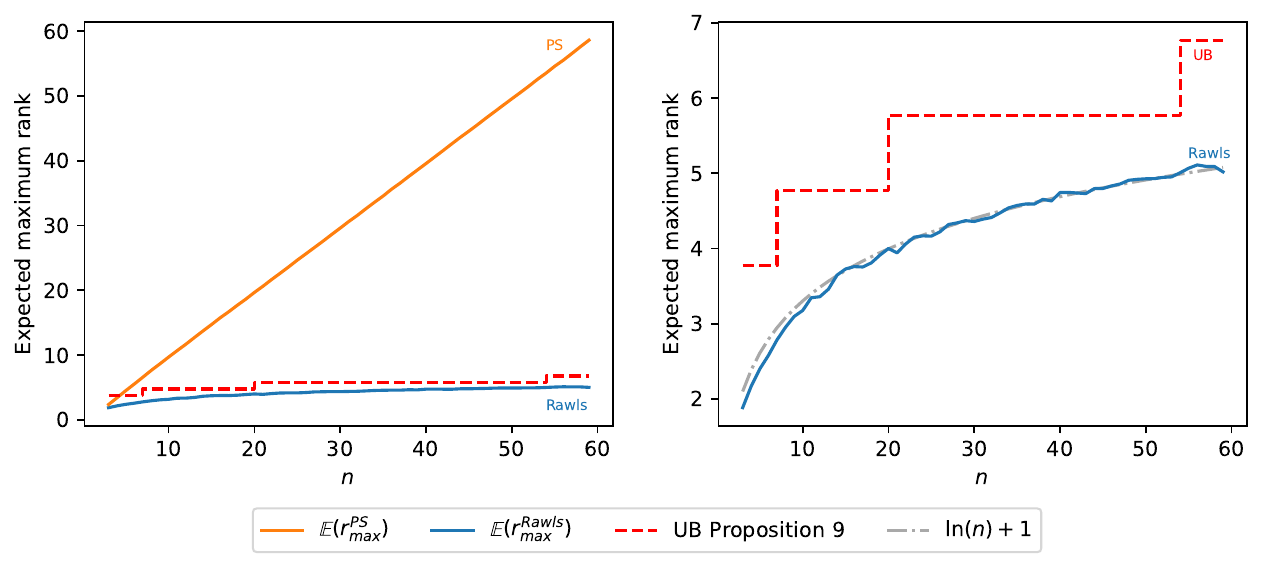}
\caption{Simulation results for the maximum ranks of the Rawlsian and PS rules.\label{fig:random}}
\end{figure}


\section{Concluding Remarks}\label{conclu}

We examined the allocation of indivisible goods to individuals when prices cannot be utilized. Our investigation draws inspiration from the context of housing cooperatives, where families express concerns about the fairness of the final assignment. Specifically, we aim to avoid assignments in which some families receive their top choices while others are assigned apartments ranked very low in their preferences. To address this, we introduce a concept called Rawlsian assignments, which prioritizes improving the allocations of individuals who are worst-off. We demonstrate that there always exists a unique Rawlsian assignment. Moreover, the Rawlsian rule is both sd-efficient and anonymous. Whereas the Rawlsian rule violates sd-strategyproofness and sd-envy-freeness, we show that this is unavoidable as either property is violated by any rule that minimizes the maximum rank of the objects to which the agents are assigned with positive probability. Furthermore, we compare our proposed rule with the PS rule and the currently employed rule. Our findings reveal that the Rawlsian rule significantly outperforms the other two rules in terms of egalitarian outcomes. 

\bibliographystyle{econometrica}
\bibliography{rawls}
\clearpage

\section{Appendix}

\subsection{Proofs}\label{app:proofs}

\subsubsection{Proof of Proposition \ref{prop1}.}

\begin{proof}
\textbf{Existence.}
The set of assignments $X$ is a compact set of $[0,1]^{n^2}$. For each $x\in X$ we construct the vector $B^x=(B^x_1, \ldots, B^x_n, B^x_{n+1}, \ldots, B^x_{2n},\ldots,B^x_{n^2}).$ Consider the function $\pi_1:X\rightarrow \mathbb{R}$ such that $\pi_1(x)=B^x_1$. Clearly, $\pi_1$ is a continuous function (it is a projection). Then, the problem $\min_{x \in X} \pi_1(x)$ has a solution. Let $S_1\subset X$ be the set of all solutions to the minimization problem.

Note that $S_1$ is compact. It is bounded because $X$ is bounded. It is also closed: take a convergent sequence $\{x^k\}_k \subset S_1$. All the elements $x^k$ have the same $B^{x^k}_1$, so the limit of the sequence must have the same $B^{x^k}_1$ as well. Then, the limit is an element of $S_1$.

Consider the function $\pi_2:S_1\rightarrow \mathbb{R}$ such that $\pi_2(x)=B^x_2$. Clearly, $\pi_2$ is a continuous function. Then, the problem $\min_{x \in S_1} \pi_2(x)$ has a solution. Let $S_2\subset X$ be the set of all the solutions to the minimization problem.

We continue in the same way for all the elements of the vector $B^x$. At the end, we will have a nonempty set $S_{n^2}$.\footnote{Note that there will be a unique vector $B^x$ that minimizes the problems. But, in principle, we could have many assignments associated with the same vector $B^x$.} The assignments in this set are Rawlsian. Indeed, suppose this is not the case, and consider $x\in S_{n^2}$ which is Rawlsian-dominated by another assignment $y$. This means that there exists an index  $j\in \{1,\ldots,n^2\}$ such that $B^{x}_j>B^y_j$, and for all $i<j$, $B^{x}_i=B^{y}_i$. Then, $y \in S_{j-1}$, and $\pi_{j}(y)<\pi_{j}(x)$. This implies that $x\notin S_j \Rightarrow x \notin S_{n^2}$, which is a contradiction. 

\textbf{Uniqueness.}
Suppose there are two Rawlsian assignments, $x$ and $y$. Both are associated with the same vector $B$. Given an assignment $z$, consider a matrix $P^z$ where agents are represented in the rows, and in column $k$ we include the probability with which each agent receives the object ranked in position $k$.\footnote{Matrix $P^z$ is created by reordering each row of $z$ based on the agents' preferences.} Starting from the last column, consider the first column where $P^x$ and $P^y$ differ (there is such a column as $x$ and $y$ are different assignments). Assume that this is the case for column $n-c$.

Note that the probabilities of column $n-c$ in each matrix are the same, but distributed differently. Until column $n-c$, the two matrices are the same, so the same agents get the same probabilities for the corresponding objects. This implies that all agents receive the same probability for objects ranked $n-(c-1),\ldots,n$ under $x$ and $y$. Therefore, they also receive the same probabilities under assignment $\frac{1}{2}(x+y)$. 

Consider the largest element of column $n-c$ of each of the matrices $P^x$, $P^y$, and $P^{\frac{1}{2}(x+y)}$. If the largest element of this column in  $P^x$ and $P^y$ corresponds to the same agent, then the probability with which this agent receives the object ranked in position $n-c$ coincides in $x$, $y$, and $\frac{1}{2}(x+y)$. If it corresponds to a different agent, then either the largest element of column $n-c$ in $P^x$ or the largest element of column $n-c$ in $P^y$, is larger than the largest element of $ P^{\frac{1}{2}(x+y)}$. In the first case, $\frac{1}{2}(x+y)$ R-dominates $x$, and in the second case $\frac{1}{2}(x+y)$ R-dominates $y$. But this contradicts the fact that $x$ and $y$ are Rawlsian assignments.
As the assignments $x$ and $y$ differ in at least one entry of column $n-c$, the assignment $\frac{1}{2}(x+y)$ R-dominates $x$ or $y$.  
\end{proof}

\subsubsection{Proof of Proposition \ref{prop2}.}\label{proof:prop2}

\begin{proof}
Consider two problems $(\succeq_{i})$ and $(\succeq_{\pi(i)})$, and let $(x_i)_i$ and $(x'_{i})_i$ be the Rawlsian assignments in each of these problems. 
We need to show that $x'_{i}=x_{\pi^{-1}(i)}$ for every $i\in \{1,\ldots,n\}$. 
Suppose this is not the case, and consider the assignments $(x_{\pi^{-1}(i)})_i$ and $(x'_i)_i$, and the first entry where vectors $B^{(x_{\pi^{-1}(i)})_i}$ and $B^{(x'_i)_i}$ differ. Because $(x'_i)_i$ is the Rawlsian assignment of $(\succeq_{\pi(i)})_i$, this entry in $B^{(x'_i)_i}$ is smaller than in the vector $B^{(x_{\pi(i)})_i}$. Note that the definition of the vector $B^x$ of each assignment $x$, does not look at the identities of the agents, thus: $B^{(x_i)_i}=B^{(x_{\pi(i)})_i}$.
But then $(x'_i)_i$ R-dominates $(x_{i})_i$ in problem $(\succeq_{\pi(i)})_i$, which is a contradiction.
\end{proof}

\subsubsection{Proof of Proposition \ref{prop:group_lower_invar}}
\label{app:group_lower_invar}
\begin{proof}
    Given a value $k\in\{1,\dots,n\}$, let $\succeq$ and $\succeq'$ denote two top-$k$-shuffled profiles. Denote by $x$ and $x'$ the Rawlsian assignments for profiles $\succeq$ and $\succeq'$, respectively. Assume, for contradiction, that $x$ and $x'$ do not give the same probabilities to some agent for an object she ranks in position $k+1, \ldots, n$, i.e., there exists an agent-object pair $(i,o) \in I \times O$ for which $r_{io}(\succeq_i) > k$ such that $x_{io} \neq x_{io}'$.

    Note that, because profiles $\succeq$ and $\succeq'$ are top-$k$-shuffled, the elements in vectors $B^x$ and $B^{x'}$ that correspond to the preferences $k+1, \ldots, n$ refer to the same agent-object pairs for both assignments $x$ and $x'$. Hence, if $x$ Rawlsian-dominates another assignment $y$ at profile $\succeq$ because $B^x$ is lexicographically smaller than $B^y$ in one of the elements of $B^y$ that correspond to preferences $k+1, \ldots, n$, then we know that $x$ would also Rawlsian-dominate $y$ at profile $\succeq'$. A similar argument holds when replacing $x$ by $x'$, and by replacing $\succeq$ by $\succeq'$ and $\succeq'$ by $\succeq$ in the previous sentence.
    

    We can consider two cases. First, despite $x$ and $x'$ having different assignment probabilities for some agent for an object she ranks in position $k+1, \ldots, n$, it could be that the elements of $B^x$ and $B^{x'}$ that correspond to preferences $k+1, \ldots, n$ are identical. In that case, assignment $z=\frac{1}{2}x+\frac{1}{2}x'$ Rawlsian-dominates both $x$ and $x'$, because $B^z$ is lexicographically smaller than $B^x=B^{x'}$ in one of the elements that correspond to preferences $k+1, \ldots, n$, contradicting the assumption that $x$ and $x'$ are the Rawlsian assignments at $\succeq$ and $\succeq'$, respectively.

    Otherwise, at least one of the elements of $B^x$ and $B^{x'}$ that correspond to preferences $k+1, \ldots, n$ differ. Note that, for each agent $i\in I$, and for all $\ell\in \{k+1, \dots, n\}$, the elements $b_i^x(\ell)$ and $b_i^{x'}(\ell)$ correspond to the cumulative assignment probabilities of being assigned to the same subset of objects. Hence, either $x$ Rawlsian-dominates $x'$ on both profiles $\succeq$ and $\succeq'$, or vice versa. Assume, without loss of generality, that $x$ Rawlsian-dominates $x'$. This contradicts the fact that $x'$ is the Rawlsian assignment at $\succeq'$.
\end{proof}

\subsection{A generalization of the Rawlsian assignment}\label{family}
Let $\sigma$ be an ordering of the set of integers $\{2,\ldots,n\}$, and let $\Sigma$ denote the set of all such orderings. We denote the $i$-th element of $\sigma$ by $\sigma_i$.\footnote{Note that we do not include the first preference in the orderings in $\Sigma$, because the total assignment probability of each agent always equals one.}
Similar to Section~\ref{rawls}, given an assignment $x$, an ordering $\sigma\in \Sigma$, and the vectors $(b_i^x)_{i\in I}$ containing the cumulative assignment probabilities, we define the vector $B^{x,\sigma}\in[0,1]^{n(n-1)}$ as follows.

\begin{enumerate}
    \item The first elements $(B^{x,\sigma}_1, \ldots, B^{x,\sigma}_n)$ are the elements $(b^x_1(\sigma_1), \ldots, b^x_n(\sigma_1))$ listed in a non-increasing order. 
    
    \item Elements  $(B^{x,\sigma}_{n+1}, \ldots, B^{x,\sigma}_{2n})$ are the elements $(b^x_1(\sigma_2), \ldots, b^x_n(\sigma_2))$ listed in a non-increasing order. 
    
    \item In general, elements $(B^{x,\sigma}_{(k-1)n+1}, \ldots, B^{x,\sigma}_{kn})$ for $k=1,\ldots, n-1$, are the elements $(b^x_1(\sigma_k), \ldots, b^x_n(\sigma_k))$ listed in a non-increasing order. 
\end{enumerate}

That is, the first elements of $B^{x,\sigma}$ are the probabilities with which each agent receives the objects with a rank between $\sigma_1$ and $n$, then the probabilities with which each agent receives the objects with a rank between $\sigma_2$ and $n$ and so on, and so forth. 
The Rawlsian assignment corresponds to $B^{x,\sigma_R}$, with $\sigma_R = (n, n-1, \ldots,2)$. Like in Section~\ref{rawls}, given two assignments $x$ and $y$, we compare the vectors $B^{x,\sigma}$ and $B^{y,\sigma}$ lexicographically.

\begin{definition}
Given two assignments $x$ and $y$ and an ordering $\sigma\in\Sigma$, $x$ \textbf{$\sigma$-dominates} $y$ if there is $j\in \{1,\ldots,n(n-1)\}$ such that $B^{x,\sigma}_j<B^{y,\sigma}_j$, and for all $i<j$, $B^{x,\sigma}_i=B^{y,\sigma}_i$.
\end{definition}

\begin{definition}
Given an ordering $\sigma\in\Sigma$, an assignment $x$ is $\bm{\sigma}$\textbf{-minimal} if it is not $\sigma$-dominated by any other assignment. \end{definition}

Similar to the results in Section~\ref{results}, a $\sigma$-minimal assignment is unique and sd-efficient, for any ordering $\sigma\in\Sigma$.

\begin{proposition}
    Each problem has a unique $\sigma$-minimal assignment, for any order $\sigma\in\Sigma$.
\end{proposition}
\begin{proof}
    Given an ordering $\sigma\in\Sigma$, suppose $x$ and $y$ are two different assignments that are both $\sigma$-minimal, and have the same vector $B^\sigma$. Following a similar reasoning as the proof of Proposition~\ref{prop1}, we can show that the assignment $\frac{1}{2}(x+y)$ $\sigma$-dominates both $x$ and $y$, contradicting the fact that $x$ and $y$ are both $\sigma$-minimal assignments.
\end{proof}

\begin{proposition}
    For any ordering $\sigma\in\Sigma$, the $\sigma$-minimal assignment is sd-efficient.
\end{proposition}
\begin{proof}
    Given any ordering $\sigma\in\Sigma$, suppose that assignment $x$ is 
    $\sigma$-minimal, but not sd-efficient. Then, there is an improving cycle $(1,o_1, 2,o_2, \ldots, K, o_K)$. We can assume wlog that $o_i\neq o_j$ for every $i,j$. Denote by $\epsilon>0$ the minimum of all the probabilities: $\epsilon=\min_{k=1,\ldots,K} x_{k o_k}$. Implement the cycle by decreasing each probability $x_{k o_k}$, for $k=1,\ldots,K$, by $\epsilon$, and increasing $x_{1 o_K}$, and $x_{k o_{k-1}}$, $k=2,\dots,K$ by the same share. We get a new random assignment $y$. We will show that it $\sigma$-dominates $x$, which is a contradiction.

    Because each agent is better off after we implement the improvement cycle, we know that $r_{k o_{k-1}} < r_{k o_k}$, for each $k=1,\ldots,K$ (we use the notation $o_{0}=o_K$). This implies that, for each agent $k = 1,\ldots,K$, the cumulative probability of being assigned to an object ranked $r_{k o_{k-1}}$-th or worse is equal in $x$ and $y$, i.e.,
    \begin{equation*}
        b^y_k(r_{k o_{k-1}}) = b^x_k(r_{k o_{k-1}}) + \epsilon - \epsilon = b^x_k(r_{k o_{k-1}}).
    \end{equation*}
    Similarly, the cumulative probability of being assigned to an object ranked $r_{k o_{k}}$-th or worse is lower in $y$ than in $x$, i.e.,
    \begin{equation*}
        b^y_k(r_{k o_{k}}) = b^x_k(r_{k o_{k}}) - \epsilon \Longleftrightarrow b^y_k(r_{k o_{k}}) < b^x_k(r_{k o_{k}}).
    \end{equation*}
    As a result, each element of the vector $B^{y,\sigma}$ associated with assignment $y$ is not larger than the corresponding element in the vector $B^{x,\sigma}$ associated to assignment $x$. Moreover, there are at least $K$ elements of $B^{y,\sigma}$ that are strictly smaller than the corresponding elements in $B^{x,\sigma}$. Hence, $y$ $\sigma$-dominates $x$.
\end{proof}

We define the $\sigma$-minimal rule as the rule that assigns the $\sigma$-minimal assignment to each problem. As for the Rawlsian rule, the $\sigma$-minimal rule satisfies anonymity.

\begin{proposition}
    The $\sigma$-minimal rule satisfies anonymity, for any ordering $\sigma\in\Sigma$.
\end{proposition}
\begin{proof}
The definition of vector $B^{x,\sigma}$ does not depend on agents' identities, only on the probabilities of the assignment $x$ and the order $\sigma$. Thus, the same proof as in Proposition \ref{prop2} (Appendix \ref{proof:prop2}) applies in this generalized setting. 
\end{proof}

\begin{remark}
    The family defined in this section changes only the order in which the cumulative-rank levels are considered. It therefore does not include the positive and prudent equality rules by \cite{duddy2025egalitarian}. Those rules can be obtained only by allowing a more general ordering of the individual sorted entries.
\end{remark}

\subsection{Relation with even-handed assignments by \texorpdfstring{\cite{duddy2025egalitarian}}{Duddy (2025}}\label{conal}

In this section, we define the egalitarian criterion of \textbf{even-handedness} that was proposed by \cite{duddy2025egalitarian}, and show that it is independent from the concept of Rawlsian assignments.

Let $t^x_i(k)$ be the probability with which agent $i$ receives her top $k$ objects under allocation $x_i$. That is:

$$
t^x_i (k)= \sum_{o \in O} \mathds{1} \{r_{io} \leq k \} x_{io}.
$$

We denote by $t_i^x$ the vector of agent $i$'s cumulative probabilities from the most to the least preferred object: $t_i^x = (t_i^x(1), t_i^x(2), \ldots, t_i^x(n) = 1)$. Note that there is a close connection between $t_i^x$ and $b_i^x$, the vector with agent $i$'s cumulative probabilities from the least to the most preferred object. Indeed, it holds that $t^x_i (k) = 1 - b^x_i (k + 1)$ for any agent $i$, for any assignment $x$, and for any $k = 1,\ldots, n - 1$. As such, all the definitions and descriptions of rules in this section can also be expressed in terms of the vector $(b_i^x)_{i\in I}$.

\citeauthor{duddy2025egalitarian}'s interpretation of Rawls' maximin principle is as follows. Suppose that in a given profile there exists an assignment $y$ and an agent $j$ such that $t_i^{y}(k) \geq t_j^{y}(k)$ for all values of~$k$. That is, agent $j$ has a weakly lower probability of being assigned to her top $k$ choices at $y$ than any other agent, for all values of $k$. Now assume that there exists another assignment $x$ at which the same agent $j$ is strictly worse off than at assignment $y$. That is, $t_j^{y}(k) \geq t_j^{x}(k)$ for all values of $k$, and there exists at least one value of $k$ for which this inequality is strict. This assignment $x$ should then be rejected by \citeauthor{duddy2025egalitarian}'s criterion, as it does not maximize the benefit of a least advantaged agent. 

\begin{definition}
    A random assignment $x$ is \textbf{even-handed} if there does not exist a random assignment $y$ and an agent $j\in I$ such that, for all $i\in I$, and for all $k=1,\ldots,n$,
    \[
    t^{y}_i(k) \geq t^{y}_j(k) \geq t^{x}_j(k),
    \]
    where the second inequality is strict for at least one value of $k$.
\end{definition}

Note that a given profile may admit multiple even-handed assignments. A rule is said to be even-handed if it selects an even-handed assignment for every profile.

\cite{duddy2025egalitarian} proposes three rules that are even-handed and sd-efficient. In particular, we will highlight two rules that are defined by lexicographically optimizing a differently ordered vector of the cumulative probabilities than the one that is used in the construction of the Rawlsian rule.

Denote by $(T_1^x(k), \ldots, T_n^x(k))$ the elements of $(t_1^x(k), \ldots, t_n^x(k))$ listed in non-decreasing order, for some value of $k$. That is, $T_1^x(k)$ denotes the smallest probability with which any agent is assigned to her top $k$ objects at $x$, $T_2^x(k)$ denotes the second-smallest such probability, etc. Note that $T_i^x(n) = 1$ for all agents $i$ and for all assignments $x\in X$. Using this notation, \cite{duddy2025egalitarian} proposes the following two rules. Both rules ensure guarantees on the minimum probabilities with which any agent is assigned to their top $k$ objects, but the order of importance in which these probabilities are considered differs. 

The positive equality rule first maximizes the lowest probability with which any agent is assigned to her \textit{first} choice, then maximizes the lowest probability with which any agent is assigned to her first two choices, etc.

\begin{definition}
    The \textbf{positive equality} rule selects the assignment $x$ that lexicographically maximizes the vector
    \[\left(T_1^x(1), T_1^x(2), \ldots, T_1^x(n), T_2^x(1), T_2^x(2), \ldots, T_2^x(n), \ldots, T_n^x(1), T_n^x(2), \ldots, T_n^x(n)\right).\]
\end{definition}

The prudent equality rule, on the other hand, first maximizes the lowest probability with which any agent is assigned to her first $n-1$ choices, then maximizes the lowest probability with which any agent is assigned to her first $n-2$ choices, etc. This is equivalent to first minimizing the probability with which any agent is assigned to her last choice, then minimizing the probability with which any agent is assigned to her last two choices, etc.

\begin{definition}
    The \textbf{prudent equality} rule selects the assignment $x$ that lexicographically maximizes the vector
    \[\left(T_1^x(n), T_1^x(n-1), \ldots, T_1^x(1), T_2^x(n), T_2^x(n-1), \ldots, T_2^x(1), \ldots, T_n^x(n), T_n^x(n-1), \ldots, T_n^x(1)\right).\]
\end{definition}

Note that the Rawlsian rule can also be expressed using this notation as the rule that selects the assignment $x$ that lexicographically maximizes the vector
\[\left(T_1^x(n), T_2^x(n), \ldots, T_n^x(n), T_1^x(n-1), T_2^x(n-1), \ldots, T_n^x(n-1), \ldots, T_1^x(1), T_2^x(1), \ldots, T_n^x(1)\right).\]





We conclude this section by observing that there is no logical relation between even-handed and Rawlsian assignments.

\begin{remark}
\label{remark:duddy1}
An even-handed assignment may not be Rawlsian. This follows immediately from the observation that some profiles admit multiple even-handed assignments, which is in contrast with the uniqueness of the Rawlsian assignment for any profile. 








\end{remark}

\begin{remark}
The Rawlsian assignment may violate even-handedness. Indeed, \citet[][Section 6]{duddy2025egalitarian} describes an example where the Rawlsian assignment violates the even-handedness criterion.  








\end{remark}

\subsection{Relation with strategyproofness axioms by \texorpdfstring{\cite{mennle2021partial}}{Mennle \& Seuken (2021)}}
\label{SP_axioms_Mennle_Seuken}
\cite{mennle2021partial} show that a rule is strategyproof if and only if it satisfies the axioms of swap monotonicity, upper invariance and lower invariance. Following the result by \cite{carroll2012local}, these three axioms only consider misreports in which the order of two consecutively ranked objects in a preference list is swapped. Upper invariance is defined in Definition \ref{def:upper_invar}. 
\begin{definition}
\label{def:swap_monoton}
    A rule $\phi$ is \textbf{swap monotonic} if, for all agents $i\in I$, all preference profiles $(\succeq_i,\succeq_{-i})$, and all misreports $\succeq_i'$ that are obtained by swapping two consecutively ranked objects $a$ and $b$, i.e., $a\succeq_i b$ but $b\succeq'_i a$, one of the following two conditions holds:
    \begin{itemize}
        \item either: $\phi_i(\succeq_i,\succeq_{-i}) = \phi_i(\succeq'_i,\succeq_{-i})$,
        \item or: $\phi_{i,b}(\succeq'_i,\succeq_{-i}) > \phi_{i,b}(\succeq_i,\succeq_{-i})$.
    \end{itemize}
\end{definition}

\begin{definition}
    \label{def:lower_invar}
    A rule $\phi$ is \textbf{lower invariant} if, for all agents $i\in I$, all preference profiles $(\succeq_i,\succeq_{-i})$, and all misreports $\succeq_i'$ that are obtained by swapping two consecutively ranked objects $a$ and $b$, i.e., $a\succeq_i b$ but $b\succeq'_i a$, it holds that $\phi_{ij}(\succeq_i,\succeq_{-i})=\phi_{ij}(\succeq'_i,\succeq_{-i})$ for all objects $j$ for which $b\succeq_i j$.
\end{definition}

The intuition behind Definition~\ref{def:swap_monoton} is that a rule is swap monotonic if increasing the preference for an object $b$ either increases the probability of being assigned to object $b$, or does not affect the agent's assignment probabilities at all. Moreover, a rule is upper (resp.\ lower) invariant if swapping the order of two consecutively ranked objects $a$ and $b$ does not affect the assignment probabilities of the objects that are more preferred than $a$ (resp.\ less preferred than $b$).

\begin{remark} {The Rawlsian rule satisfies lower invariance, because group lower invariance implies lower invariance (see Proposition \ref{prop:group_lower_invar}).}



\end{remark}

\begin{remark} {The Rawlsian rule violates swap monotonicity.}
    \begin{proof}
         Let $I=\{1,2,3\}, O = \{a,b,c\}$, and consider the preference profile $\succeq$, and the Rawlsian assignment $x(\succeq)$ at $\succeq$.
    \[
\begin{array}{c|ccc}
\succeq_1 & c & a & b \\ 
\succeq_2 & b & c & a \\ 
\succeq_3 & b & c & a \\ 
\end{array} \qquad \qquad
x(\succeq) = 
\begin{pmatrix}
1 & 0 & 0 \\
0 & \frac{1}{2}  & \frac{1}{2}  \\
0 &  \frac{1}{2} & \frac{1}{2}
\end{pmatrix}.
\]
If agent 1 alternatively reveals $\succeq'_1 = (c,b,a)$, swapping the order of objects $a$ and $b$, the Rawlsian assignment $x'(\succeq'_1, \succeq_{-1})$ becomes
\begin{equation*}
x'(\succeq'_1, \succeq_{-1}) = 
\begin{pmatrix}
\frac{1}{3} & 0 & \frac{2}{3} \\
\frac{1}{3} & \frac{1}{2}  & \frac{1}{6}  \\
\frac{1}{3} &  \frac{1}{2} & \frac{1}{6}
\end{pmatrix}.
\end{equation*}

The probability of being assigned to object $b$ for agent $1$ remains the same, while the probability of being assigned to object $a$ decreases, which violates the swap monotonicity axiom.

Note that this example also illustrates a violation of upper invariance: swapping the order of objects $a$ and $b$ causes the probability of agent $1$ being assigned to object $c$ to decrease.
    \end{proof}
\end{remark}

\subsection{Group lower invariance of other rules}
\label{app:group_lower_invar_other_rules}
In this section, we show that PS and RSD, as well as the positive and prudent equality rules by \cite{duddy2025egalitarian} all violate group lower invariance.

\begin{example}[PS violates group lower invariance]
We show that PS violates even the weaker notion of lower invariance, which implies a violation of group lower invariance. The following example is due to \cite{mennle2021partial}. Let $I=\{1,2,3\}$, $O=\{a,b,c\}$, and consider preference profile $\succeq$.
    \[
\begin{array}{c|ccc}
\succeq_1 & a & b & c  \\ 
\succeq_2 & b & a & c  \\ 
\succeq_3 & b & c & a \\ 
\end{array}
\]
Agent 1 receives probabilities $(\frac{3}{4}, 0, \frac{1}{4})$ for objects $a,b,c$, respectively, under PS at $\succeq$. However, if agent 1 reports $b\succeq'_1 a \succeq'_1 c$, profiles $\succeq$ and $(\succeq'_1, \succeq_{-1})$ are top-2-shuffled. In that case, she receives probabilities $(\frac{1}{2}, \frac{1}{3}, \frac{1}{6})$ for objects $a,b,c$, respectively, under PS at $\succeq'$. Hence, the probability with which agent 1 is assigned to object $c$ changes.\hfill$\qed$
\end{example}

\begin{example}[RSD violates group lower invariance]
Let $I=\{1,2,3,4\}$, $O=\{a,b,c,d\}$, and consider the preference profile $\succeq$, and the assignment probabilities under RSD at $\succeq$.
    \[
\begin{array}{c|cccc}
\succeq_1 & a & b & c &d \\ 
\succeq_2 & a & b & c &d \\ 
\succeq_3 & a & b & d &c \\ 
\succeq_4 & a & b & d &c \\ 
\end{array}
\qquad \qquad 
RSD(\succeq) = 
\begin{pmatrix}
\frac{1}{4} & \frac{1}{4} & \frac{5}{12} & \frac{1}{12}\\
\frac{1}{4} & \frac{1}{4} & \frac{5}{12} & \frac{1}{12}\\
\frac{1}{4} & \frac{1}{4} & \frac{1}{12} & \frac{5}{12} \\
\frac{1}{4} & \frac{1}{4} & \frac{1}{12} & \frac{5}{12}
\end{pmatrix}.
\]
However, if agents 1 and 2 change their preferences to $\succeq'_1$ and $\succeq'_2$, we obtain the following profile and corresponding RSD assignment:
    \[
\begin{array}{c|cccc}
\succeq'_1 & b & a & c &d \\ 
\succeq'_2 & a & c & b &d \\ 
\succeq_3 & a & b & d &c \\ 
\succeq_4 & a & b & d &c \\ 
\end{array}
\qquad \qquad 
RSD(\succeq'_1, \succeq'_2, \succeq_3, \succeq_4) = 
\begin{pmatrix}
0 & \frac{7}{12} & \frac{1}{4} & \frac{1}{6}\\
\frac{1}{3} & 0 & \frac{7}{12} & \frac{1}{12}\\
\frac{1}{3} & \frac{5}{24} & \frac{1}{12} & \frac{3}{8} \\
\frac{1}{3} & \frac{5}{24} & \frac{1}{12} & \frac{3}{8}
\end{pmatrix}.
\]
Note that profiles $\succeq$ and $(\succeq'_1, \succeq'_2, \succeq_3, \succeq_4)$ are top-3-shuffled. Nevertheless, the probability with which agent 1 is assigned to her fourth choice ($d$) changes.\hfill$\qed$
\end{example}

\begin{example}[Positive equality rule violates group lower invariance]
Let $I=\{1,2,3\}$, $O=\{a,b,c\}$, and consider the following two preference profiles
\[
\begin{array}{c|ccc}
\succeq_1 & a & b & c \\
\succeq_2 & a & b & c \\
\succeq_3 & a & c & b \\
\end{array}
\qquad
\text{and}
\qquad
\begin{array}{c|ccc}
\succeq_1 & a & b & c \\
\succeq_2 & a & b & c \\
\succeq'_3 & c & a & b \\
\end{array}.
\]
The positive equality rule selects assignments $x^+(\succeq)$ and $x^+(\succeq'_3, \succeq_{-3})$ at these profiles
\[
x^+(\succeq)=
\begin{pmatrix}
\frac13 & \frac12 & \frac16 \\
\frac13 & \frac12 & \frac16 \\
\frac13 & 0 & \frac23
\end{pmatrix}, \qquad
\qquad 
x^+(\succeq'_3, \succeq_{-3})=
\begin{pmatrix}
\frac12 & \frac12 & 0 \\
\frac12 & \frac12 & 0 \\
0 & 0 & 1
\end{pmatrix}.
\]
Note that profiles $\succeq$ and $(\succeq'_3, \succeq_{-3})$ are top-2-shuffled. Nevertheless, the probabilities with which agents 1 and 2 are assigned to their third choice ($c$) change from $\frac16$ to 0.\hfill$\qed$
\end{example}

\begin{example}[Prudent equality violates group lower invariance]
Let $I=\{1,2,3,4\}$, $O=\{a,b,c,d\}$, and consider the two preference profiles $\succeq$ and $\succeq'$
\[
\begin{array}{c|cccc}
\succeq_1 & a & b & d & c \\
\succeq_2 & b & a & c & d \\
\succeq_3 & b & a & c & d \\
\succeq_4 & b & c & d & a \\
\end{array}
\qquad
\text{and}
\qquad
\begin{array}{c|cccc}
\succeq'_1 & b & a & d & c \\
\succeq'_2 & a & b & c & d \\
\succeq'_3 & a & b & c & d \\
\succeq'_4 & c & b & d & a \\
\end{array}.
\]
These profiles are top-$2$-shuffled: for every agent, only the top two objects
are reshuffled, while positions $3$ and $4$ are unchanged.

The prudent equality rule selects assignments $x^P(\succeq)$ and $x^P(\succeq')$ at these profiles
\[
x^P(\succeq)=
\begin{pmatrix}
\frac12 & 0 & 0 & \frac12 \\
\frac14 & \frac13 & \frac{5}{12} & 0 \\
\frac14 & \frac13 & \frac{5}{12} & 0 \\
0 & \frac13 & \frac16 & \frac12
\end{pmatrix}, \qquad \qquad
x^P(\succeq')=
\begin{pmatrix}
0 & \frac12 & 0 & \frac12 \\
\frac12 & \frac14 & \frac14 & 0 \\
\frac12 & \frac14 & \frac14 & 0 \\
0 & 0 & \frac12 & \frac12
\end{pmatrix}.
\]
For agents $2$ and $3$, object $c$ is ranked in position $3$ in both profiles. However,
\[
x^P_{2c}(\succeq)=x^P_{3c}(\succeq)=\frac{5}{12},
\qquad
x^P_{2c}(\succeq')=x^P_{3c}(\succeq')=\frac14.
\]    \hfill$\qed$
\end{example}

\subsection{Relation with consistency and non-bossiness}
\label{app:cons_non_boss}
In this section, we formally define consistency and prove that the Rawlsian rule satisfies consistency and non-bossiness. We conclude by providing examples in which the positive and prudent equality rules by \cite{duddy2025egalitarian} violate both properties.

An allocation rule is \textbf{consistent} if the assignment it outputs for each problem ``agrees'' with the assignment it outputs for every reduced problem obtained when some agents leave with their initial assignments \citep{thomson2012axiomatics}. Consistency has been argued to contribute to the fairness, robustness, stability, and reinforcement of a mechanism, and we refer to \cite{thomson2012axiomatics} for a detailed discussion.

The consistency property has received extensive attention in the context of deterministic allocation rules \citep[see, e.g.,][]{ergin2000consistency, ehlers2007consistent}. We use the definition by \cite{han2016consistency} for a probabilistic interpretation of consistency.\footnote{Another probabilistic interpretation of consistency was proposed by \cite{chambers2004consistency}. He defined a rule to satisfy probabilistic consistency if, when an agent leaves, we randomly select the object with which she leaves according to her assignment probabilities. However, this extension of consistency to probabilistic rules is too strong for our purposes, as it characterizes, in combination with equal treatment of equals, the uniform rule, which gives identical probabilities to each agent-object pair. } It requires that when an agent leaves with her assignment probabilities, the rule assigns the same assignment probabilities for the remaining agents and the remaining capacities of the objects. Because this implies that objects need not have unit capacity after a set of agents leaves, we adopt a generalized notion of a problem instance. 

Given a capacity vector $(q_o)_{o\in O}$, we denote an instance as $e=(I,O,\succeq, q)$. Moreover, given an agent subset $S\subseteq I$ and an allocation $x$, we define the restricted instance $e \vert_{S,x}$ that is obtained when the agents in $S$ leave with their assignment probabilities as the following instance:
\begin{enumerate}[label=(\roman*)]
    \item $I \vert_{S,x}=I \setminus S$,
    \item $O\vert_{S,x}=O$,
    \item $\succeq \vert_{S,x}=(\succeq_i)_{i \in I\setminus S}$,
    \item $q\vert_{S,x}$ is such that the capacity of each object $o \in O$ is $q_o - \sum_{i\in S}x_{io}$.
\end{enumerate}

 Note that our definitions of random assignments naturally extend.\footnote{A random assignment for an instance $e=(I,O,\succeq, q)$ is then defined by $x=(x_i)_{i\in I}$, where each $x_i$ is a probability distribution over $O$, and for every $o\in O$, $\sum_{i\in I}x_{io} \leq q_o$. The generalized Birkhoff-von Neumann theorem \citep{budish2013designing} states that any random assignment for this generalized instance can be written as a convex combination over deterministic assignments.} 

 The following definition extends the consistency notion to probabilistic settings.

 \begin{definition}
 \label{def:cons}
    A rule $\phi$ satisfies \textbf{consistency} if, for each instance $e = (I,O,\succeq, q)$, where $q_o = 1$ for each $o\in O$, it holds for each $S\subseteq I$, and for each $i\in I\setminus S$, that 
    \begin{equation*}
        \phi_i(e\vert_{S,\phi(e)})=\phi_i(e).
    \end{equation*}
\end{definition}

Note that we will restrict our attention to instances where the sum of the capacities in the original market is equal to the number of agents. That is, we will only consider instances $(I,O,\succeq, q)$ where $\sum_{o\in O} q_o = |I|$. Definition \ref{def:cons} can be reformulated for arbitrary capacities in the original market.

We continue to show that the Rawlsian rule satisfies consistency and non-bossiness. \medskip

\noindent \textbf{Proof of Proposition \ref{prop:cons_non_boss}.}
\begin{proof}

We first prove that the Rawlsian rule satisfies consistency.

Consider an instance $e=(I,O,\succeq, q)$, with $q_o = 1$ for each $o\in O$. Let $x = \phi(e)$ denote the assignment probabilities by the Rawlsian rule in the original instance $e$. Consider a subset $S\subseteq I$ of the agents, and the restricted instance $e\vert_{S,x}$. Let $x' = \phi(e\vert_{S,x})$ denote the assignment probabilities of the agents in $I\setminus S$ in the modified instance $e\vert_{S,x}$.

Assume, for contradiction, that the assignment probabilities in $x$ and $x'$ differ for the agents in $I\setminus S$. We denote the assignment probabilities of a subset $S\subseteq I$ of the agents in assignment $x$ by $x_S$.

We know that $x'_{I\setminus S}$ R-dominates $x_{I\setminus S}$ in instance $e\vert_{S,x}$. We define the following aggregated assignment:
\begin{equation*}x'_i = 
    \begin{cases}
    x'_i &\text{if } i\in I\setminus S,\\
    x_i &\text{otherwise.}
\end{cases}
\end{equation*}
By construction of the capacities in the modified instance $e\vert_{S,x}$, $x'$ is a feasible assignment in instance $e$. The Rawlsian-dominance relationship between two assignments is not affected when extending both assignments by adding the same agents with the same assignment probabilities and preferences. Because the preferences of the agents in $I\setminus S$ are identical in instances $e\vert_{S,x}$ and $e$, this implies that $x'$ R-dominates $x$ in instance $e$. However, this is in contradiction with the fact that $x$ is the unique Rawlsian assignment in instance~$e$.

It is well-known that consistency implies non-bossiness \citep{thomson2016non}. For completeness, we include a proof of this result.

Consider an instance $e=(I,O,\succeq, q)$, where $q_o = 1$ for each $o\in O$. Consider an agent $i\in I$ who manipulates by reporting $\succeq'_i$, resulting in instance $e'=(I,O,(\succeq'_i, \succeq_{-i}), q)$, but receives the same probabilities by the Rawlsian rule $\phi$, i.e., $\phi_i(e) = \phi_i(e')$. Then, the restricted instances after removing $i$ from either instance are identical, i.e., $e|_{\{i\},\phi(e)} = e'|_{\{i\},\phi(e')}$. Hence, the Rawlsian rule finds the same probabilities: $\phi(e|_{\{i\},\phi(e)}) = \phi(e'|_{\{i\},\phi(e')})$. Because the Rawlsian rule is consistent, it holds for each agent $j\in I\setminus \{i\}$ that $\phi_j(e) = \phi_j(e|_{\{i\},\phi(e)}) = \phi_j(e'|_{\{i\},\phi(e')}) = \phi_j(e')$.
\end{proof}

We conclude this section by showing that the positive and prudent equality rules by \cite{duddy2025egalitarian} violate non-bossiness. Hence, they also violate consistency.

\begin{example}[Positive equality violates non-bossiness] Let $I=\{1,2,3,4\}$, $O=\{a,b,c,d\}$, and consider the following preference
profile $\succeq$, and the assignment $x^+(\succeq)$ by the positive equality rule at $\succeq$.
\[
\begin{array}{c|cccc}
\succeq_1 & b & a & c & d\\
\succeq_2 & d & a & b & c\\
\succeq_3 & b & c & d & a\\
\succeq_4 & d & a & c & b
\end{array}
\qquad \qquad
x^+(\succeq)=
\begin{pmatrix}
\frac13 & \frac12 & \frac16 & 0\\
\frac13 & 0 & \frac16 & \frac12\\
0 & \frac12 & \frac12 & 0\\
\frac13 & 0 & \frac16 & \frac12
\end{pmatrix}.
\]
Consider the following profile $(\succeq'_3, \succeq_{-3})$, in which agent 3 swaps the order of objects $c$ and $d$. The positive equality rule finds the following assignment $x^+(\succeq'_3, \succeq_{-3})$ at this profile.
\[
\begin{array}{c|cccc}
\succeq_1 & b & a & c & d\\
\succeq_2 & d & a & b & c\\
\succeq'_3 & b & d & c & a\\
\succeq_4 & d & a & c & b
\end{array}
\qquad \qquad
x^+(\succeq'_3, \succeq_{-3})=
\begin{pmatrix}
\frac14 & \frac12 & \frac14 & 0\\
\frac12 & 0 & 0 & \frac12\\
 0 & \frac12 & \frac12 & 0\\
\frac14 & 0 & \frac14 & \frac12
\end{pmatrix}.
\]
Note that the allocation of agent 3 is unaffected by her manipulation, whereas the allocations of agents 1, 2, and 4 are different at $\succeq$ and $(\succeq'_3, \succeq_{-3})$. Hence, the positive equality rule violates non-bossiness.
\hfill $\qed$
\end{example}

\begin{example}[Prudent equality violates non-bossiness] Let $I=\{1,2,\ldots,7\}$, $O=\{a,b,\ldots,g\}$, and consider the following preference
profile $\succeq$.\footnote{This example is inspired by the instance of \citet[][Section 6]{duddy2025egalitarian}, but three dummy agents 5, 6, and 7 are added, as well as three dummy objects $e$, $f$, and $g$.} The assignment $x^P(\succeq)$ by the prudent equality rule at $\succeq$ is displayed below, where the probabilities for the objects are ranked in decreasing order of preference.
    \[
\begin{array}{c|ccccccc}
\succeq_1 & a & b & c & \ldots\\ 
\succeq_2 & a & b & c & \ldots\\ 
\succeq_3 & a & d & c  & \ldots\\
\succeq_4 & d & a & b  & \ldots\\ 
\succeq_5 & a & e & b  & \ldots\\ 
\succeq_6 & a & f & b  & \ldots\\ 
\succeq_7 & a & g & b  & \ldots\\ 
\end{array}\qquad \qquad
\begin{array}{c|cccc}
x^P_1(\succeq) & \frac16 & \frac12 & \frac13 & \ldots\\ 
x^P_2(\succeq) & \frac16 & \frac12 & \frac13 & \ldots\\
x^P_3(\succeq) & \frac23 & 0 & \frac13 & \ldots\\
x^P_4(\succeq) & 1 & 0 & 0  & \ldots\\ 
x^P_5(\succeq) & 0 & 1 & 0  & \ldots\\ 
x^P_6(\succeq) & 0 & 1 & 0  & \ldots\\ 
x^P_7(\succeq) & 0 & 1 & 0  & \ldots\\ 
\end{array}
\]
Consider the following profile $\succeq'=(\succeq'_7, \succeq_{-7})$, in which agent 7 swaps the order of objects $a$ and $g$. The prudent equality rule finds the following assignment $x^P(\succeq')$ at this profile, where the objects are ranked in decreasing order of preference.
 \[
\begin{array}{c|ccccccc}
\succeq_1 & a & b & c & \ldots\\ 
\succeq_2 & a & b & c & \ldots\\ 
\succeq_3 & a & d & c  & \ldots\\
\succeq_4 & d & a & b  & \ldots\\ 
\succeq_5 & a & e & b  & \ldots\\ 
\succeq_6 & a & f & b  & \ldots\\ 
\succeq'_7 & g & a & b  & \ldots\\ 
\end{array}\qquad \qquad
\begin{array}{c|cccc}
x^P_1(\succeq') & \frac13 & \frac13 & \frac13 & \ldots\\ 
x^P_2(\succeq') & \frac13 & \frac13 & \frac13 & \ldots\\
x^P_3(\succeq') & \frac13 & \frac13 & \frac13 & \ldots\\
x^P_4(\succeq') & \frac23 & 0 & \frac13  & \ldots\\ 
x^P_5(\succeq') & 0 & 1 & 0  & \ldots\\ 
x^P_6(\succeq') & 0 & 1 & 0  & \ldots\\ 
x^P_7(\succeq') & 1 & 0 & 0  & \ldots\\ 
\end{array}
\]
Note that the allocation of agent 7 is unaffected by her manipulation, whereas the allocations of agents 1, 2, 3, and 4 are different at $\succeq$ and $\succeq'$. Hence, the prudent equality rule violates non-bossiness. 

The intuition behind this change is that the manipulation of agent 7 controls which cumulative probabilities should be maximized first. At profile $\succeq$ the prudent equality rule will first maximize the fourth-lowest probability with which any agent is assigned to her first two choices, thus enforcing agents 4-7 to be assigned to their first two choices with probability 1 (implying that $x^P_{4b}(\succeq) = 0$). At profile $\succeq'$, however, the prudent equality rule will first maximize the third-lowest probability with which any agent is assigned to her first choice, thus enforcing that agents 1-3 are assigned to their first choice with at least $\frac13$ (implying that $x^P_{1a}(\succeq')=x^P_{2a}(\succeq')=x^P_{3a}(\succeq') = \frac13$). This, in turn, implies that agents 1-2 must also be assigned to their second choice $b$ with $\frac13$, and hence agent 4 must also be assigned to her third choice $b$ with  $\frac13$.
\hfill $\qed$
\end{example}

\subsection{The MTAV rule.}\label{MTAV}
In this section, we define the MTAV following \cite{Paleo2021}. Let $\mathcal{M}$ be the set of deterministic assignments. Each assignment in $\mathcal{M}$ is represented by an $n \times n$ matrix $M$ where $m_{ij} =1$ if, and only if, agent $i$ receives object $j$. Given two matrices $M$ and $M'$ in $\mathcal{M}$, we denote by $M \odot M'$ the matrix where each element is the product of the corresponding elements of $M$ and $M'$: $(M \odot M')_{ij}=(m_{ij})(m'_{ij})$. Also, denote by $\max(M)$ and $sum(M)$ the maximum and the sum of the elements of $M$, respectively. 

Given a problem $(\succeq)$, define the matrix $P$ with agents' preferences, where $P_{ij}$ is the rank of object $j$ in agent $i$'s preferences (equivalently, $P_{ij}=r_{ij}$). The \textbf{MTAV} is defined as follows.

\begin{enumerate}
    \item For each $M \in \mathcal{M}$ compute $P\odot M$.
    \item Compute $\max (P\odot M)$.
    \item Select the assignments that minimize $\max (P\odot M)$.
    \item Among the assignments selected in Step 3, select those that minimize $sum (P\odot M)$.
    \item If multiple assignments are selected in the last step, take one assignment at random. 
\end{enumerate}

To implement MTAV, we have used its publicly available code (\url{https://github.com/eze91/MTAV}). In the last step, this code will simply generate one of the multiple assignments, rather than generating all of them and selecting each with equal probability. Note that both approaches are not equivalent, as discussed in detail in \cite{demeulemeester2023fair}.

\section{Proof of Proposition \ref{large}}\label{large_proof}

We prove Proposition~\ref{large} and refer to Section \ref{proofs_large2} of the Online Appendix for the proof of Proposition~\ref{prop:LB_Rawls}.

\begin{proof}

Consider a complete bipartite graph $(I,O,E)$ where one set of nodes $I$ is the set of agents, and the other set of nodes $O$ is the set of objects. There is an edge in $E$ between an agent and an object if the agent ranks the object. When we choose exactly $k$ edges for each agent, this is the same as considering the $k$ most preferred objects of each agent.

For a given preference profile $\succeq$, denote by $\tilde{\mathcal{G}}(n,k)$ the resulting bipartite graph when we consider the $k$ most preferred objects of each agent. The maximum rank of the Rawlsian assignment, $r^{\text{Rawls}}_{max}$, is equal to some value $k\in\mathbb{N}$ when there exists a perfect matching in $\tilde{\mathcal{G}}(n,k)$, but not in $\tilde{\mathcal{G}}(n,k-1)$. Indeed, if there does not exist a perfect matching in $\tilde{\mathcal{G}}(n,k-1)$, then not all agents can be assigned to an object they rank $(k-1)$-th or better. Because, in our model, the number of agents equals the number of objects and agents have complete preferences, this implies that there must exist an agent who is assigned to an object ranked $k$-th or worse. As such, the maximum rank of the Rawlsian assignment is higher than $k-1$.

Additionally, let $\tilde{X}_{n,k}$ be the indicator of the event that $\tilde{\mathcal{G}}(n,k)$ contains a perfect matching, i.e., 
    
        $$
        \tilde X_{n,k} =  \begin{cases}
      1 & \text{ if } \tilde{\mathcal{G}}(n,k) \text{ has a perfect matching}, \\
      0        & \text{ otherwise.}
    \end{cases}
        $$
        
Let
\[
A_{n,k}:=\{\widetilde X_{n,k}=1\}
\]
be the event that the graph $\widetilde{\mathcal{G}}(n,k)$ has a perfect matching. Since $\widetilde{\mathcal{G}}(n,k)$ contains all edges from $\widetilde{\mathcal{G}}(n,k-1)$ plus additional edges, we have
\[
A_{n,k-1}\subseteq A_{n,k}.
\]

Therefore, the event $A_{n,k}$ can be decomposed into two disjoint parts:
\[
A_{n,k}
=
\big(A_{n,k}\cap A_{n,k-1}\big)
\cup
\big(A_{n,k}\cap A_{n,k-1}^c\big).
\]
Since $A_{n,k-1}\subseteq A_{n,k}$, the first term satisfies
\[
A_{n,k}\cap A_{n,k-1}=A_{n,k-1}.
\]
Hence,
\[
A_{n,k}
=
A_{n,k-1}
\cup
\big(A_{n,k}\cap A_{n,k-1}^c\big),
\]
where the two events on the right-hand side are disjoint. Taking probabilities gives
\[
\mathbb P(A_{n,k})
=
\mathbb P(A_{n,k-1})
+
\mathbb P(A_{n,k}\cap A_{n,k-1}^c).
\]
Rearranging,
\[
\mathbb P(A_{n,k}\cap A_{n,k-1}^c)
=
\mathbb P(A_{n,k})
-
\mathbb P(A_{n,k-1}).
\]

Therefore,
\begin{equation}
    \mathbb P(r^{\mathrm{Rawls}}_{\max}=k)
=
\mathbb P(\widetilde X_{n,k}=1)
-
\mathbb P(\widetilde X_{n,k-1}=1). \label{eq:rank}
\end{equation}

Additionally, denote by $\mathcal{G}(n,p)$ the random bipartite graph with $n$ nodes in each set, and in which each edge is independently selected with probability $p$. Similarly to $\tilde{\mathcal{G}}(n,k)$, let $X_{n,p}$ denote the event that $\mathcal{G}(n,p)$ contains a perfect matching. 
        
Denote the minimum degree of a graph~$\mathcal{G}$ by $\delta(\mathcal{G})$. \citeapp{erdHos1966existence} and \citeapp{erdHos1968random} showed for various types of graphs that, in the limit (when $n\rightarrow +\infty$), the probability that a perfect matching exists is equal to the probability that the minimum degree of the graph is at least one (no isolated vertices). Because half of the nodes in $\tilde{\mathcal{G}}(n,k)$ have a guaranteed degree of~$k$, the probability of having a minimum degree of at least one is not smaller in $\tilde{\mathcal{G}}(n,k)$ than in $\mathcal{G}(n,\frac{k}{n})$. As a consequence, it holds that, for every $k \in \{1,\ldots,n\}$:
        
     \begin{equation}
            \lim_{n\to\infty}\mathbb{P}(\tilde{X}_{n,k} = 1) = \lim_{n\to\infty} \mathbb{P}(\delta(\tilde{\mathcal{G}}(n,k))\geq 1) \geq \lim_{n\to\infty} \mathbb{P}(\delta(\mathcal{G}(n,\frac{k}{n})\geq 1)) = \lim_{n\to\infty}\mathbb{P}(X_{n,\frac{k}{n}} = 1).\label{eq:linkErdos}
        \end{equation}


        The expected maximum rank of the Rawlsian assignment is equal to:
        \begin{align}
         \mathbb{E}(r_n^{max}) &= \sum_{k=1}^n k\cdot \mathbb{P}(r^{max}_n = k).
        \end{align}
       By applying Equation~(\ref{eq:rank}), we obtain that
        \begin{align}
            \mathbb{E}(r^{\text{Rawls}}_{max}) &= \sum_{k=1}^n k\cdot \left(\mathbb{P}(\tilde{X}_{n,k} = 1) - \mathbb{P}(\tilde{X}_{n,k-1} = 1)\right)\\
            &=n - \sum_{k=1}^{n-1} \mathbb{P}(\tilde{X}_{n,k} = 1), \label{fact1}
        \end{align}
        where $\mathbb{P}(\tilde{X}_{n,0} = 1)=0$, and the second equality follows from the fact that $\mathbb{P}(\tilde{X}_{n,n} = 1) = 1$.

We are interested in evaluating this expression in the limit. By applying Equation~(\ref{eq:linkErdos}), we can bound the limit of the expected maximum rank by the Rawlsian assignment in terms of the probability that a perfect matching exists in the well-studied class of bipartite random graphs~$\mathcal{G}(n,p)$ in which each edge is selected with uniform probability~$p$.
        \begin{align}
            \lim_{n\to\infty}\mathbb{E}(r^{\text{Rawls}}_{max}) &= \lim_{n\to\infty}\left(n - \sum_{k=1}^{n-1} \mathbb{P}(\tilde{X}_{n,k} = 1)\right)\\
            &\leq \lim_{n\to\infty}\left(n - \sum_{k=1}^{n-1} \mathbb{P}(X_{n,\frac{k}{n}} = 1)\right).\label{eq:rank0}
        \end{align}

After some manipulations, the expression of which we take the limit can be written as:
\begin{align}
    &n-\sum_{k=1}^{n-1}\mathbb{P}(X_{n,\frac{k}{n}} = 1)= n-\sum_{k=1}^{n-1} (1 - (1-\mathbb{P}(X_{n,\frac{k}{n}} = 1)))\\
    &= n-\sum_{k=1}^{n-1}1 + \sum_{k=1}^{n-1}(1-\mathbb{P}(X_{n,\frac{k}{n}} = 1))\\
    &= 1 + \sum_{k=1}^{n-1}(1-\mathbb{P}(X_{n,\frac{k}{n}} = 1))
    \label{eq:last}
\end{align}

We know from \citeapp{erdHos1968random} that 
        \begin{equation}
\label{eq:prob_perfect_proof}
    \lim_{n\to\infty}(1-\mathbb{P}(X_{n,\frac{\ln(n)+c_n}{n}} = 1)) =\\
    \begin{cases}
        1 &  \text{ if } c_n  \to -\infty\\
        1-e^{-2e^{-c}} & \text{ if } c_n \to c\\
        0 & \text{ if } c_n \to \infty
    \end{cases}.
\end{equation}

We can then rewrite Equation (\ref{eq:rank0}) as
\begin{align}
    \lim_{n\to\infty}\mathbb{E}(r^{\text{Rawls}}_{max}) &\leq 1 + \lim_{n\to\infty}\left(\sum_{k=1-\ln(n)}^{n-1 - \ln(n)} (1-\mathbb{P}(X_{n,\frac{\ln(n)+k}{n}}=1))\right)\\
    &= 1 + \lim_{n\to\infty}\left(\sum_{k=1-\ln(n)}^{\lfloor\ln(n)\rfloor - \ln(n) -1} (1-\mathbb{P}(X_{n,\frac{\ln(n)+k}{n}}=1))\right. \notag\\ &\quad \quad \left.+ \sum_{k=\lfloor\ln(n)\rfloor - \ln(n)}^{n-1-\ln(n)}(1-\mathbb{P}(X_{n,\frac{\ln(n)+k}{n}}=1))\right)\label{eq:start}
\end{align} 
Note that these expressions are a bit unfamiliar because the indices of the summations are not integer: the interpretation is that, starting from the lower index, we sum for values of $k$ that increase with step size one. 

If we can show that the limit of the first summation is well-defined, and that second summation is a constant, we can rewrite this expression as:\footnote{In fact, we will show that this expression is similar to $\lim_{n\to\infty}(\ln(n) + constant) = \lim_{n\to\infty}(\ln(n)) + constant$.}
\begin{align}
    &1+  \lim_{n\to\infty}\left(\sum_{k=1-\ln(n)}^{\lfloor\ln(n)\rfloor - \ln(n)-1} (1-\mathbb{P}(X_{n,\frac{\ln(n)+k}{n}}=1))\right) \notag\\ &\:
+\lim_{n\to\infty}\left( \sum_{k=\lfloor\ln(n)\rfloor - \ln(n)}^{n-1-\ln(n)}(1-\mathbb{P}(X_{n,\frac{\ln(n)+k}{n}}=1))\right)\label{eq:result}.
\end{align}

Because probabilities are non-negative, an easy bound on the limit of the first summation would simply be: 
\begin{align}
     \sum_{k=1-\ln(n)}^{\lfloor\ln(n)\rfloor - \ln(n)-1} (1-\mathbb{P}(X_{n,\frac{\ln(n)+k}{n}}=1)) \leq  \sum_{k=1-\ln(n)}^{\lfloor\ln(n)\rfloor - \ln(n)-1} 1\notag\\
     =\lfloor\ln(n)\rfloor-\ln(n) - 1 - (1-\ln(n)) + 1 = \lfloor\ln(n)\rfloor -1.\label{eq:easybound}
\end{align}
For the limit of the second summation, we can write that
\begin{align}
    \lim_{n\to\infty}\left( \sum_{k=\lfloor\ln(n)\rfloor - \ln(n)}^{n-1-\ln(n)}(1-\mathbb{P}(X_{n,\frac{\ln(n)+k}{n}}=1))\right) &\leq \lim_{n\to\infty}\left( \sum_{k=-1}^{n-1-\lceil\ln(n)\rceil}(1-\mathbb{P}(X_{n,\frac{\ln(n)+k}{n}}=1))\right)\\
    &\leq \lim_{n\to\infty}\left( \sum_{k=-1}^{+\infty}(1-\mathbb{P}(X_{n,\frac{\ln(n)+k}{n}}=1))\right)\label{eq:second}
\end{align}
Because $1-\mathbb{P}(X_{n,\frac{\ln(n)+k}{n}}=1)$ is a decreasing function in $k$, we find an upper bound on this expression by shifting the values for which this function is evaluated downwards with a positive value of $\lceil\ln(n)\rceil - \ln(n)$ (first inequality). 
Because $1-\mathbb{P}(X_{n,\frac{\ln(n)+k}{n}}=1))\geq0$, we obtain an upper bound by changing the ending index of the summation from $n-1-\lceil\ln(n)\rceil$ to $+\infty$ (second inequality). 

Now observe from (\ref{eq:prob_perfect_proof}) that $\lim_{n\to\infty}(1-\mathbb{P}(X_{n,\frac{\ln(n)+k}{n}}=1)) = 1-e^{-2e^{-k}}$ for a constant value of $k$. If we can show that 
\begin{equation}
    \sum_{k=-1}^{+\infty}1-e^{-2e^{-k}}
\end{equation} converges to a constant, then we can bring the limit inside the summation in Expression (\ref{eq:second}). We will use the ratio test to show the convergence of the series $a_k = 1-e^{-2e^{-k}}$.
\begin{align}
    \lim_{k\to\infty}\left|\frac{a_{k+1}}{a_k}\right| &= \lim_{k\to\infty}\left|\frac{1-e^{-2e^{-(k+1)}}}{1-e^{-2e^{-k}}}\right| = \lim_{k\to\infty}\frac{1-e^{-2e^{-(k+1)}}}{1-e^{-2e^{-k}}},
\end{align}
where the last equality holds because both the numerator and the denominator are positive. We can simplify the numerator as:
\begin{align}
    1-e^{-2e^{-(k+1)}} = 1-e^\frac{-2e^{-k}}{e}
\end{align}
For large values of $k$, both the numerator and the denominator go to zero. We can use L'Hôpital's rule, which states that
\begin{align}
    \lim_{k\to\infty}\frac{1-e^{-2e^{-(k+1)}}}{1-e^{-2e^{-k}}} &= \lim_{k\to\infty}\frac{\frac{d}{dk}\left(1-e^{-2e^{-(k+1)}}\right)}{\frac{d}{dk}\left(1-e^{-2e^{-k}}\right)}\\
    &= \lim_{k\to\infty}\frac{-2e^{-k-1}e^{-2e^{-k-1}}}{-2e^{-k}e^{-2e^{-k}}}\\
    &= \lim_{k \to \infty} \frac{e^{-k-1}}{e^{-k}} \cdot \frac{e^{-2e^{-k-1}}}{e^{-2e^{-k}}}\\
&= \lim_{k \to \infty} e^{-1} \cdot e^{2(e^{-k} - e^{-k-1})}.
\end{align}
As $e^{-1}$ is a constant term, and because the exponent $2(e^{-k} - e^{-k-1})  $ goes to zero for large values of $k$, this is equivalent to:
\begin{align}
    \lim_{k \to \infty} e^{-1} \cdot e^{2(e^{-k} - e^{-k-1})} = e^{-1} \cdot \lim_{k \to \infty} e^{2(e^{-k} - e^{-k-1})} = e^{-1} \cdot e^{0} = e^{-1}.
\end{align}
Note that $\frac{1}{e} \approx 0.3679 < 1$. As a consequence, $\sum_{k=-1}^{+\infty}1-e^{-2e^{-k}}$ converges to a constant. As a result, we can rewrite Expression (\ref{eq:start}) as Expression (\ref{eq:result}).

Putting the bounds of Expressions (\ref{eq:easybound}) and (\ref{eq:second}) together in Expression (\ref{eq:result}), we obtain:
\begin{align}
    \lim_{n\to\infty}\mathbb{E}(r^{Rawls}_{max}) &\leq 1+  \lim_{n\to\infty}\left(\sum_{k=1-\ln(n)}^{\lfloor\ln(n)\rfloor - \ln(n)-1} (1-\mathbb{P}(X_{n,\frac{\ln(n)+k}{n}}=1))\right) \notag\\ &\quad \quad +\lim_{n\to\infty}\left( \sum_{k=\lfloor\ln(n)\rfloor - \ln(n)}^{n-1-\ln(n)}(1-\mathbb{P}(X_{n,\frac{\ln(n)+k}{n}}=1))\right)\\
    &\leq \lim_{n\to\infty}\lfloor\ln(n)\rfloor + \sum_{k=-1}^{+\infty}\left(1-e^{-2e^{-k}}\right)\\
    &\approx \lim_{n\to\infty}\lfloor\ln(n)\rfloor + 2.77026.
\end{align}

\end{proof}

\newpage
\setcounter{section}{0}
\clearpage
\setcounter{page}{1} 
\appendix

\begin{center}
    \LARGE{\textbf{Online Appendix (not for publication)}}
\end{center}
\resetcounters

\section{Relation with rank efficiency}\label{rank}

Given an assignment $x$, we define $M^x(k)=\sum_{i \in I} b_i^x(k)$ for $k=1,\ldots,n$. So, for example, $M^x(n)$ is the sum of the probabilities with which each agent is assigned her least preferred option. 
Equivalently, $M^x(k)$ is the expected number of agents who receive an object ranked in position $k$ or lower at $x$.

\begin{definition}\label{featherstone}
An assignment, $x$, is said to \textbf{rank-dominate} another assignment, $y$, if the rank distribution of $y$ first-order stochastically dominates that of $x$, that is, if:

$$
M_{y}(k) \geq  M_{x}(k).
$$

for all ranks, $k$, with a strict inequality for at least one $k$. 
\medskip

\noindent If an assignment is not rank-dominated by any other assignment, it is \textbf{rank efficient}.
    
\end{definition}

Definition \ref{featherstone} is equivalent to the original definition of \citeapp{featherstone2020rank}, where the sum of the probabilities are computed from the most preferred object to $k$. The following example shows a problem where the Rawlsian assignment is not rank-efficient. 

\begin{remark}{A rank efficient assignment may not be Rawlsian.}\label{Rem_rank_eff}
\begin{proof}
    Consider the following problem:

\begin{center}

\[
\begin{array}{c|ccccc}
\succeq_1 & a & \fbox{d} & c & b &e\\ 
\succeq_2 & \fbox{b} & c & a & d & e\\ 
\succeq_3 & \fbox{c} & b & a & d & e\\ 
\succeq_4 & b & \fbox{a} & d & c & e\\ 
\succeq_5 & b & a & \fbox{d} & e & c\\ 
\end{array}
\]
\end{center}

 The following assignment is rank efficient:

$$y=
\begin{pmatrix}
1 & 0 & 0  & 0 & 0\\
0 & 1 & 0  & 0 & 0 \\
0 & 0 & 1  & 0 & 0 \\
0 & 0 & 0  & 1 & 0 \\
0 & 0 & 0  & 0 & 1 
\end{pmatrix},
$$

but it is not Rawlsian. Indeed, its associated vector is 

$$
B^y=(0,0,0,0,0;1,0,0,0,0;1,1,0,0,0;1,1,0,0,0;1,1,1,1,1).
$$

Consider the following assignment $x$ (boxed in agents' preferences):

$$x=
\begin{pmatrix}
0 & 0 & 0 & 1 & 0\\
0 & 1 & 0  & 0 & 0 \\
0 & 0 & 1  & 0 & 0 \\
1 & 0 & 0  & 0 & 0 \\
0 & 0 & 0  & 0 & 1 
\end{pmatrix}.
$$

It is easy to see that $x$ $R$-dominates $y$.
\end{proof}
\end{remark}

\begin{remark}
A Rawlsian assignment may not be rank-efficient. 
\begin{proof}
Consider the following problem:

\[
\begin{array}{c|ccc}
\succeq_1 & a & b & c \\ 
\succeq_2 & a & b & c \\ 
\succeq_3 & b & a & c \\ 
\end{array}
\]

The following is the Rawlsian assignment of the problem:
$$
x=
\begin{pmatrix}
 \frac{1}{2} &  \frac{1}{6} &  \frac{1}{3}  \\
 \frac{1}{2} &  \frac{1}{6} &  \frac{1}{3}  \\
0 &  \frac{2}{3} &  \frac{1}{3}
\end{pmatrix}.
$$
Consider the following assignment:
$$
y=
\begin{pmatrix}
1 & 0 & 0  \\
0 &  0 & 1  \\
0 & 1 & 0
\end{pmatrix}.
$$

Compute $M_x(k)$ and $M_y(k)$:

\begin{center}
\begin{tabular}{ |c|c|c|c| } 
 \hline
  & $M_j(3)$ & $M_j(2)$ & $M_j(1)$  \\ \hline 
 $x$ & 1 & $\frac{4}{3}$ & 3  \\  \hline
 $y$ & 1 & 1 & 3 \\ \hline
\end{tabular}.
\end{center}

Assignment $y$ Rank-dominates assignment $x$.

\end{proof}

\end{remark}

\section{Relation with fractional Boston assignment}
\label{frac_Boston}
The fractional Boston rule by \citeapp{bogomolnaia2015random}, which was further studied by \citeapp{chen2023fractional}, is the outcome of the following procedure.

\begin{itemize}
    \item In Step 1, all agents start by ``consuming'' their most preferred object simultaneously at equal speeds. An agent stops consuming when the object is exhausted. Step 1 ends by removing all agents who are assigned to their first-ranked object with probability one.
    \item In Step $k$, all remaining agents consume from their $k$-th preferred object. An agent stops consuming when the object is exhausted, or when the sum of the assignment probabilities of the agent equals one. Step $k$ ends by removing the agents for whom the sum of the assignment probabilities equals one.
\end{itemize}
The algorithm ends when the last agent is removed.

While both the fractional Boston rule and the $\sigma^B$-Rawlsian assignment for $\sigma^B = (2,3,\ldots,n)$ start by lexicographically maximizing the probabilities of being assigned to the most preferred object, the following example shows that they are not identical.
\begin{remark}
    The outcome of the fractional Boston rule does not correspond to the $\sigma_B$-minimal assignment for $\sigma^B = (2,3,\ldots,n)$. 
    
    \begin{proof}
        Consider the following problem with 7 agents:
                \[
\begin{array}{c|cccc}
\succeq_1 & a & b & e & \ldots \\ 
\succeq_2 & a & b & c & \ldots\\ 
\succeq_3, \succeq_4 & c & b & a & \ldots\\
\succeq_5, \succeq_6, \succeq_7 & d & a & e & \ldots\\
\end{array}
\]

The assignment probabilities for the three most preferred objects of each agent in the $\sigma^B$-minimal assignment and in the fractional Boston assignment are (element $x_{ij}$ is the probability with which agent $i$ is assigned to her $j$-th most preferred object):

\begin{equation*}
x^{\sigma^B} = 
\begin{pmatrix}
\frac{1}{2} & \frac{1}{4} & \bm{0} & \ldots\\
\frac{1}{2} & \frac{1}{4} & 0 & \ldots\\
\frac{1}{2} & \frac{1}{4} & 0 & \ldots\\
\frac{1}{2} & \frac{1}{4} & 0 & \ldots\\
\frac{1}{3} & 0 & {\frac{\bm{1}}{\bm{3}}} & \ldots\\
\frac{1}{3} & 0 & {\frac{\bm{1}}{\bm{3}}} & \ldots\\
\frac{1}{3} & 0 & {\frac{\bm{1}}{\bm{3}}} & \ldots
\end{pmatrix}, \quad x^{FB} = 
\begin{pmatrix}
\frac{1}{2} & \frac{1}{4} & {\frac{\bm{1}}{\bm{4}}} & \ldots\\
\frac{1}{2} & \frac{1}{4} & 0 & \ldots\\
\frac{1}{2} & \frac{1}{4} & 0 & \ldots\\
\frac{1}{2} & \frac{1}{4} & 0 & \ldots\\
\frac{1}{3} & 0 & {\frac{\bm{1}}{\bm{4}}} & \ldots\\
\frac{1}{3} & 0 & {\frac{\bm{1}}{\bm{4}}} & \ldots\\
\frac{1}{3} & 0 & {\frac{\bm{1}}{\bm{4}}} & \ldots
\end{pmatrix}.
\end{equation*}

The difference in the assignment probabilities for the object of third choice comes from the following reasoning. It can be shown that the $\sigma^B$-minimal assignment can be obtained by first lexicographically maximizing the vector of the probabilities with which the agents are assigned to the object of their first choice, then to lexicographically maximize the vector with the probabilities with which the agents are assigned to the objects of their first or second choice, etc. As such, after having fixed the probabilities of being assigned to each agent's two most preferred objects, the $\sigma^B$-minimal assignment maximizes the vector of the cumulative probabilities of being assigned to one of the three most preferred objects, and, therefore, divides object $e$ equally among agents 5,6, and 7. In the fractional Boston assignment, however, at the beginning of the thirds step, agents 1,5,6, and 7 start eating with equal speeds from object $e$, which allows agent 1 to consume one quarter of the object. Thus, in the $\sigma^B$-minimal assignment the minimum probability of being assigned to the three most preferred objects is $\frac{2}{3}$, while it is $\frac{7}{12}<\frac{2}{3}$ in the fractional Boston assignment.

    \end{proof}
    \end{remark}

\section{Rawlsian assignments and cardinal utilities}\label{cardinal}

In this section, we discuss the relation between the Rawlsian assignment and cardinal representations of agents' ordinal preferences. \citeapp{mclennan2002ordinal} characterizes sd-efficiency in terms of utilitarian welfare: an assignment is sd-efficient if and only if there exists a profile of cardinal utilities, consistent with agents' ordinal preferences, for which the assignment maximizes total expected utility. See also \citeapp{manea2008constructive}. 

One may wonder whether an analogous statement holds for the Rawlsian assignment. In particular: can the Rawlsian assignment be characterized as maximizing the expected utility of the worst-off agent for some cardinal representation? A first observation is that every sd-efficient assignment admits such a cardinal maximin representation. That is, if $x$ is sd-efficient, there exist cardinal utilities consistent with the ordinal preferences such that:

\begin{equation}\label{maxmin}
x \in \argmax_{x' \in X} \{\min_{i \in I} U_i(x')\},
\end{equation}

\noindent where $U_i(x')=\sum_{o \in O} x'_{io}u_{io}$, and $u_{io}$ is agent $i$'s cardinal utility for object $o$. The Rawlsian assignment is sd-efficient, so the result implies that for each problem there exist cardinal utilities under which the Rawlsian assignment maximizes the expected utility of the worst-off agent.

\begin{proposition}
\label{prop:max_worst_utility}
Consider a problem $\succeq$ and an sd-efficient assignment $x$ for problem $\succeq$. Then, there exists a cardinal representation of $\succeq$, denoted by $U = (U_i)_{i\in I}$, such that:

\[
x \in \argmax_{x' \in X} \{\min_{i \in I} U_i(x')\}.
\]

\end{proposition}

\begin{proof}
Let $x$ be an sd-efficient assignment. By \citeapp{mclennan2002ordinal}, there exists a cardinal representation $u=(u_i)_{i\in I}$ of the agents' ordinal preferences such that $x$ maximizes utilitarian welfare over the set of feasible assignments. That is,
$$
x \in \argmax_{z\in X}
\sum_{i\in I}\sum_{o\in O} z_{io}u_i(o).
$$

For each agent $i$, let
$$
\alpha_i=\sum_{o\in O}x_{io}u_i(o)
$$
be agent $i$'s expected utility under $x$. Define a new cardinal
representation $U=(U_i)_{i\in I}$ by
$$
U_i(o)=u_i(o)-\alpha_i
\qquad \text{for every } o\in O.
$$
Since $U_i$ is obtained from $u_i$ by adding a constant to the utility of all objects, it represents the same ordinal preference as $u_i$. Moreover, under assignment
$x$, every agent obtains expected utility zero:
$$
\sum_{o\in O}x_{io}U_i(o)
=
\sum_{o\in O}x_{io}u_i(o)-\alpha_i
=
0
\qquad \text{for every } i\in I.
$$
Thus,
$$
\min_{i\in I}\sum_{o\in O}x_{io}U_i(o)=0.
$$

We claim that $x$ maximizes the expected utility of the worst-off agent under
the cardinal representation $U$. Suppose, by contradiction, that there exists an assignment $y\in X$ such that
$$
\min_{i\in I}\sum_{o\in O}y_{io}U_i(o)>0.
$$
Then every agent obtains strictly positive expected utility under $y$:
$$
\sum_{o\in O}y_{io}U_i(o)>0
\qquad \text{for every } i\in I.
$$
Summing over agents, we obtain
$$
\sum_{i\in I}\sum_{o\in O}y_{io}U_i(o)>0.
$$
However, since $U_i(o)=u_i(o)-\alpha_i$, and each row of an assignment sums to one, we have for every assignment $z\in X$
$$
\sum_{i\in I}\sum_{o\in O}z_{io}U_i(o)
=
\sum_{i\in I}\sum_{o\in O}z_{io}u_i(o)-\sum_{i\in I}\alpha_i.
$$
Therefore, because $x$ maximizes
$\sum_{i\in I}\sum_{o\in O}z_{io}u_i(o)$, it also maximizes
$\sum_{i\in I}\sum_{o\in O}z_{io}U_i(o)$. In particular,
$$
\sum_{i\in I}\sum_{o\in O}y_{io}U_i(o)
\leq
\sum_{i\in I}\sum_{o\in O}x_{io}U_i(o)
=0,
$$
which contradicts the strict inequality above. Hence, no assignment gives the
worst-off agent expected utility strictly greater than zero under $U$. Since
$x$ gives every agent expected utility zero, we conclude that
$$
x\in \argmax_{z\in X}
\left\{
\min_{i\in I}\sum_{o\in O}z_{io}U_i(o)
\right\}.
$$
\end{proof}

What about the converse? Fix a profile of ordinal preferences, and suppose that an assignment maximizes the expected utility of the worst-off agent for some cardinal representation consistent with those preferences. Must this assignment be Rawlsian? The answer is no. The cardinal maximin objective depends on the particular cardinal representation chosen, while the Rawlsian assignment depends only on ordinal preferences. Thus, changing the utility assigned to intermediate ranks can change the assignment that maximizes the expected utility of the worst-off agent, even though the ordinal preferences remain fixed. Therefore, the converse should not be interpreted as a Rawlsian-specific statement, it illustrates a general distinction between ordinal Rawlsianism and cardinal welfare criteria: cardinal criteria may select different assignments for different utility representations of the same ordinal preferences.
Furthermore, Proposition \ref{prop:infinite_max_worst_case} also shows that the difficulty is more fundamental: no ordinal rule satisfying equal treatment of equals can select, for every cardinal representation consistent with the ordinal preferences, an assignment that maximizes the expected utility of the worst-off agent.


\begin{proposition}
\label{prop:infinite_max_worst_case}
Given a problem $\succeq$, no rule treats equals equally and maximizes the utility of the worst-off agent for all cardinal utilities consistent with ordinal preferences~$\succeq$.
\end{proposition}


\begin{proof}
Consider the following preference profile:
\[
\begin{array}{ccc}
\succ_1 & \succ_2 & \succ_3 \\
\hline
a & a & b \\
b & b & a \\
c & c & c
\end{array}
\]

Consider a cardinal representation in which the object ranked in position $k$ has the same utility for every agent. Normalize utilities so that
the utility of the first-ranked object is $u(1)=1$, the utility of the
third-ranked object is $u(3)=0$, and the utility of the second-ranked object is
$u(2)=u\in(0,1)$. Thus,
\[
U_1(a)=U_2(a)=U_3(b)=1,
\]
\[
U_1(b)=U_2(b)=U_3(a)=u,
\]
and
\[
U_1(c)=U_2(c)=U_3(c)=0.
\]

We first characterize the assignments that can maximize the expected utility of the worst-off agent subject to equal treatment of equals. Since agents 1 and 2 must receive the same lottery, any such assignment gives them the same probability of each object. Moreover, agents 1 and 2 must each receive object $a$ with probability $1/2$, while agent 3 receives object $a$ with probability zero. Indeed, if agents 1 and 2 received object $a$ with probability strictly less than $1/2$, then they would receive object $b$ with positive probability. We could then transfer a small amount of probability of object $b$ from agents 1 and 2 to agent 3, and transfer the same amount of probability of object $a$ from agent 3 to agents 1 and 2. This would
strictly increase the expected utility of all agents. Hence such an assignment cannot maximize the expected utility of the worst-off agent.

Therefore, any candidate assignment satisfying equal treatment of equals and maximizing the expected utility of the worst-off agent must be of the form
\[
x^\alpha=
\begin{pmatrix}
\frac12 & \frac12-\alpha & \alpha \\
\frac12 & \frac12-\alpha & \alpha \\
0 & 2\alpha & 1-2\alpha
\end{pmatrix},
\qquad \alpha\in\left[0,\frac12\right].
\]
For this assignment, agents 1 and 2 obtain expected utility
\[
U_1(x^\alpha_1)=U_2(x^\alpha_2)
=
\frac12+\left(\frac12-\alpha\right)u,
\]
while agent 3 obtains expected utility
\[
U_3(x^\alpha_3)=2\alpha.
\]
The first expression is decreasing in $\alpha$, while the second is increasing in $\alpha$. Hence the value of $\alpha$ that maximizes the expected utility of the worst-off agent is the value $\alpha^*$ at which these two expressions
are equal:
\[
\frac12+\left(\frac12-\alpha^*\right)u=2\alpha^*.
\]
Solving for \(\alpha^*\), we obtain
\[
\alpha^*=\frac{1+u}{4+2u}.
\]
Thus, for each cardinal representation with $u\in(0,1)$, the assignment that maximizes the expected utility of the worst-off agent subject to equal treatment of equals is $x^{\alpha^*}$. In particular, as $u$ varies in $(0,1)$, $\alpha^*$ varies in $(1/4,1/3)$. On the other hand, the ordinal Rawlsian assignment for this preference profile is $x^{1/3}$. Therefore, except in the limiting case in which $u=1$, the assignment that
maximizes the expected utility of the worst-off agent for the cardinal representation above is not the Rawlsian assignment.
\end{proof}

\section{Relation with weak sd-strategyproofness}
\label{weak_sdSP}
The following example shows that the Rawlsian rule is not weakly sd-strategyproof. 


\begin{example}\label{ex_sp}
Let $I=\{1,2,3\}$, $O=\{a,b,c\}$, and

\[
\begin{array}{c|ccc}
\succeq_1 & a & b & c \\ 
\succeq_2 & b & c & a \\ 
\succeq_3 & b & c & a \\ 
\end{array}
\]


The Rawlsian assignment of the problem $(\succeq_1,\succeq_2,\succeq_3)$ is:

\begin{equation*}
x = 
\begin{pmatrix}
1 & 0 & 0 \\
0 & \frac{1}{2}  & \frac{1}{2}  \\
0 &  \frac{1}{2} & \frac{1}{2}
\end{pmatrix}.
\end{equation*}

Suppose agent 2 reports the preferences $\succeq'_2=(b,a,c)$. Then, the Rawlsian assignment of the problem $(\succeq_1,\succeq'_2,\succeq_3)$ is:

\begin{equation*}
y = 
\begin{pmatrix}
1 & 0 & 0 \\
0 & 1  & 0  \\
0 &  0 & 1
\end{pmatrix}.
\end{equation*}

Agent 2 prefers (according to $\succeq_2$) her allocation under $y$ to her allocation under $x$. \hfill$\qed$

\end{example}
\section{Descriptive Statistics}\label{charac}
We first illustrate some basic characteristics of the cooperatives we analyze. Table \ref{statistics} shows the number of families in each cooperative (which equals the number of apartments), and the (cumulative) number of different objects ranked in the top $k$ positions ($k=1,2,3,4$) in absolute terms, and as percentage of the total number of objects.

\begin{table}[htbp]
     \centering
     \begin{tabular}{|c|r|rrrr|rrrr|}
     \hline
      & Size & 1st & 2nd & 3rd & 4th & 1st& 2nd & 3rd & 4th \\ \hline
     $C_1$ & 26 & 16 & 20 & 21 & 22 & 61 & 76 & 80 & 84 \\ 
     $C_2$ & 18 & 9 & 11 & 13 & 14 & 50 & 61 & 72 & 77 \\
     $C_3$& 4 & 3 & 4 & 4 & 4 & 75 & 100 & 100 & 100 \\
     $C_4$ & 4 & 2 & 2 & 4 & 4 & 50 & 50 & 100 & 100 \\
     $C_5$ & 28 & 17 & 23 & 24 & 25 & 60 & 82 & 85 & 89 \\
     $C_6$ & 8 & 5 & 8 & 8 & 8 & 62 & 100 & 100 & 100 \\
     $C_7$ & 29 & 16 & 22 & 24 & 25 & 55 & 75 & 82 & 86 \\
     $C_8$ & 12 & 6 & 9 & 11 & 11 & 50 & 75 & 91 & 91 \\
     $C_9$ & 15 & 9 & 11 & 13 & 15 & 60 & 73 & 86 & 100 \\
     $C_{10}$ & 4 & 3 & 3 & 3 & 4 & 75 & 75 & 75 & 100 \\
     $C_{11}$ & 11 & 6 & 7 & 10 & 10 & 54 & 63 & 90 & 90 \\
     $C_{12}$ & 16 & 11 & 14 & 15 & 15 & 68 & 87 & 93 & 93 \\
     $C_{13}$ & 39 & 13 & 23 & 28 & 31 & 33 & 58 & 71 & 79 \\
     $C_{14}$ & 42 & 19 & 24 & 27 & 28 & 45 & 57 & 64 & 66 \\
     $C_{15}$ & 14 & 5 & 8 & 9 & 10 & 35 & 57 & 64 & 71 \\
     $C_{16}$ & 6 & 3 & 4 & 5 & 5 & 50 & 66 & 83 & 83 \\
     $C_{17}$ & 9 & 6 & 8 & 9 & 9 & 66 & 88 & 100 & 100 \\
     $C_{18}$ & 15 & 10 & 10 & 12 & 12 & 66 & 66 & 80 & 80 \\
     $C_{19}$ & 9 & 4 & 5 & 5 & 5 & 44 & 55 & 55 & 55 \\
     $C_{20}$ & 20 & 8 & 11 & 15 & 17 & 40 & 55 & 75 & 85 \\
     $C_{21}$& 24 & 16 & 21 & 23 & 23 & 66 & 87 & 95 & 95 \\
     $C_{22} $& 7 & 4 & 7 & 7 & 7 & 57 & 100 & 100 & 100 \\
     $C_{23} $& 40 & 17 & 26 & 30 & 34 & 42 & 65 & 75 & 85 \\
     $C_{24}$ & 8 & 3 & 5 & 5 & 5 & 37 & 62 & 62 & 62 \\
     \hline
     \end{tabular}
    \caption{Size is the number of families in each cooperative. The next four columns, 1st, 2nd, 3rd, and 4th, have the cumulative number of different apartments ranked in the top 1, 2 ,3 and 4 positions, respectively. The last four columns express the corresponding count columns as a percentage of the total number of apartments.}
     \label{statistics}
     \end{table}

\newpage

\section{Bootstrap Analysis}\label{OA:bootstrap}

We first conduct a bootstrap exercise based on the observed preferences in each cooperative. For each cooperative with $n$ families, we generate bootstrap preference profiles by drawing $n$ preference lists independently with replacement from the submitted rankings in that cooperative. For each bootstrap profile, we compute the maximum rank under the Rawlsian rule and under PS. This exercise preserves the empirical distribution of submitted preferences within each cooperative, while allowing us to evaluate how sensitive the comparison is to resampling variation in the observed preference profile. The results are reported in Table~\ref{tab:bootstrap_max_rank}. 

Second, we added a controlled simulation exercise that varies the degree of preference correlation. For each cooperative, we first construct a cooperative-specific central ranking of apartments using the average observed rank of each apartment. We then generate artificial preference profiles around this central ranking. Specifically, we compute the average observed rank of apartment $o$ as
\[
\bar r(o)=\frac{1}{n}\sum_{i\in I} r_{io}.
\]
The central ranking, $\rho_c$, orders apartments by increasing values of $\bar r(o)$. 

We then generate artificial preference profiles around this central ranking. For each value of the concentration parameter ($\theta\geq 0$), we draw a complete ranking sequentially. At each rank position, among the set ($S$) of apartments not yet ranked, apartment $o\in S$ is selected with probability
\[
\Pr(o\text{ is selected next}\mid S)=
\frac{\exp(-\theta \rho_c(o))}
{\sum_{o'\in S}\exp(-\theta \rho_c(o'))},
\]
where $\rho_c(o)$ is the position of apartment $o$ in the central ranking. Hence, apartments that are higher in the central ranking are more likely to be selected earlier.

When $\theta=0$, all remaining apartments are equally likely to be selected at each step, so the induced distribution over complete rankings is uniform. As $\theta$ increases, agents are increasingly likely to rank highly the apartments that are high in the central ranking. Therefore, simulated preferences become more similar to each other. In the limit, as $\theta$ becomes large, all agents’ rankings converge to the same central ranking. For each value of $\theta$, we generate simulated preference profiles and compute the maximum rank under the Rawlsian rule and under probabilistic serial. In Figure \ref{fig:theta} we illustrate the selection probabilities for different values of $\theta$.

\begin{table}[htbp]
     \begin{center}
     \caption{Size and maximum rank in bootstrap samples of each cooperative for the Rawlsian and PS assignment.}
     \label{tab:bootstrap_max_rank}
     \begin{tabular}{lcccc}
     \hline
     Coop. & $n$ & Max.\ Rawls & Max.\ Rawls (\%) & Max.\ PS \\ 
     \hline
     $C_1$ & 26 & 14.76 & 57 & 25.99 \\
     $C_2$ & 18 & 12.72 & 71 & 18.00 \\
     $C_3$ & 4 & 2.63 & 66 & 3.26 \\
     $C_4$ & 4 & 3.13 & 78 & 3.64 \\
     $C_5$ & 28 & 10.66 & 38 & 27.64 \\
     $C_6$ & 8 & 4.41 & 55 & 7.56 \\
     $C_7$ & 29 & 10.06 & 35 & 28.57 \\
     $C_8$ & 12 & 8.39 & 70 & 11.98 \\
     $C_9$ & 15 & 9.37 & 62 & 14.30 \\
     $C_{10}$ & 4 & 4.00 & 100 & 4.00 \\
     $C_{11}$ & 11 & 5.58 & 51 & 10.84 \\
     $C_{12}$ & 16 & 7.28 & 46 & 15.28 \\
     $C_{13}$ & 39 & 22.25 & 57 & 39.00 \\
     $C_{14}$ & 42 & 33.83 & 81 & 42.00 \\
     $C_{15}$ & 14 & 9.96 & 71 & 14.00 \\
     $C_{16}$ & 6 & 6.00 & 100 & 6.00 \\
     $C_{17}$ & 9 & 4.31 & 48 & 8.62 \\
     $C_{18}$ & 15 & 8.30 & 55 & 14.62 \\
     $C_{19}$ & 9 & 9.00 & 100 & 9.00 \\
     $C_{20}$ & 20 & 11.16 & 56 & 19.98 \\
     $C_{21}$ & 24 & 10.05 & 42 & 23.63 \\
     $C_{22}$ & 7 & 3.69 & 53 & 6.45 \\
     $C_{23}$ & 40 & 13.85 & 35 & 39.86 \\
     $C_{24}$ & 8 & 7.10 & 89 & 7.95 \\
     \hline
     \end{tabular}
     \end{center}
   \footnotesize{\textit{Notes}:} 
\footnotesize{Coop.\ stands for cooperative, each denoted as $C_i$ for $i=1,\ldots,24$. The table reports averages across bootstrap samples. For each cooperative, we generate 500 bootstrap samples by drawing, with replacement, the same number of preference lists as in the original cooperative from the observed submitted rankings. For Max.\ Rawls (Max.\ PS), we compute the rank of the least preferred object assigned with positive probability by the Rawlsian (PS) rule for each bootstrap sample, and then take the average across bootstrap samples. Max.\ Rawls (\%) expresses Max.\ Rawls as a percentage of the length of families' preferences, or equivalently, the size of the cooperative.}
\end{table}

\begin{figure}[htbp]
\centering
\begin{tikzpicture}
\begin{axis}[
    width=0.85\textwidth,
    height=0.48\textwidth,
    xlabel={Central rank $\rho_c(o)$},
    ylabel={Selection probability},
    title={Selection probabilities for different values of $\theta$},
    xmin=1, xmax=10,
    ymin=0, ymax=0.95,
    xtick={1,2,3,4,5,6,7,8,9,10},
    legend pos=north east,
    grid=both,
    grid style={gray!20},
]

\addplot[
    mark=*,
    thick
]
coordinates {
    (1,0.1000)
    (2,0.1000)
    (3,0.1000)
    (4,0.1000)
    (5,0.1000)
    (6,0.1000)
    (7,0.1000)
    (8,0.1000)
    (9,0.1000)
    (10,0.1000)
};
\addlegendentry{$\theta=0$}

\addplot[
    mark=square*,
    thick
]
coordinates {
    (1,0.2096)
    (2,0.1716)
    (3,0.1405)
    (4,0.1151)
    (5,0.0942)
    (6,0.0771)
    (7,0.0631)
    (8,0.0517)
    (9,0.0423)
    (10,0.0347)
};
\addlegendentry{$\theta=0.2$}

\addplot[
    mark=triangle*,
    thick
]
coordinates {
    (1,0.3961)
    (2,0.2403)
    (3,0.1457)
    (4,0.0884)
    (5,0.0536)
    (6,0.0325)
    (7,0.0197)
    (8,0.0120)
    (9,0.0073)
    (10,0.0044)
};
\addlegendentry{$\theta=0.5$}

\addplot[
    mark=diamond*,
    thick
]
coordinates {
    (1,0.6321)
    (2,0.2326)
    (3,0.0856)
    (4,0.0315)
    (5,0.0116)
    (6,0.0043)
    (7,0.0016)
    (8,0.0006)
    (9,0.0002)
    (10,0.0001)
};
\addlegendentry{$\theta=1$}

\addplot[
    mark=*,
    thick,
    dashed
]
coordinates {
    (1,0.8647)
    (2,0.1170)
    (3,0.0158)
    (4,0.0021)
    (5,0.0003)
    (6,0.0000)
    (7,0.0000)
    (8,0.0000)
    (9,0.0000)
    (10,0.0000)
};
\addlegendentry{$\theta=2$}

\end{axis}
\end{tikzpicture}
\caption{The figure considers an illustrative remaining set $S$ with ten apartments, ordered according to the central ranking $\rho_c(o)=1,\ldots,10$. For $\theta=0$, all apartments are selected with equal probability. As $\theta$ increases, probability mass shifts toward apartments ranked higher in the central ranking.}
\label{fig:theta}
\end{figure}

\newpage

\subsection{Maximum rank under different levels of preference correlation for each of the cooperatives.}

 The following tables report averages across simulated preference profiles for each cooperative. For each value of $\theta$, artificial preference profiles are generated using a distribution centered at the cooperative-specific central ranking. Rawlsian and PS report the average maximum rank under each rule. The normalized columns divide maximum ranks by the number of agents in the cooperative.

\begin{table}[htbp]
\centering
\footnotesize
\caption{Maximum rank under different levels of preference correlation: $C_{1}$}
\label{tab:correlation_simulation_andamios2d}
\begin{threeparttable}
\setlength{\tabcolsep}{4pt}
\begin{tabular}{lrrrrrr}
\hline
 $\theta$ &  Rawlsian &    PS &  PS $-$ Rawlsian &  Rawlsian$/n$ &  PS$/n$ &  $(PS-R)/n$ \\
\hline
     0.00 &      4.21 & 25.61 &            21.40 &          0.16 &    0.98 &        0.82 \\
     0.05 &      5.13 & 25.83 &            20.70 &          0.20 &    0.99 &        0.80 \\
     0.10 &      7.95 & 25.95 &            18.01 &          0.31 &    1.00 &        0.69 \\
     0.20 &     14.62 & 26.00 &            11.38 &          0.56 &    1.00 &        0.44 \\
     0.40 &     19.62 & 26.00 &             6.38 &          0.75 &    1.00 &        0.25 \\
     0.80 &     22.61 & 26.00 &             3.39 &          0.87 &    1.00 &        0.13 \\
     1.60 &     24.29 & 26.00 &             1.71 &          0.93 &    1.00 &        0.07 \\
     3.20 &     25.34 & 26.00 &             0.66 &          0.97 &    1.00 &        0.03 \\
\hline
\end{tabular}

\end{threeparttable}
\end{table}

\begin{table}[htbp]
\centering
\footnotesize
\caption{Maximum rank under different levels of preference correlation: $C_{2}$}
\label{tab:correlation_simulation_andamios3d}
\begin{threeparttable}
\setlength{\tabcolsep}{4pt}
\begin{tabular}{lrrrrrr}
\hline
 $\theta$ &  Rawlsian &    PS &  PS $-$ Rawlsian &  Rawlsian$/n$ &  PS$/n$ &  $(PS-R)/n$ \\
\hline
     0.00 &      3.91 & 17.66 &            13.76 &          0.22 &    0.98 &        0.76 \\
     0.05 &      4.31 & 17.82 &            13.51 &          0.24 &    0.99 &        0.75 \\
     0.10 &      5.34 & 17.92 &            12.58 &          0.30 &    1.00 &        0.70 \\
     0.20 &      8.45 & 17.99 &             9.53 &          0.47 &    1.00 &        0.53 \\
     0.40 &     12.64 & 18.00 &             5.36 &          0.70 &    1.00 &        0.30 \\
     0.80 &     15.10 & 18.00 &             2.89 &          0.84 &    1.00 &        0.16 \\
     1.60 &     16.56 & 18.00 &             1.44 &          0.92 &    1.00 &        0.08 \\
     3.20 &     17.47 & 18.00 &             0.53 &          0.97 &    1.00 &        0.03 \\
\hline
\end{tabular}
\end{threeparttable}
\end{table}

\begin{table}[htbp]
\centering
\footnotesize
\caption{Maximum rank under different levels of preference correlation: $C_{3}$}
\label{tab:correlation_simulation_andamiosampliables}
\begin{threeparttable}
\setlength{\tabcolsep}{4pt}
\begin{tabular}{lrrrrrr}
\hline
 $\theta$ &  Rawlsian &   PS &  PS $-$ Rawlsian &  Rawlsian$/n$ &  PS$/n$ &  $(PS-R)/n$ \\
\hline
     0.00 &      2.19 & 3.34 &             1.15 &          0.55 &    0.84 &        0.29 \\
     0.05 &      2.23 & 3.37 &             1.14 &          0.56 &    0.84 &        0.28 \\
     0.10 &      2.18 & 3.39 &             1.21 &          0.54 &    0.85 &        0.30 \\
     0.20 &      2.22 & 3.39 &             1.17 &          0.56 &    0.85 &        0.29 \\
     0.40 &      2.41 & 3.49 &             1.09 &          0.60 &    0.87 &        0.27 \\
     0.80 &      2.78 & 3.69 &             0.91 &          0.69 &    0.92 &        0.23 \\
     1.60 &      3.41 & 3.92 &             0.52 &          0.85 &    0.98 &        0.13 \\
     3.20 &      3.86 & 3.99 &             0.13 &          0.96 &    1.00 &        0.03 \\
\hline
\end{tabular}
\end{threeparttable}
\end{table}

\begin{table}[htbp]
\centering
\footnotesize
\caption{Maximum rank under different levels of preference correlation: $C_{4}$}
\label{tab:correlation_simulation_coviagricola1d}
\begin{threeparttable}
\setlength{\tabcolsep}{4pt}
\begin{tabular}{lrrrrrr}
\hline
 $\theta$ &  Rawlsian &   PS &  PS $-$ Rawlsian &  Rawlsian$/n$ &  PS$/n$ &  $(PS-R)/n$ \\
\hline
     0.00 &      2.15 & 3.31 &             1.16 &          0.54 &    0.83 &        0.29 \\
     0.05 &      2.23 & 3.43 &             1.20 &          0.56 &    0.86 &        0.30 \\
     0.10 &      2.18 & 3.38 &             1.21 &          0.54 &    0.85 &        0.30 \\
     0.20 &      2.24 & 3.42 &             1.18 &          0.56 &    0.86 &        0.30 \\
     0.40 &      2.46 & 3.58 &             1.11 &          0.62 &    0.89 &        0.28 \\
     0.80 &      2.84 & 3.75 &             0.91 &          0.71 &    0.94 &        0.23 \\
     1.60 &      3.44 & 3.93 &             0.49 &          0.86 &    0.98 &        0.12 \\
     3.20 &      3.83 & 3.99 &             0.17 &          0.96 &    1.00 &        0.04 \\
\hline
\end{tabular}
\end{threeparttable}
\end{table}

\begin{table}[htbp]
\centering
\footnotesize
\caption{Maximum rank under different levels of preference correlation: $C_{5}$}
\label{tab:correlation_simulation_coviagricola2d}
\begin{threeparttable}
\setlength{\tabcolsep}{4pt}
\begin{tabular}{lrrrrrr}
\hline
 $\theta$ &  Rawlsian &    PS &  PS $-$ Rawlsian &  Rawlsian$/n$ &  PS$/n$ &  $(PS-R)/n$ \\
\hline
     0.00 &      4.30 & 27.59 &            23.29 &          0.15 &    0.99 &        0.83 \\
     0.05 &      5.40 & 27.88 &            22.48 &          0.19 &    1.00 &        0.80 \\
     0.10 &      8.91 & 27.97 &            19.06 &          0.32 &    1.00 &        0.68 \\
     0.20 &     16.15 & 27.99 &            11.85 &          0.58 &    1.00 &        0.42 \\
     0.40 &     21.47 & 28.00 &             6.53 &          0.77 &    1.00 &        0.23 \\
     0.80 &     24.51 & 28.00 &             3.49 &          0.88 &    1.00 &        0.12 \\
     1.60 &     26.27 & 28.00 &             1.73 &          0.94 &    1.00 &        0.06 \\
     3.20 &     27.27 & 28.00 &             0.73 &          0.97 &    1.00 &        0.03 \\
\hline
\end{tabular}
\end{threeparttable}
\end{table}

\begin{table}[htbp]
\centering
\footnotesize
\caption{Maximum rank under different levels of preference correlation: $C_{6}$}
\label{tab:correlation_simulation_coviagricola2dc}
\begin{threeparttable}
\setlength{\tabcolsep}{4pt}
\begin{tabular}{lrrrrrr}
\hline
 $\theta$ &  Rawlsian &   PS &  PS $-$ Rawlsian &  Rawlsian$/n$ &  PS$/n$ &  $(PS-R)/n$ \\
\hline
     0.00 &      3.01 & 7.70 &             4.70 &          0.38 &    0.96 &        0.59 \\
     0.05 &      2.95 & 7.72 &             4.77 &          0.37 &    0.96 &        0.60 \\
     0.10 &      3.14 & 7.69 &             4.55 &          0.39 &    0.96 &        0.57 \\
     0.20 &      3.44 & 7.81 &             4.37 &          0.43 &    0.98 &        0.55 \\
     0.40 &      4.65 & 7.92 &             3.28 &          0.58 &    0.99 &        0.41 \\
     0.80 &      6.01 & 7.98 &             1.97 &          0.75 &    1.00 &        0.25 \\
     1.60 &      7.00 & 7.99 &             0.99 &          0.88 &    1.00 &        0.12 \\
     3.20 &      7.71 & 8.00 &             0.29 &          0.96 &    1.00 &        0.04 \\
\hline
\end{tabular}
\end{threeparttable}
\end{table}

\begin{table}[htbp]
\centering
\footnotesize
\caption{Maximum rank under different levels of preference correlation: $C_{7}$}
\label{tab:correlation_simulation_coviagricola3d}
\begin{threeparttable}
\setlength{\tabcolsep}{4pt}
\begin{tabular}{lrrrrrr}
\hline
 $\theta$ &  Rawlsian &    PS &  PS $-$ Rawlsian &  Rawlsian$/n$ &  PS$/n$ &  $(PS-R)/n$ \\
\hline
     0.00 &      4.37 & 28.66 &            24.29 &          0.15 &    0.99 &        0.84 \\
     0.05 &      5.55 & 28.93 &            23.38 &          0.19 &    1.00 &        0.81 \\
     0.10 &      9.13 & 28.97 &            19.84 &          0.31 &    1.00 &        0.68 \\
     0.20 &     16.85 & 29.00 &            12.15 &          0.58 &    1.00 &        0.42 \\
     0.40 &     22.40 & 29.00 &             6.60 &          0.77 &    1.00 &        0.23 \\
     0.80 &     25.52 & 29.00 &             3.48 &          0.88 &    1.00 &        0.12 \\
     1.60 &     27.21 & 29.00 &             1.79 &          0.94 &    1.00 &        0.06 \\
     3.20 &     28.23 & 29.00 &             0.77 &          0.97 &    1.00 &        0.03 \\
\hline
\end{tabular}
\end{threeparttable}
\end{table}

\begin{table}[htbp]
\centering
\footnotesize
\caption{Maximum rank under different levels of preference correlation: $C_{8}$}
\label{tab:correlation_simulation_coviagricola4d}
\begin{threeparttable}
\setlength{\tabcolsep}{4pt}
\begin{tabular}{lrrrrrr}
\hline
 $\theta$ &  Rawlsian &    PS &  PS $-$ Rawlsian &  Rawlsian$/n$ &  PS$/n$ &  $(PS-R)/n$ \\
\hline
     0.00 &      3.36 & 11.76 &             8.40 &          0.28 &    0.98 &        0.70 \\
     0.05 &      3.48 & 11.73 &             8.25 &          0.29 &    0.98 &        0.69 \\
     0.10 &      3.83 & 11.82 &             7.99 &          0.32 &    0.99 &        0.67 \\
     0.20 &      5.02 & 11.91 &             6.88 &          0.42 &    0.99 &        0.57 \\
     0.40 &      7.51 & 11.99 &             4.48 &          0.63 &    1.00 &        0.37 \\
     0.80 &      9.55 & 11.99 &             2.44 &          0.80 &    1.00 &        0.20 \\
     1.60 &     10.82 & 12.00 &             1.18 &          0.90 &    1.00 &        0.10 \\
     3.20 &     11.66 & 12.00 &             0.34 &          0.97 &    1.00 &        0.03 \\
\hline
\end{tabular}
\end{threeparttable}
\end{table}

\begin{table}[htbp]
\centering
\footnotesize
\caption{Maximum rank under different levels of preference correlation: $C_{9}$}
\label{tab:correlation_simulation_coviedu2d}
\begin{threeparttable}
\setlength{\tabcolsep}{4pt}
\begin{tabular}{lrrrrrr}
\hline
 $\theta$ &  Rawlsian &    PS &  PS $-$ Rawlsian &  Rawlsian$/n$ &  PS$/n$ &  $(PS-R)/n$ \\
\hline
     0.00 &      3.65 & 14.69 &            11.05 &          0.24 &    0.98 &        0.74 \\
     0.05 &      3.86 & 14.75 &            10.89 &          0.26 &    0.98 &        0.73 \\
     0.10 &      4.49 & 14.87 &            10.38 &          0.30 &    0.99 &        0.69 \\
     0.20 &      6.57 & 14.98 &             8.41 &          0.44 &    1.00 &        0.56 \\
     0.40 &     10.07 & 15.00 &             4.93 &          0.67 &    1.00 &        0.33 \\
     0.80 &     12.26 & 15.00 &             2.74 &          0.82 &    1.00 &        0.18 \\
     1.60 &     13.66 & 15.00 &             1.34 &          0.91 &    1.00 &        0.09 \\
     3.20 &     14.51 & 15.00 &             0.49 &          0.97 &    1.00 &        0.03 \\
\hline
\end{tabular}
\end{threeparttable}
\end{table}

\begin{table}[htbp]
\centering
\footnotesize
\caption{Maximum rank under different levels of preference correlation: $C_{10}$}
\label{tab:correlation_simulation_coviedu3d}
\begin{threeparttable}
\setlength{\tabcolsep}{4pt}
\begin{tabular}{lrrrrrr}
\hline
 $\theta$ &  Rawlsian &   PS &  PS $-$ Rawlsian &  Rawlsian$/n$ &  PS$/n$ &  $(PS-R)/n$ \\
\hline
     0.00 &      2.21 & 3.44 &             1.24 &          0.55 &    0.86 &        0.31 \\
     0.05 &      2.19 & 3.39 &             1.20 &          0.55 &    0.85 &        0.30 \\
     0.10 &      2.20 & 3.39 &             1.20 &          0.55 &    0.85 &        0.30 \\
     0.20 &      2.26 & 3.41 &             1.14 &          0.57 &    0.85 &        0.29 \\
     0.40 &      2.37 & 3.50 &             1.13 &          0.59 &    0.88 &        0.28 \\
     0.80 &      2.81 & 3.72 &             0.91 &          0.70 &    0.93 &        0.23 \\
     1.60 &      3.42 & 3.94 &             0.52 &          0.85 &    0.98 &        0.13 \\
     3.20 &      3.84 & 3.99 &             0.15 &          0.96 &    1.00 &        0.04 \\
\hline
\end{tabular}
\end{threeparttable}
\end{table}

\begin{table}[htbp]
\centering
\footnotesize
\caption{Maximum rank under different levels of preference correlation: $C_{11}$}
\label{tab:correlation_simulation_delnavio2d}
\begin{threeparttable}
\setlength{\tabcolsep}{4pt}
\begin{tabular}{lrrrrrr}
\hline
 $\theta$ &  Rawlsian &    PS &  PS $-$ Rawlsian &  Rawlsian$/n$ &  PS$/n$ &  $(PS-R)/n$ \\
\hline
     0.00 &      3.36 & 10.72 &             7.36 &          0.31 &    0.97 &        0.67 \\
     0.05 &      3.39 & 10.72 &             7.33 &          0.31 &    0.97 &        0.67 \\
     0.10 &      3.70 & 10.82 &             7.13 &          0.34 &    0.98 &        0.65 \\
     0.20 &      4.48 & 10.90 &             6.41 &          0.41 &    0.99 &        0.58 \\
     0.40 &      6.82 & 10.99 &             4.17 &          0.62 &    1.00 &        0.38 \\
     0.80 &      8.63 & 11.00 &             2.37 &          0.78 &    1.00 &        0.22 \\
     1.60 &      9.84 & 11.00 &             1.16 &          0.89 &    1.00 &        0.11 \\
     3.20 &     10.63 & 11.00 &             0.37 &          0.97 &    1.00 &        0.03 \\
\hline
\end{tabular}
\end{threeparttable}
\end{table}

\begin{table}[htbp]
\centering
\footnotesize
\caption{Maximum rank under different levels of preference correlation: $C_{12}$}
\label{tab:correlation_simulation_delnavio3d}
\begin{threeparttable}
\setlength{\tabcolsep}{4pt}
\begin{tabular}{lrrrrrr}
\hline
 $\theta$ &  Rawlsian &    PS &  PS $-$ Rawlsian &  Rawlsian$/n$ &  PS$/n$ &  $(PS-R)/n$ \\
\hline
     0.00 &      3.71 & 15.67 &            11.95 &          0.23 &    0.98 &        0.75 \\
     0.05 &      4.02 & 15.76 &            11.74 &          0.25 &    0.98 &        0.73 \\
     0.10 &      4.72 & 15.91 &            11.20 &          0.29 &    0.99 &        0.70 \\
     0.20 &      7.31 & 15.96 &             8.66 &          0.46 &    1.00 &        0.54 \\
     0.40 &     10.92 & 15.99 &             5.07 &          0.68 &    1.00 &        0.32 \\
     0.80 &     13.23 & 16.00 &             2.77 &          0.83 &    1.00 &        0.17 \\
     1.60 &     14.62 & 16.00 &             1.38 &          0.91 &    1.00 &        0.09 \\
     3.20 &     15.52 & 16.00 &             0.48 &          0.97 &    1.00 &        0.03 \\
\hline
\end{tabular}
\end{threeparttable}
\end{table}

\begin{table}[htbp]
\centering
\footnotesize
\caption{Maximum rank under different levels of preference correlation: $C_{13}$}
\label{tab:correlation_simulation_gardelianasur}
\begin{threeparttable}
\setlength{\tabcolsep}{4pt}
\begin{tabular}{lrrrrrr}
\hline
 $\theta$ &  Rawlsian &    PS &  PS $-$ Rawlsian &  Rawlsian$/n$ &  PS$/n$ &  $(PS-R)/n$ \\
\hline
     0.00 &      4.64 & 38.58 &            33.94 &          0.12 &    0.99 &        0.87 \\
     0.05 &      7.22 & 38.97 &            31.75 &          0.19 &    1.00 &        0.81 \\
     0.10 &     14.84 & 39.00 &            24.16 &          0.38 &    1.00 &        0.62 \\
     0.20 &     25.51 & 39.00 &            13.48 &          0.65 &    1.00 &        0.35 \\
     0.40 &     31.67 & 39.00 &             7.33 &          0.81 &    1.00 &        0.19 \\
     0.80 &     35.08 & 39.00 &             3.92 &          0.90 &    1.00 &        0.10 \\
     1.60 &     37.17 & 39.00 &             1.83 &          0.95 &    1.00 &        0.05 \\
     3.20 &     38.20 & 39.00 &             0.80 &          0.98 &    1.00 &        0.02 \\
\hline
\end{tabular}
\end{threeparttable}
\end{table}

\begin{table}[htbp]
\centering
\footnotesize
\caption{Maximum rank under different levels of preference correlation: $C_{14}$}
\label{tab:correlation_simulation_maleconmaua}
\begin{threeparttable}
\setlength{\tabcolsep}{4pt}
\begin{tabular}{lrrrrrr}
\hline
 $\theta$ &  Rawlsian &    PS &  PS $-$ Rawlsian &  Rawlsian$/n$ &  PS$/n$ &  $(PS-R)/n$ \\
\hline
     0.00 &      4.68 & 41.61 &            36.94 &          0.11 &    0.99 &        0.88 \\
     0.05 &      8.11 & 41.96 &            33.86 &          0.19 &    1.00 &        0.81 \\
     0.10 &     17.09 & 41.99 &            24.90 &          0.41 &    1.00 &        0.59 \\
     0.20 &     28.37 & 42.00 &            13.63 &          0.68 &    1.00 &        0.32 \\
     0.40 &     34.60 & 42.00 &             7.40 &          0.82 &    1.00 &        0.18 \\
     0.80 &     38.01 & 42.00 &             3.99 &          0.91 &    1.00 &        0.09 \\
     1.60 &     40.03 & 42.00 &             1.97 &          0.95 &    1.00 &        0.05 \\
     3.20 &     41.14 & 42.00 &             0.86 &          0.98 &    1.00 &        0.02 \\
\hline
\end{tabular}
\end{threeparttable}
\end{table}

\begin{table}[htbp]
\centering
\footnotesize
\caption{Maximum rank under different levels of preference correlation: $C_{15}$}
\label{tab:correlation_simulation_nuevaera2d}
\begin{threeparttable}
\setlength{\tabcolsep}{4pt}
\begin{tabular}{lrrrrrr}
\hline
 $\theta$ &  Rawlsian &    PS &  PS $-$ Rawlsian &  Rawlsian$/n$ &  PS$/n$ &  $(PS-R)/n$ \\
\hline
     0.00 &      3.61 & 13.74 &            10.13 &          0.26 &    0.98 &        0.72 \\
     0.05 &      3.86 & 13.79 &             9.93 &          0.28 &    0.98 &        0.71 \\
     0.10 &      4.26 & 13.90 &             9.64 &          0.30 &    0.99 &        0.69 \\
     0.20 &      6.08 & 13.97 &             7.89 &          0.43 &    1.00 &        0.56 \\
     0.40 &      9.14 & 13.99 &             4.85 &          0.65 &    1.00 &        0.35 \\
     0.80 &     11.38 & 14.00 &             2.62 &          0.81 &    1.00 &        0.19 \\
     1.60 &     12.70 & 14.00 &             1.30 &          0.91 &    1.00 &        0.09 \\
     3.20 &     13.56 & 14.00 &             0.44 &          0.97 &    1.00 &        0.03 \\
\hline
\end{tabular}
\end{threeparttable}
\end{table}

\begin{table}[htbp]
\centering
\footnotesize
\caption{Maximum rank under different levels of preference correlation: $C_{16}$}
\label{tab:correlation_simulation_nuevaera3d}
\begin{threeparttable}
\setlength{\tabcolsep}{4pt}
\begin{tabular}{lrrrrrr}
\hline
 $\theta$ &  Rawlsian &   PS &  PS $-$ Rawlsian &  Rawlsian$/n$ &  PS$/n$ &  $(PS-R)/n$ \\
\hline
     0.00 &      2.58 & 5.48 &             2.90 &          0.43 &    0.91 &        0.48 \\
     0.05 &      2.65 & 5.60 &             2.95 &          0.44 &    0.93 &        0.49 \\
     0.10 &      2.61 & 5.57 &             2.95 &          0.44 &    0.93 &        0.49 \\
     0.20 &      2.83 & 5.58 &             2.75 &          0.47 &    0.93 &        0.46 \\
     0.40 &      3.46 & 5.79 &             2.33 &          0.58 &    0.96 &        0.39 \\
     0.80 &      4.35 & 5.91 &             1.56 &          0.72 &    0.98 &        0.26 \\
     1.60 &      5.23 & 5.98 &             0.74 &          0.87 &    1.00 &        0.12 \\
     3.20 &      5.78 & 6.00 &             0.22 &          0.96 &    1.00 &        0.04 \\
\hline
\end{tabular}
\end{threeparttable}
\end{table}

\begin{table}[htbp]
\centering
\footnotesize
\caption{Maximum rank under different levels of preference correlation: $C_{17}$}
\label{tab:correlation_simulation_puebla2d}
\begin{threeparttable}
\setlength{\tabcolsep}{4pt}
\begin{tabular}{lrrrrrr}
\hline
 $\theta$ &  Rawlsian &   PS &  PS $-$ Rawlsian &  Rawlsian$/n$ &  PS$/n$ &  $(PS-R)/n$ \\
\hline
     0.00 &      3.09 & 8.70 &             5.62 &          0.34 &    0.97 &        0.62 \\
     0.05 &      3.15 & 8.71 &             5.56 &          0.35 &    0.97 &        0.62 \\
     0.10 &      3.30 & 8.78 &             5.48 &          0.37 &    0.98 &        0.61 \\
     0.20 &      3.85 & 8.89 &             5.05 &          0.43 &    0.99 &        0.56 \\
     0.40 &      5.31 & 8.97 &             3.66 &          0.59 &    1.00 &        0.41 \\
     0.80 &      6.84 & 8.99 &             2.16 &          0.76 &    1.00 &        0.24 \\
     1.60 &      7.93 & 9.00 &             1.07 &          0.88 &    1.00 &        0.12 \\
     3.20 &      8.64 & 9.00 &             0.36 &          0.96 &    1.00 &        0.04 \\
\hline
\end{tabular}
\end{threeparttable}
\end{table}

\begin{table}[htbp]
\centering
\footnotesize
\caption{Maximum rank under different levels of preference correlation: $C_{18}$}
\label{tab:correlation_simulation_puebla3d}
\begin{threeparttable}
\setlength{\tabcolsep}{4pt}
\begin{tabular}{lrrrrrr}
\hline
 $\theta$ &  Rawlsian &    PS &  PS $-$ Rawlsian &  Rawlsian$/n$ &  PS$/n$ &  $(PS-R)/n$ \\
\hline
     0.00 &      3.63 & 14.67 &            11.05 &          0.24 &    0.98 &        0.74 \\
     0.05 &      3.91 & 14.75 &            10.84 &          0.26 &    0.98 &        0.72 \\
     0.10 &      4.44 & 14.87 &            10.43 &          0.30 &    0.99 &        0.70 \\
     0.20 &      6.60 & 14.98 &             8.38 &          0.44 &    1.00 &        0.56 \\
     0.40 &     10.00 & 14.99 &             4.99 &          0.67 &    1.00 &        0.33 \\
     0.80 &     12.28 & 15.00 &             2.72 &          0.82 &    1.00 &        0.18 \\
     1.60 &     13.61 & 15.00 &             1.39 &          0.91 &    1.00 &        0.09 \\
     3.20 &     14.54 & 15.00 &             0.46 &          0.97 &    1.00 &        0.03 \\
\hline
\end{tabular}
\end{threeparttable}
\end{table}

\begin{table}[htbp]
\centering
\footnotesize
\caption{Maximum rank under different levels of preference correlation: $C_{19}$}
\label{tab:correlation_simulation_puertofabini3d}
\begin{threeparttable}
\setlength{\tabcolsep}{4pt}
\begin{tabular}{lrrrrrr}
\hline
 $\theta$ &  Rawlsian &   PS &  PS $-$ Rawlsian &  Rawlsian$/n$ &  PS$/n$ &  $(PS-R)/n$ \\
\hline
     0.00 &      3.10 & 8.72 &             5.62 &          0.34 &    0.97 &        0.62 \\
     0.05 &      3.10 & 8.66 &             5.56 &          0.34 &    0.96 &        0.62 \\
     0.10 &      3.23 & 8.75 &             5.52 &          0.36 &    0.97 &        0.61 \\
     0.20 &      3.86 & 8.83 &             4.98 &          0.43 &    0.98 &        0.55 \\
     0.40 &      5.28 & 8.96 &             3.68 &          0.59 &    1.00 &        0.41 \\
     0.80 &      6.90 & 9.00 &             2.10 &          0.77 &    1.00 &        0.23 \\
     1.60 &      8.00 & 9.00 &             0.99 &          0.89 &    1.00 &        0.11 \\
     3.20 &      8.66 & 9.00 &             0.34 &          0.96 &    1.00 &        0.04 \\
\hline
\end{tabular}
\end{threeparttable}
\end{table}

\begin{table}[htbp]
\centering
\footnotesize
\caption{Maximum rank under different levels of preference correlation: $C_{20}$}
\label{tab:correlation_simulation_santamonica}
\begin{threeparttable}
\setlength{\tabcolsep}{4pt}
\begin{tabular}{lrrrrrr}
\hline
 $\theta$ &  Rawlsian &    PS &  PS $-$ Rawlsian &  Rawlsian$/n$ &  PS$/n$ &  $(PS-R)/n$ \\
\hline
     0.00 &      4.03 & 19.65 &            15.62 &          0.20 &    0.98 &        0.78 \\
     0.05 &      4.42 & 19.83 &            15.42 &          0.22 &    0.99 &        0.77 \\
     0.10 &      5.86 & 19.95 &            14.08 &          0.29 &    1.00 &        0.70 \\
     0.20 &     10.02 & 19.99 &             9.97 &          0.50 &    1.00 &        0.50 \\
     0.40 &     14.30 & 20.00 &             5.70 &          0.72 &    1.00 &        0.28 \\
     0.80 &     16.89 & 20.00 &             3.11 &          0.84 &    1.00 &        0.16 \\
     1.60 &     18.47 & 20.00 &             1.53 &          0.92 &    1.00 &        0.08 \\
     3.20 &     19.39 & 20.00 &             0.61 &          0.97 &    1.00 &        0.03 \\
\hline
\end{tabular}
\end{threeparttable}
\end{table}

\begin{table}[htbp]
\centering
\footnotesize
\caption{Maximum rank under different levels of preference correlation: $C_{21}$}
\label{tab:correlation_simulation_tataypy2d}
\begin{threeparttable}
\setlength{\tabcolsep}{4pt}
\begin{tabular}{lrrrrrr}
\hline
 $\theta$ &  Rawlsian &    PS &  PS $-$ Rawlsian &  Rawlsian$/n$ &  PS$/n$ &  $(PS-R)/n$ \\
\hline
     0.00 &      4.18 & 23.60 &            19.42 &          0.17 &    0.98 &        0.81 \\
     0.05 &      4.97 & 23.85 &            18.88 &          0.21 &    0.99 &        0.79 \\
     0.10 &      7.22 & 23.95 &            16.74 &          0.30 &    1.00 &        0.70 \\
     0.20 &     12.84 & 23.99 &            11.15 &          0.53 &    1.00 &        0.46 \\
     0.40 &     17.92 & 24.00 &             6.08 &          0.75 &    1.00 &        0.25 \\
     0.80 &     20.77 & 24.00 &             3.23 &          0.87 &    1.00 &        0.13 \\
     1.60 &     22.38 & 24.00 &             1.62 &          0.93 &    1.00 &        0.07 \\
     3.20 &     23.40 & 24.00 &             0.60 &          0.97 &    1.00 &        0.03 \\
\hline
\end{tabular}
\end{threeparttable}
\end{table}

\begin{table}[htbp]
\centering
\footnotesize
\caption{Maximum rank under different levels of preference correlation: $C_{22}$}
\label{tab:correlation_simulation_tataypy3d}
\begin{threeparttable}
\setlength{\tabcolsep}{4pt}
\begin{tabular}{lrrrrrr}
\hline
 $\theta$ &  Rawlsian &   PS &  PS $-$ Rawlsian &  Rawlsian$/n$ &  PS$/n$ &  $(PS-R)/n$ \\
\hline
     0.00 &      2.79 & 6.59 &             3.80 &          0.40 &    0.94 &        0.54 \\
     0.05 &      2.77 & 6.61 &             3.83 &          0.40 &    0.94 &        0.55 \\
     0.10 &      2.87 & 6.63 &             3.77 &          0.41 &    0.95 &        0.54 \\
     0.20 &      3.19 & 6.71 &             3.51 &          0.46 &    0.96 &        0.50 \\
     0.40 &      3.85 & 6.89 &             3.04 &          0.55 &    0.98 &        0.43 \\
     0.80 &      5.26 & 6.99 &             1.73 &          0.75 &    1.00 &        0.25 \\
     1.60 &      6.09 & 7.00 &             0.91 &          0.87 &    1.00 &        0.13 \\
     3.20 &      6.76 & 7.00 &             0.24 &          0.97 &    1.00 &        0.03 \\
\hline
\end{tabular}
\end{threeparttable}
\end{table}

\begin{table}[htbp]
\centering
\footnotesize
\caption{Maximum rank under different levels of preference correlation: $C_{23}$}
\label{tab:correlation_simulation_virazon2d}
\begin{threeparttable}
\setlength{\tabcolsep}{4pt}
\begin{tabular}{lrrrrrr}
\hline
 $\theta$ &  Rawlsian &    PS &  PS $-$ Rawlsian &  Rawlsian$/n$ &  PS$/n$ &  $(PS-R)/n$ \\
\hline
     0.00 &      4.74 & 39.56 &            34.82 &          0.12 &    0.99 &        0.87 \\
     0.05 &      7.48 & 39.96 &            32.48 &          0.19 &    1.00 &        0.81 \\
     0.10 &     15.63 & 39.99 &            24.36 &          0.39 &    1.00 &        0.61 \\
     0.20 &     26.26 & 40.00 &            13.73 &          0.66 &    1.00 &        0.34 \\
     0.40 &     32.65 & 40.00 &             7.35 &          0.82 &    1.00 &        0.18 \\
     0.80 &     36.12 & 40.00 &             3.88 &          0.90 &    1.00 &        0.10 \\
     1.60 &     38.06 & 40.00 &             1.94 &          0.95 &    1.00 &        0.05 \\
     3.20 &     39.14 & 40.00 &             0.86 &          0.98 &    1.00 &        0.02 \\
\hline
\end{tabular}
\end{threeparttable}
\end{table}

\begin{table}[htbp]
\centering
\footnotesize
\caption{Maximum rank under different levels of preference correlation: $C_{24}$}
\label{tab:correlation_simulation_virazon3d}
\begin{threeparttable}
\setlength{\tabcolsep}{4pt}
\begin{tabular}{lrrrrrr}
\hline
 $\theta$ &  Rawlsian &   PS &  PS $-$ Rawlsian &  Rawlsian$/n$ &  PS$/n$ &  $(PS-R)/n$ \\
\hline
     0.00 &      2.98 & 7.64 &             4.66 &          0.37 &    0.95 &        0.58 \\
     0.05 &      2.96 & 7.69 &             4.74 &          0.37 &    0.96 &        0.59 \\
     0.10 &      3.06 & 7.71 &             4.65 &          0.38 &    0.96 &        0.58 \\
     0.20 &      3.49 & 7.77 &             4.28 &          0.44 &    0.97 &        0.54 \\
     0.40 &      4.62 & 7.93 &             3.31 &          0.58 &    0.99 &        0.41 \\
     0.80 &      6.05 & 7.98 &             1.93 &          0.76 &    1.00 &        0.24 \\
     1.60 &      7.05 & 7.99 &             0.94 &          0.88 &    1.00 &        0.12 \\
     3.20 &      7.72 & 8.00 &             0.28 &          0.97 &    1.00 &        0.03 \\
\hline
\end{tabular}
\end{threeparttable}
\end{table}

\newpage

\section{Expected number of families that are assigned apartments with rank \texorpdfstring{$k=1,\ldots,n$}{k=1,...,n} by the Rawlsian, PS and MTAV rules}\label{dist_agents}

\begin{table}[htbp]
\small
     \centering
     \caption{Expected number of families that are assigned apartments with rank up to $k=1,\ldots,16$ by the Rawlsian, PS and MTAV rules.}
     \begin{tabular}{|c|ccc|ccc|ccc|ccc|}
          \hline
        & \multicolumn{3}{c|}{$C_{1}$} & \multicolumn{3}{c|}{$C_{2}$} & \multicolumn{3}{c|}{$C_{3}$} & \multicolumn{3}{c|}{$C_{4}$}  \\ \hline
     Position & Rawls & PS & MTAV & Rawls & PS & MTAV & Rawls & PS & MTAV & Rawls & PS & MTAV \\  \hline
     1 & 7.17 & 12.54 & 10 & 2.4 & 6.42 & 8 & 2 & 2.44 & 2 & 2 & 2 & 2 \\2 & 10.36 & 14.61 & 15 & 4.4 & 7.13 & 10 & 4 & 3.33 & 4 & 2 & 2 & 2 \\3 & 13.33 & 15.64 & 16 & 5 & 7.64 & 10 & 4 & 3.56 & 4 & 4 & 4 & 4 \\4 & 16.67 & 16.51 & 18 & 5 & 8.34 & 10 & 4 & 4 & 4 & 4 & 4 & 4 \\5 & 17 & 16.99 & 18 & 8 & 9.42 & 11 & - & - & - & - & - & - \\6 & 18 & 17.73 & 19 & 11 & 10.55 & 13 & - & - & - & - & - & - \\7 & 20 & 18.69 & 20 & 15 & 11.91 & 13 & - & - & - & - & - & - \\8 & 22 & 19.62 & 23 & 15 & 12.7 & 13 & - & - & - & - & - & - \\9 & 23 & 19.78 & 23 & 17 & 13.46 & 14 & - & - & - & - & - & - \\10 & 23 & 20.01 & 23 & 17 & 14.13 & 14 & - & - & - & - & - & - \\11 & 24 & 20.15 & 23 & 17 & 14.66 & 16 & - & - & - & - & - & - \\12 & 25 & 20.45 & 24 & 18 & 15.04 & 18 & - & - & - & - & - & - \\13 & 26 & 21.49 & 26 & 18 & 15.34 & 18 & - & - & - & - & - & - \\14 & 26 & 21.91 & 26 & 18 & 15.98 & 18 & - & - & - & - & - & - \\15 & 26 & 22.51 & 26 & 18 & 16.32 & 18 & - & - & - & - & - & - \\16 & 26 & 22.76 & 26 & 18 & 17.13 & 18 & - & - & - & - & - & - \\Total & 26 & 26 & 26 & 18 & 18 & 18 & - & - & - & - & - & -      \\
     \hline 
     \end{tabular}
     \label{dist1}
     \end{table}

\begin{table}[htbp]
    \small
     \centering
      \caption{Expected number of families that are assigned apartments with rank up to  $k=1,\ldots,16$ by the Rawlsian, PS and MTAV rules.}
     \begin{tabular}{|c|ccc|ccc|ccc|ccc|}
 \hline
        & \multicolumn{3}{c|}{$C_{5}$} & \multicolumn{3}{c|}{$C_{6}$} & \multicolumn{3}{c|}{$C_{7}$} & \multicolumn{3}{c|}{$C_{8}$}  \\ \hline
     Position & Rawls & PS & MTAV & Rawls & PS & MTAV & Rawls & PS & MTAV & Rawls & PS & MTAV \\  \hline
     1 & 7.75 & 13.82 & 11 & 4 & 4.28 & 4 & 11.5 & 11.72 & 13 & 3.5 & 4.51 & 4 \\2 & 13.58 & 18.36 & 17 & 7 & 5.78 & 7 & 14.5 & 14.83 & 16 & 7.67 & 7.52 & 8 \\3 & 17 & 18.97 & 20 & 8 & 6.72 & 8 & 15.5 & 16.63 & 18 & 11 & 9.21 & 11 \\4 & 19 & 20.03 & 21 & 8 & 7.22 & 8 & 18 & 17.94 & 20 & 11 & 9.33 & 11 \\5 & 21.5 & 21.04 & 22 & 8 & 7.56 & 8 & 21.5 & 19.36 & 20 & 11 & 9.66 & 11 \\6 & 24.5 & 21.52 & 23 & 8 & 7.89 & 8 & 23 & 20.09 & 22 & 11 & 10.02 & 11 \\7 & 27 & 22.07 & 26 & 8 & 7.89 & 8 & 27 & 21.42 & 27 & 12 & 10.27 & 12 \\8 & 28 & 22.52 & 28 & 8 & 8 & 8 & 29 & 22.46 & 29 & 12 & 10.76 & 12 \\9 & 28 & 23.05 & 28 & - & - & - & 29 & 23.12 & 29 & 12 & 10.94 & 12 \\10 & 28 & 23.34 & 28 & - & - & - & 29 & 23.95 & 29 & 12 & 11.05 & 12 \\11 & 28 & 24.04 & 28 & - & - & - & 29 & 24.87 & 29 & 12 & 11.25 & 12 \\12 & 28 & 24.55 & 28 & - & - & - & 29 & 25.42 & 29 & 12 & 12 & 12 \\13 & 28 & 25.14 & 28 & - & - & - & 29 & 25.91 & 29 & - & - & - \\14 & 28 & 25.33 & 28 & - & - & - & 29 & 26.2 & 29 & - & - & - \\15 & 28 & 25.93 & 28 & - & - & - & 29 & 27.03 & 29 & - & - & - \\16 & 28 & 26.11 & 28 & - & - & - & 29 & 27.8 & 29 & - & - & - \\Total & 28 & 28 & 28 & - & - & - & 29 & 29 & 29 & - & - & -  \\
     \hline 
     \end{tabular}
     \label{dist2}
     \end{table}

\begin{table}[htbp]
     \small
     \centering
      \caption{Expected number of families that are assigned apartments with rank up to  $k=1,\ldots,16$ by the Rawlsian, PS and MTAV rules.}
     \begin{tabular}{|c|ccc|ccc|ccc|ccc|}
    \hline
        & \multicolumn{3}{c|}{$C_{9}$} & \multicolumn{3}{c|}{$C_{10}$} & \multicolumn{3}{c|}{$C_{11}$} & \multicolumn{3}{c|}{$C_{12}$}  \\ \hline
     Position & Rawls & PS & MTAV & Rawls & PS & MTAV & Rawls & PS & MTAV & Rawls & PS & MTAV \\ \hline
     1 & 4.88 & 7.23 & 8 & 2.42 & 2.42 & 3 & 4 & 5.63 & 6 & 7.81 & 9.05 & 8 \\2 & 7.25 & 8.59 & 9 & 2.75 & 2.75 & 3 & 7 & 6.33 & 7 & 11.38 & 11.73 & 12 \\3 & 9 & 9.4 & 10 & 3 & 3 & 3 & 10 & 8.56 & 9 & 13 & 12.44 & 13 \\4 & 13 & 10.85 & 12 & 4 & 4 & 4 & 10 & 9.13 & 10 & 13.75 & 12.99 & 14 \\5 & 14 & 12.16 & 13 & - & - & - & 11 & 9.59 & 11 & 15 & 13.53 & 15 \\6 & 15 & 13.03 & 15 & - & - & - & 11 & 10.31 & 11 & 16 & 13.88 & 16 \\7 & 15 & 13.29 & 15 & - & - & - & 11 & 10.44 & 11 & 16 & 14.26 & 16 \\8 & 15 & 13.35 & 15 & - & - & - & 11 & 10.6 & 11 & 16 & 14.58 & 16 \\9 & 15 & 13.69 & 15 & - & - & - & 11 & 10.86 & 11 & 16 & 15.03 & 16 \\10 & 15 & 14.47 & 15 & - & - & - & 11 & 10.89 & 11 & 16 & 15.29 & 16 \\11 & 15 & 14.74 & 15 & - & - & - & 11 & 11 & 11 & 16 & 15.4 & 16 \\12 & 15 & 14.8 & 15 & - & - & - & - & - & - & 16 & 15.48 & 16 \\13 & 15 & 14.93 & 15 & - & - & - & - & - & - & 16 & 15.65 & 16 \\14 & 15 & 15 & 15 & - & - & - & - & - & - & 16 & 15.85 & 16 \\15 & 15 & 15 & 15 & - & - & - & - & - & - & 16 & 15.91 & 16 \\16 & - & - & - & - & - & - & - & - & - & 16 & 16 & 16 \\Total & - & - & - & - & - & - & - & - & - & - & - & -  \\
     \hline 
     \end{tabular}
     \label{dist3}
     \end{table}

\begin{table}[htbp]
\small
     \centering
      \caption{Expected number of families that are assigned apartments with rank up to  $k=1,\ldots,16$ by the Rawlsian, PS and MTAV rules.}
     \begin{tabular}{|c|ccc|ccc|ccc|ccc|}
    \hline
        & \multicolumn{3}{c|}{$C_{13}$} & \multicolumn{3}{c|}{$C_{14}$} & \multicolumn{3}{c|}{$C_{15}$} & \multicolumn{3}{c|}{$C_{16}$}  \\ \hline
     Position & Rawls & PS & MTAV & Rawls & PS & MTAV & Rawls & PS & MTAV & Rawls & PS & MTAV \\ \hline
     1 & 4.5 & 9.29 & 9 & 3.11 & 12 & 7 & 1.78 & 3.65 & 4 & 2.17 & 2.44 & 3 \\2 & 11.5 & 16.52 & 14 & 8.46 & 15.38 & 16 & 3.75 & 5.88 & 8 & 3.5 & 3.53 & 4 \\3 & 15.5 & 19.01 & 19 & 9.97 & 17.15 & 19 & 4.58 & 6.55 & 8 & 4.83 & 3.92 & 5 \\4 & 19.5 & 20.97 & 24 & 14.15 & 18.66 & 20 & 7.33 & 7.75 & 8 & 5 & 4.51 & 5 \\5 & 23 & 23.22 & 28 & 18.4 & 19.78 & 24 & 10 & 9.37 & 10 & 5 & 5 & 5 \\6 & 26 & 24.68 & 30 & 19.99 & 21.01 & 24 & 13 & 10.43 & 11 & 6 & 6 & 6 \\7 & 29 & 25.84 & 31 & 22.96 & 22.7 & 26 & 13 & 10.52 & 12 & - & - & - \\8 & 29.5 & 27.09 & 31 & 24.79 & 24.02 & 31 & 13 & 11.06 & 12 & - & - & - \\9 & 34 & 27.76 & 33 & 26.08 & 24.51 & 32 & 14 & 11.71 & 14 & - & - & - \\10 & 36 & 28.49 & 34 & 27.08 & 24.81 & 33 & 14 & 12.02 & 14 & - & - & - \\11 & 36 & 29.12 & 35 & 28.58 & 25.49 & 33 & 14 & 12.47 & 14 & - & - & - \\12 & 37 & 29.59 & 36 & 30.33 & 26.69 & 34 & 14 & 13.19 & 14 & - & - & - \\13 & 38 & 30.44 & 37 & 33.5 & 28.07 & 35 & 14 & 13.52 & 14 & - & - & - \\14 & 39 & 31.04 & 39 & 35.5 & 29.28 & 35 & 14 & 14 & 14 & - & - & - \\15 & 39 & 31.27 & 39 & 36 & 30.03 & 36 & - & - & - & - & - & - \\16 & 39 & 31.66 & 39 & 37 & 30.37 & 36 & - & - & - & - & - & - \\Total & 39 & 39 & 39 & 42 & 42 & 42 & - & - & - & - & - & -  \\
     \hline 
     \end{tabular}
     \label{dist4}
     \end{table}

\begin{table}[htbp]
\small
     \centering
     \caption{Expected number of families that are assigned apartments with rank up to  $k=1,\ldots,16$ by the Rawlsian, PS and MTAV rules.}
     \begin{tabular}{|c|ccc|ccc|ccc|ccc|}
      \hline
        & \multicolumn{3}{c|}{$C_{17}$} & \multicolumn{3}{c|}{$C_{18}$} & \multicolumn{3}{c|}{$C_{19}$} & \multicolumn{3}{c|}{$C_{20}$}  \\ \hline
     Position & Rawls & PS & MTAV & Rawls & PS & MTAV & Rawls & PS & MTAV & Rawls & PS & MTAV \\ \hline
     1 & 3.5 & 4.62 & 5 & 5.33 & 8.39 & 9 & 2.8 & 3.14 & 4 & 2.21 & 5.34 & 5 \\2 & 8 & 6.75 & 7 & 7 & 8.58 & 9 & 3.8 & 3.97 & 5 & 5.92 & 6.74 & 8 \\3 & 9 & 7.68 & 9 & 9.5 & 9.09 & 9 & 3.8 & 4.34 & 5 & 7.17 & 8.66 & 9 \\4 & 9 & 8.6 & 9 & 9.5 & 9.15 & 9 & 4.83 & 4.57 & 5 & 10.42 & 10.57 & 12 \\5 & 9 & 8.8 & 9 & 10.25 & 9.69 & 10 & 5.75 & 5.33 & 5 & 14.75 & 12.5 & 14 \\6 & 9 & 8.83 & 9 & 10.5 & 10.17 & 10 & 7 & 6.33 & 7 & 16 & 13.53 & 16 \\7 & 9 & 8.91 & 9 & 14 & 12.18 & 14 & 7.89 & 7.11 & 7 & 16 & 13.84 & 16 \\8 & 9 & 8.98 & 9 & 15 & 13.09 & 15 & 8 & 8 & 8 & 17 & 14.64 & 16 \\9 & 9 & 9 & 9 & 15 & 13.27 & 15 & 9 & 9 & 9 & 19 & 15.61 & 19 \\10 & - & - & - & 15 & 13.65 & 15 & - & - & - & 20 & 16.11 & 20 \\11 & - & - & - & 15 & 14.19 & 15 & - & - & - & 20 & 16.53 & 20 \\12 & - & - & - & 15 & 14.27 & 15 & - & - & - & 20 & 17.15 & 20 \\13 & - & - & - & 15 & 14.61 & 15 & - & - & - & 20 & 17.69 & 20 \\14 & - & - & - & 15 & 14.9 & 15 & - & - & - & 20 & 17.95 & 20 \\15 & - & - & - & 15 & 15 & 15 & - & - & - & 20 & 18.61 & 20 \\16 & - & - & - & - & - & - & - & - & - & 20 & 18.89 & 20 \\Total & - & - & - & - & - & - & - & - & - & 20 & 20 & 20  \\
     \hline 
     \end{tabular}
     \label{dist5}
     \end{table}

\begin{table}[htbp]
\small
     \centering
      \caption{Expected number of families that are assigned apartments with rank up to  $k=1,\ldots,16$ by the Rawlsian, PS and MTAV rules.}
     \begin{tabular}{|c|ccc|ccc|ccc|ccc|}
     \hline
        & \multicolumn{3}{c|}{$C_{21}$} & \multicolumn{3}{c|}{$C_{22}$} & \multicolumn{3}{c|}{$C_{23}$} & \multicolumn{3}{c|}{$C_{24}$}  \\ \hline
     Position & Rawls & PS & MTAV & Rawls & PS & MTAV & Rawls & PS & MTAV & Rawls & PS & MTAV \\ \hline
     1 & 8.79 & 12.44 & 12 & 4 & 3.75 & 4 & 5.5 & 12.16 & 10 & 2.38 & 2.83 & 3 \\2 & 10.71 & 13.01 & 14 & 7 & 5.38 & 7 & 11.31 & 17.44 & 15 & 4.55 & 3.4 & 5 \\3 & 13.83 & 14.38 & 17 & 7 & 6.67 & 7 & 15.94 & 20.44 & 19 & 4.8 & 3.97 & 5 \\4 & 15.17 & 15.24 & 17 & 7 & 6.83 & 7 & 24.42 & 23.16 & 27 & 4.8 & 4.29 & 5 \\5 & 22 & 18.45 & 21 & 7 & 6.88 & 7 & 28.47 & 25.91 & 29 & 5 & 5 & 5 \\6 & 23 & 19.15 & 22 & 7 & 6.92 & 7 & 32.36 & 27.93 & 34 & 7 & 6.29 & 7 \\7 & 24 & 19.79 & 24 & 7 & 7 & 7 & 34.92 & 29.24 & 35 & 8 & 7.46 & 8 \\8 & 24 & 20.24 & 24 & - & - & - & 37.5 & 30.17 & 37 & 8 & 8 & 8 \\9 & 24 & 20.51 & 24 & - & - & - & 38 & 31.08 & 37 & - & - & - \\10 & 24 & 20.88 & 24 & - & - & - & 38.5 & 31.54 & 38 & - & - & - \\11 & 24 & 21.69 & 24 & - & - & - & 40 & 32.24 & 40 & - & - & - \\12 & 24 & 21.84 & 24 & - & - & - & 40 & 32.63 & 40 & - & - & - \\13 & 24 & 22.2 & 24 & - & - & - & 40 & 32.99 & 40 & - & - & - \\14 & 24 & 22.35 & 24 & - & - & - & 40 & 33.68 & 40 & - & - & - \\15 & 24 & 22.67 & 24 & - & - & - & 40 & 33.99 & 40 & - & - & - \\16 & 24 & 22.77 & 24 & - & - & - & 40 & 34.35 & 40 & - & - & - \\Total & 24 & 24 & 24 & - & - & - & 40 & 40 & 40 & - & - & -  \\
     \hline 
     \end{tabular}
     \label{dist6}
     \end{table}
     
\newpage 

\subsection{Cumulative distribution function of the expected number of families that are assigned apartments by rank by the Rawlsian, PS and MTAV rules.}\label{cum_figs}

\begin{figure}[htp]
\centering
\begin{subfigure}{.5\textwidth}
  \centering
  \includegraphics[width=1.1\linewidth]{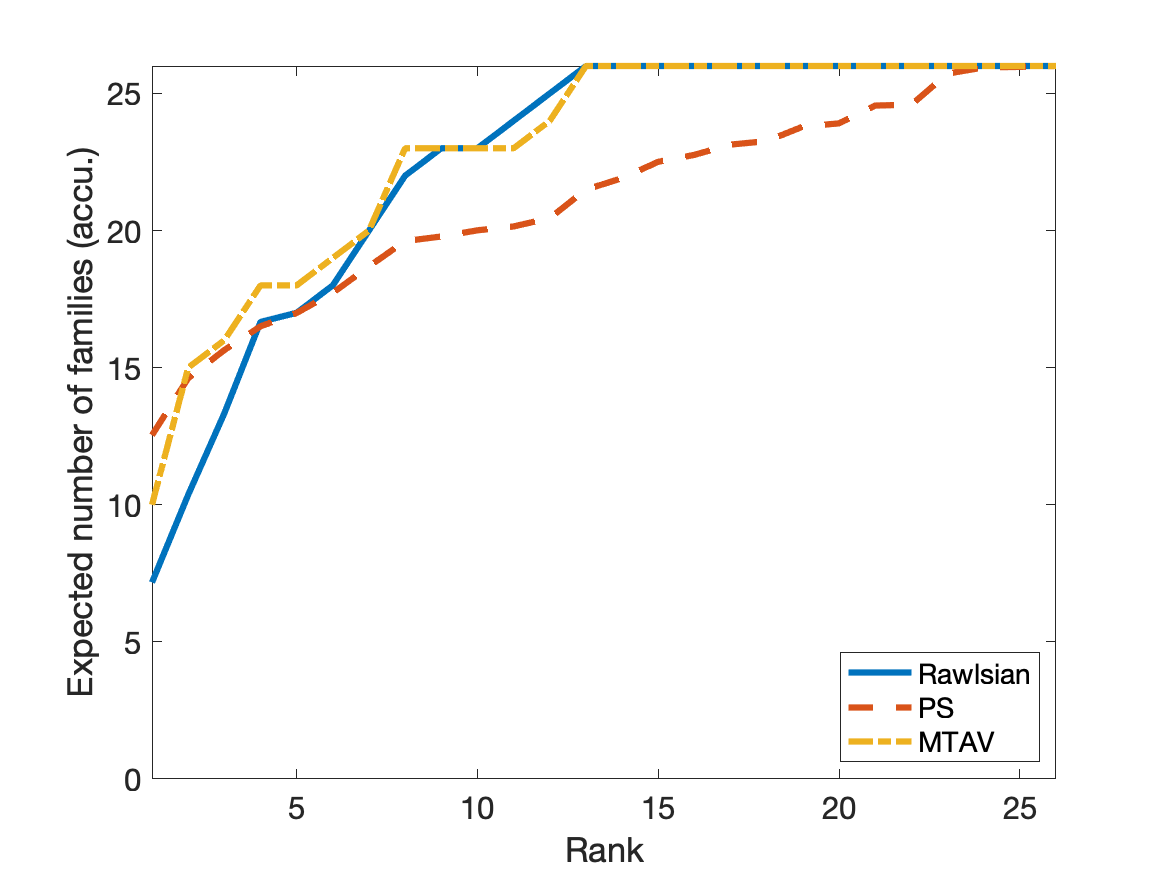}
  \caption{Cooperative $C_1$}
  \label{fig:dist}
\end{subfigure}%
\begin{subfigure}{.5\textwidth}
  \centering
  \includegraphics[width=1.1\linewidth]{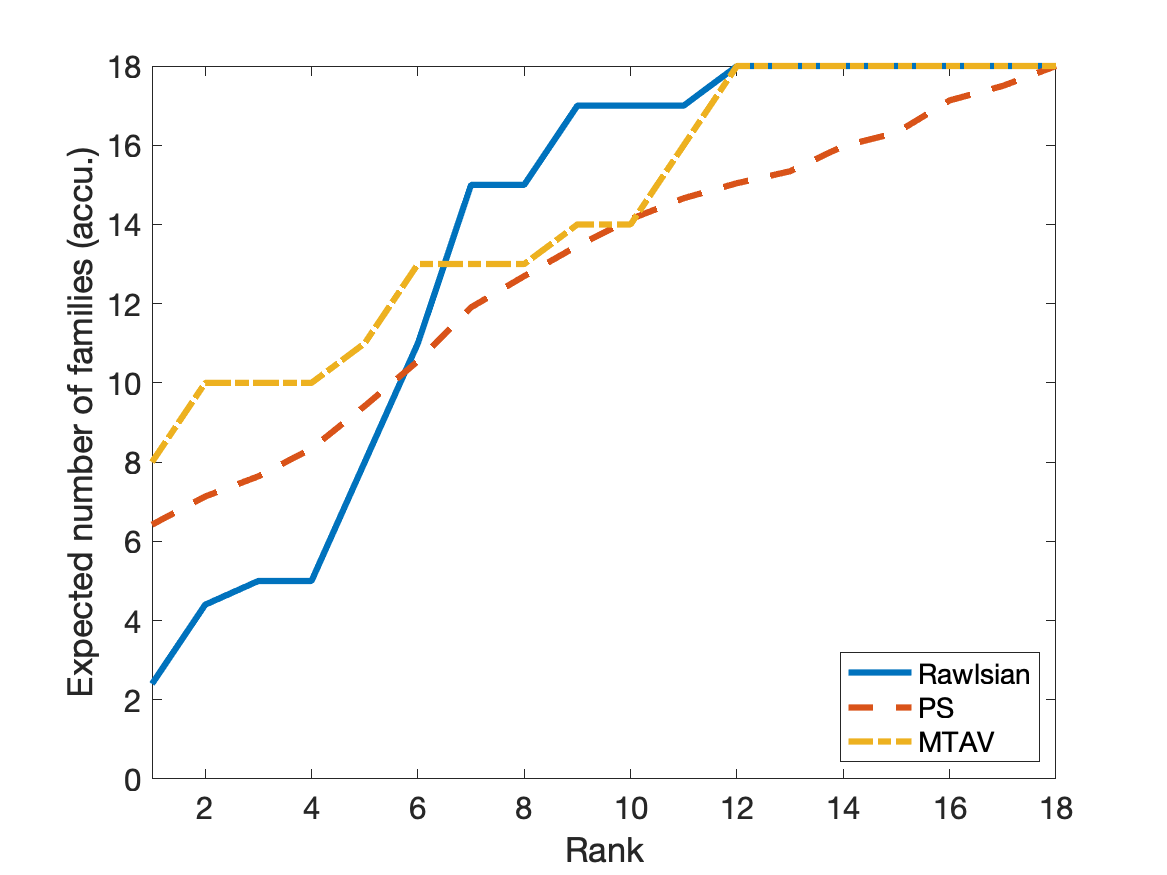}
  \caption{Cooperative $C_{2}$}
  \label{fig:sub2}
\end{subfigure}
\caption{}
\end{figure}

\begin{figure}[htp]
\centering
\begin{subfigure}{.5\textwidth}
  \centering
  \includegraphics[width=1.1\linewidth]{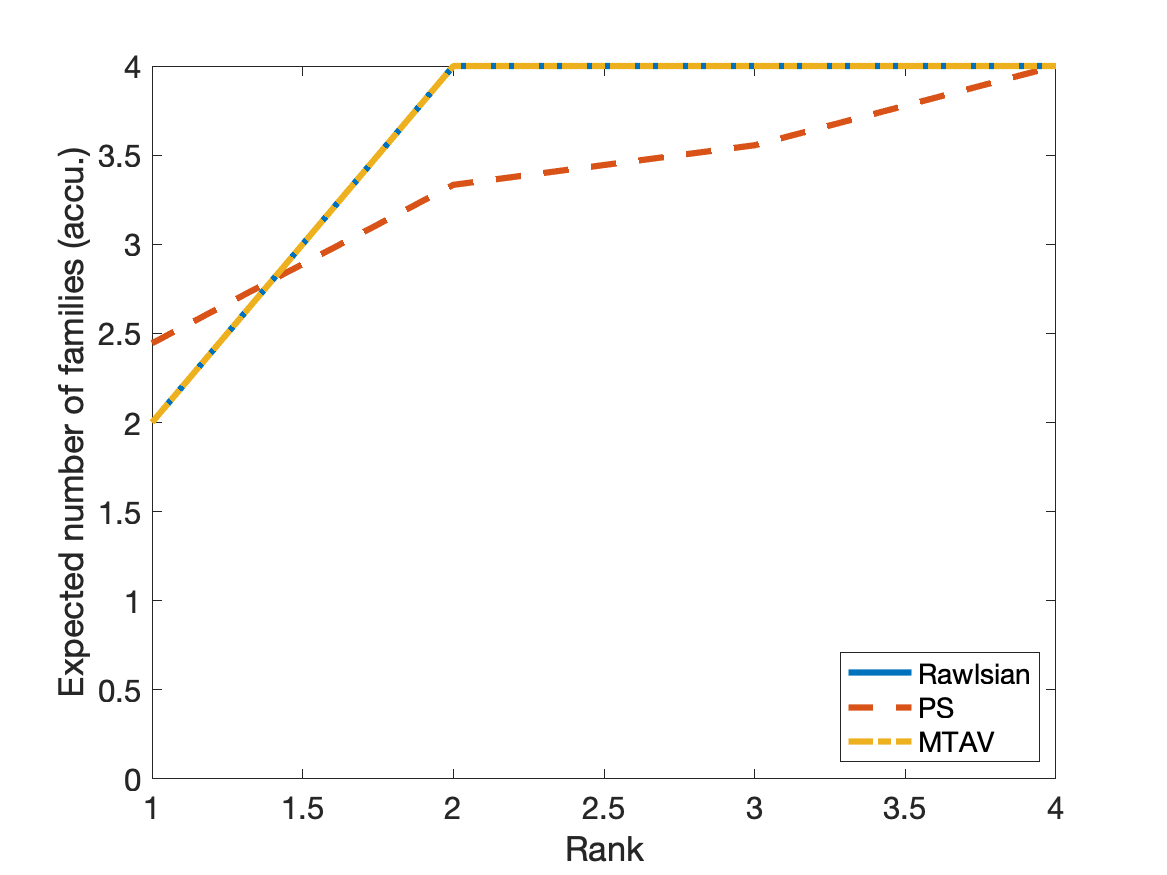}
  \caption{Cooperative $C_3$}
  \label{fig:dist}
\end{subfigure}%
\begin{subfigure}{.5\textwidth}
  \centering
  \includegraphics[width=1.1\linewidth]{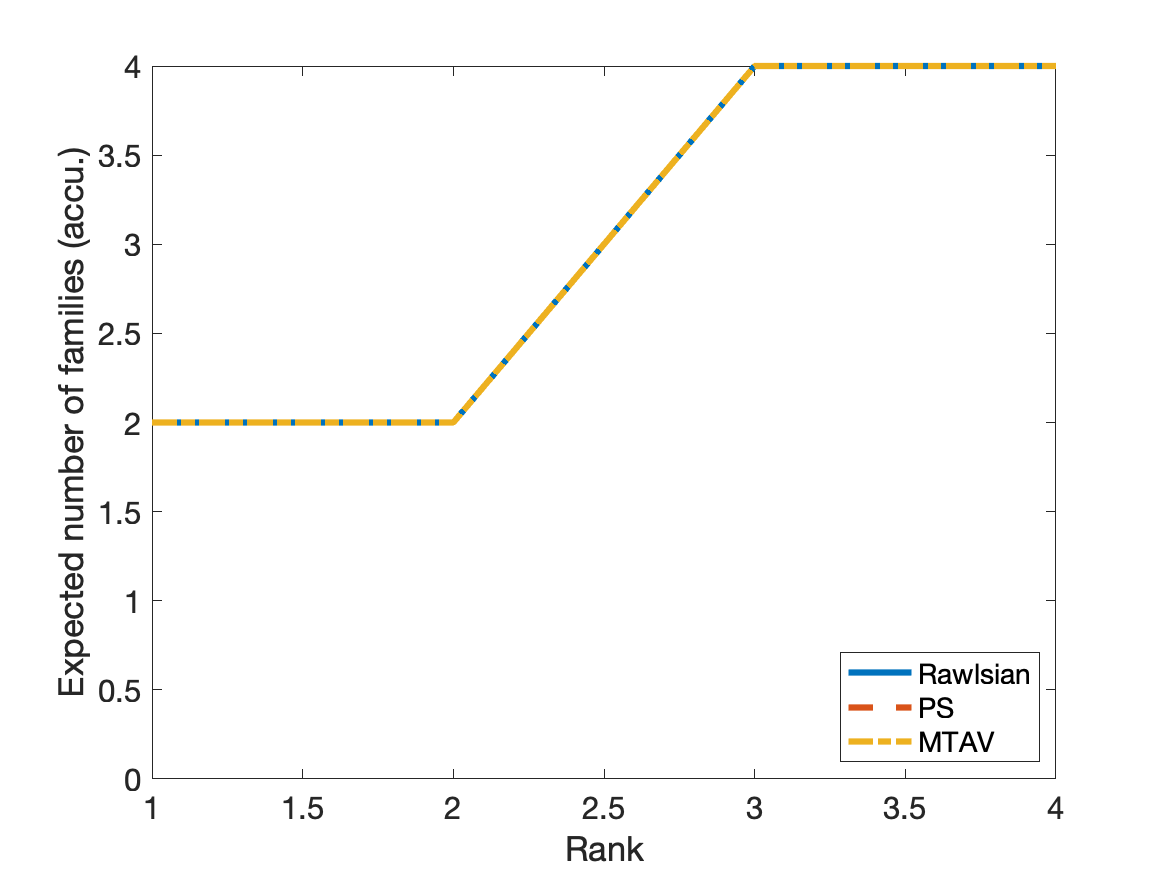}
  \caption{Cooperative $C_{4}$}
  \label{fig:sub2}
\end{subfigure}
\caption{}
\end{figure}

\begin{figure}[htp]
\centering
\begin{subfigure}{.5\textwidth}
  \centering
  \includegraphics[width=1.1\linewidth]{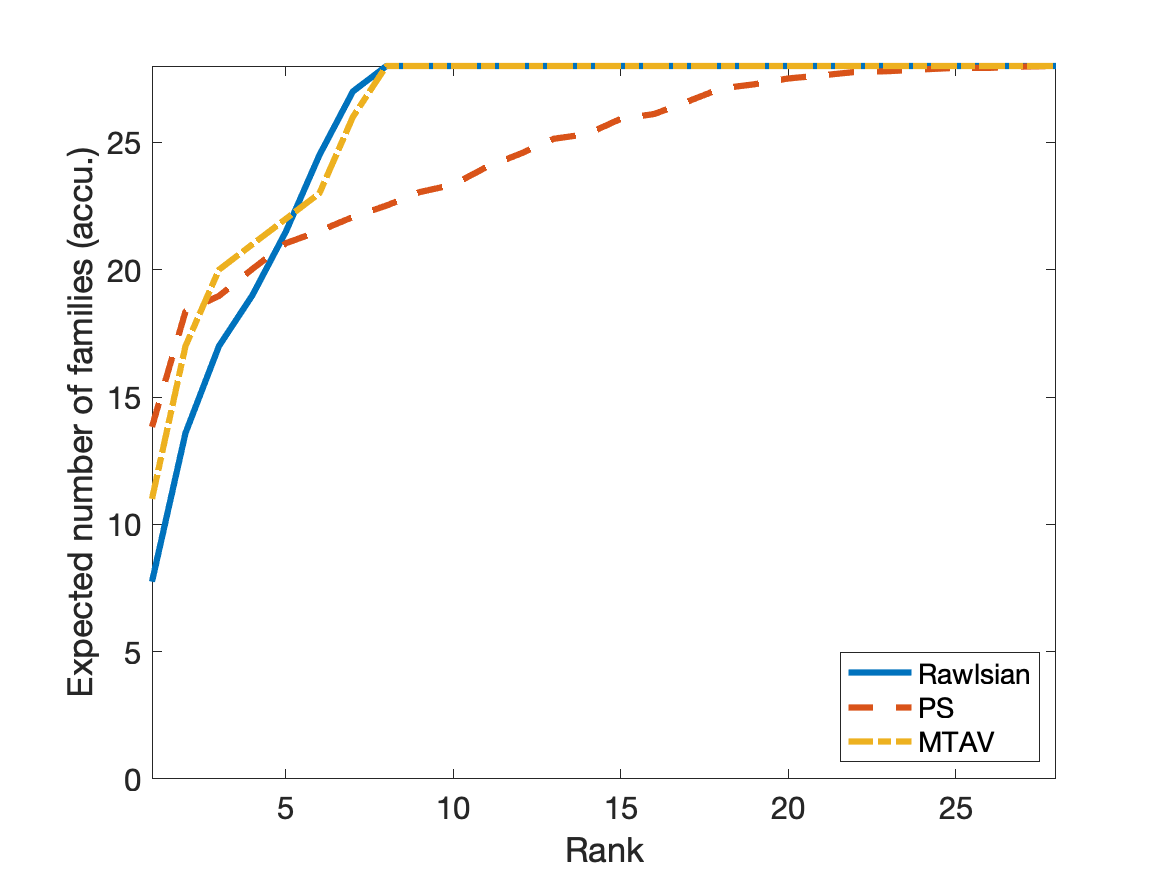}
  \caption{Cooperative $C_5$}
  \label{fig:dist}
\end{subfigure}%
\begin{subfigure}{.5\textwidth}
  \centering
  \includegraphics[width=1.1\linewidth]{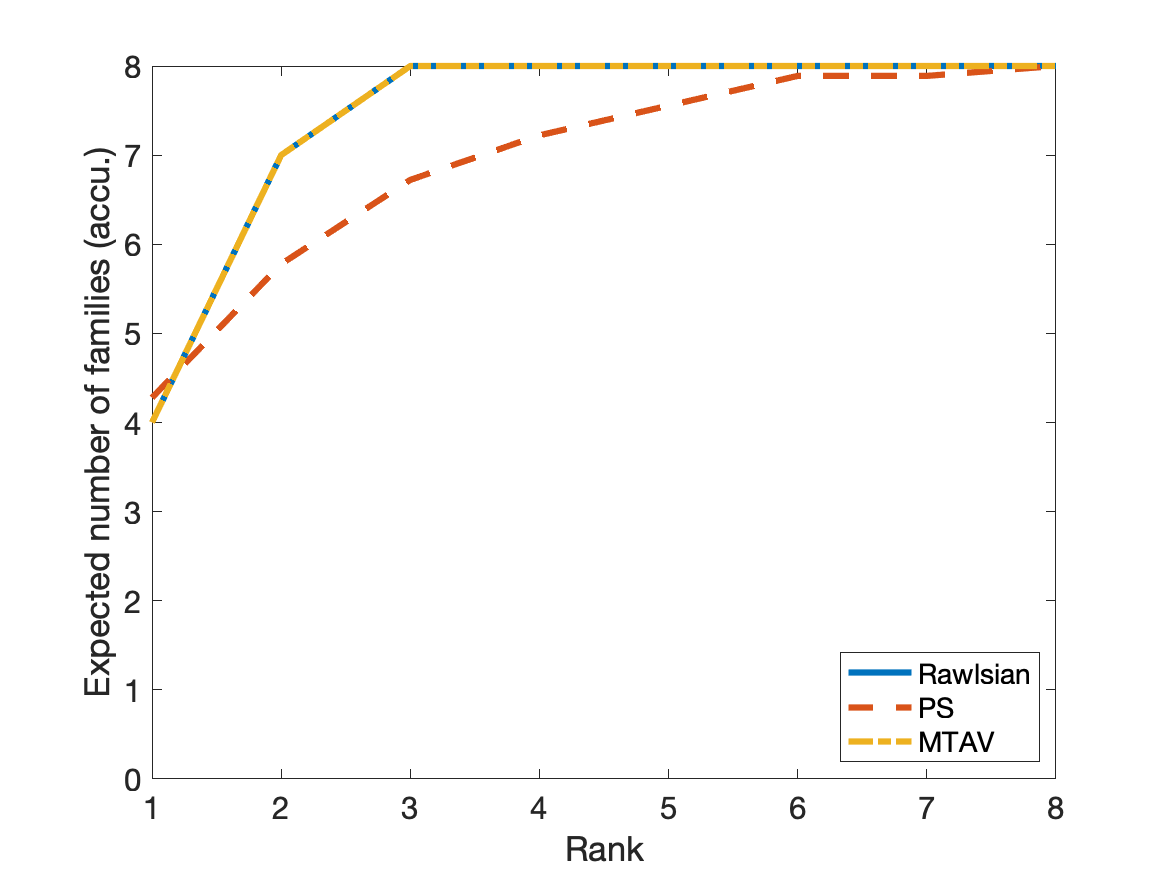}
  \caption{Cooperative $C_{6}$}
  \label{fig:sub2}
\end{subfigure}
\caption{}
\end{figure}

\begin{figure}[htp]
\centering
\begin{subfigure}{.5\textwidth}
  \centering
  \includegraphics[width=1.1\linewidth]{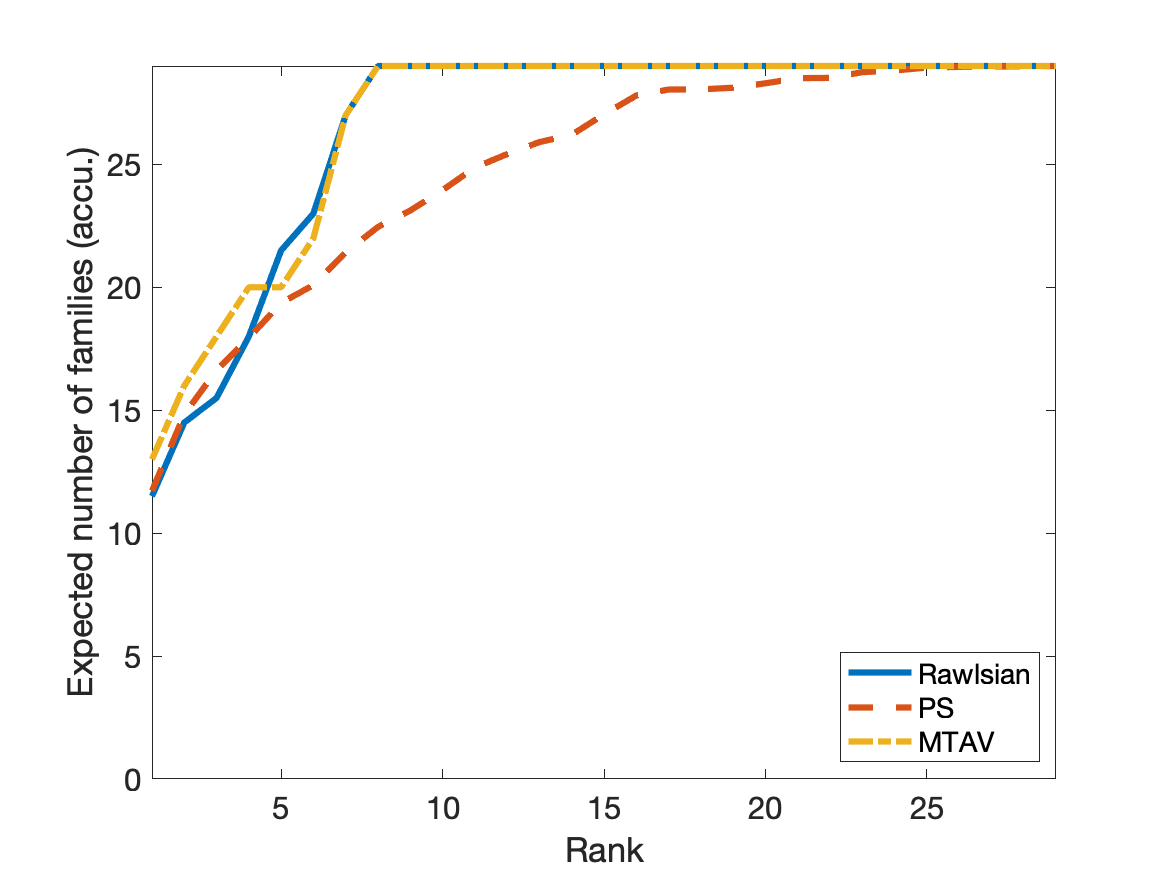}
  \caption{Cooperative $C_7$}
  \label{fig:dist}
\end{subfigure}%
\begin{subfigure}{.5\textwidth}
  \centering
  \includegraphics[width=1.1\linewidth]{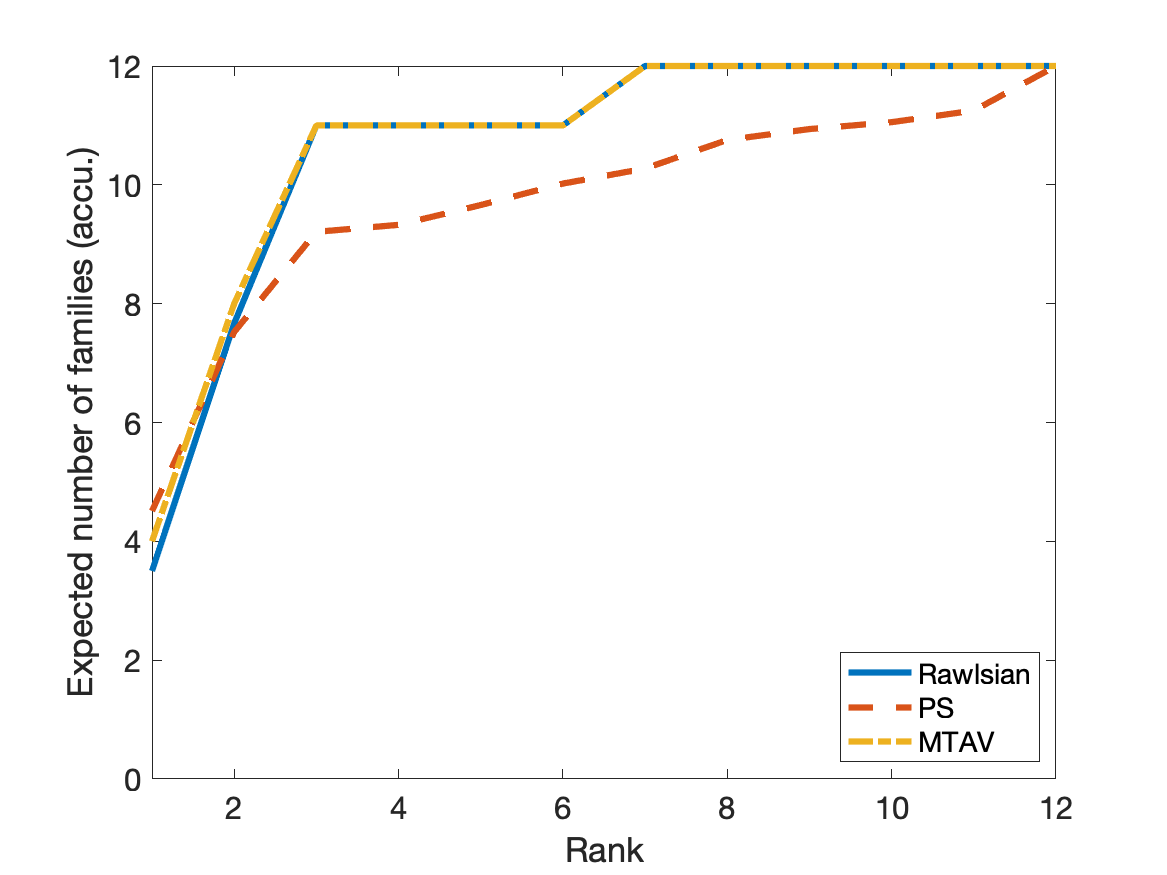}
  \caption{Cooperative $C_{8}$}
  \label{fig:sub2}
\end{subfigure}
\caption{}
\end{figure}

\begin{figure}[htp]
\centering
\begin{subfigure}{.5\textwidth}
  \centering
  \includegraphics[width=1.1\linewidth]{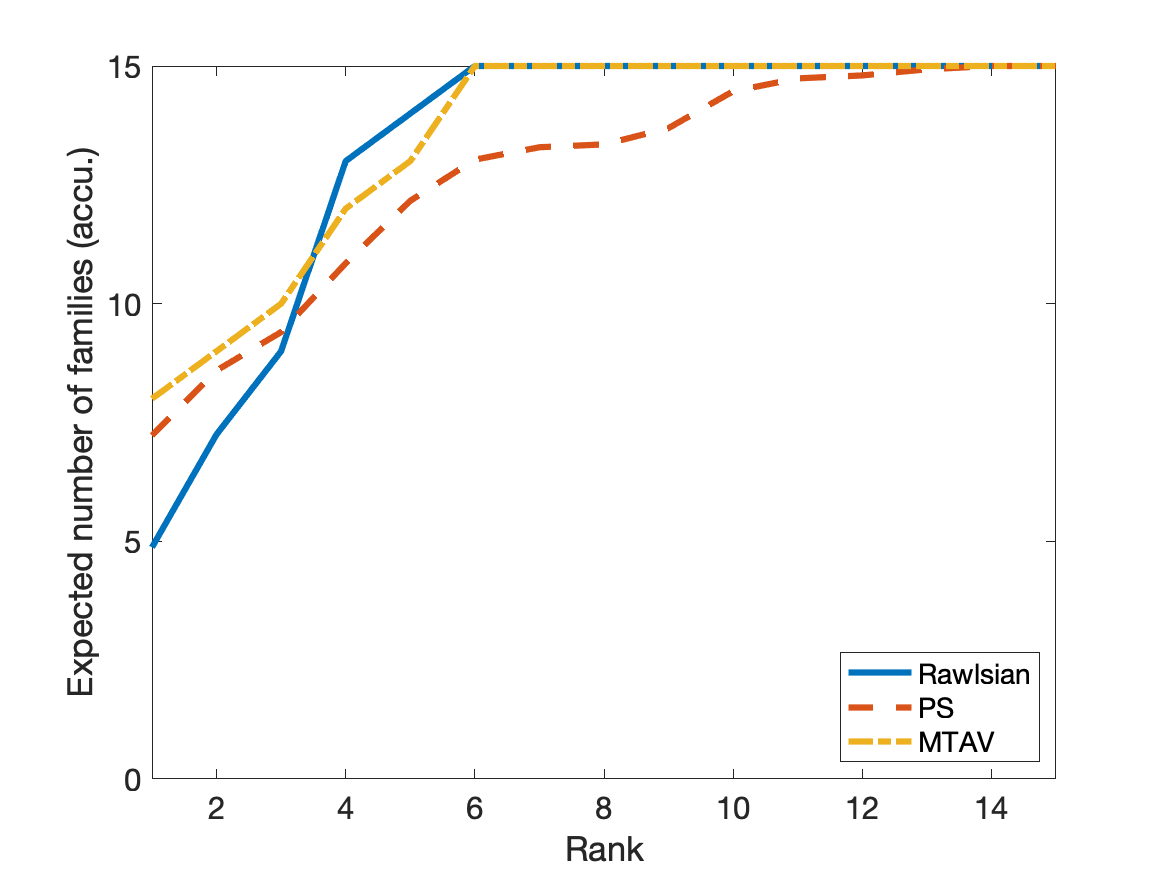}
  \caption{Cooperative $C_9$}
  \label{fig:dist}
\end{subfigure}%
\begin{subfigure}{.5\textwidth}
  \centering
  \includegraphics[width=1.1\linewidth]{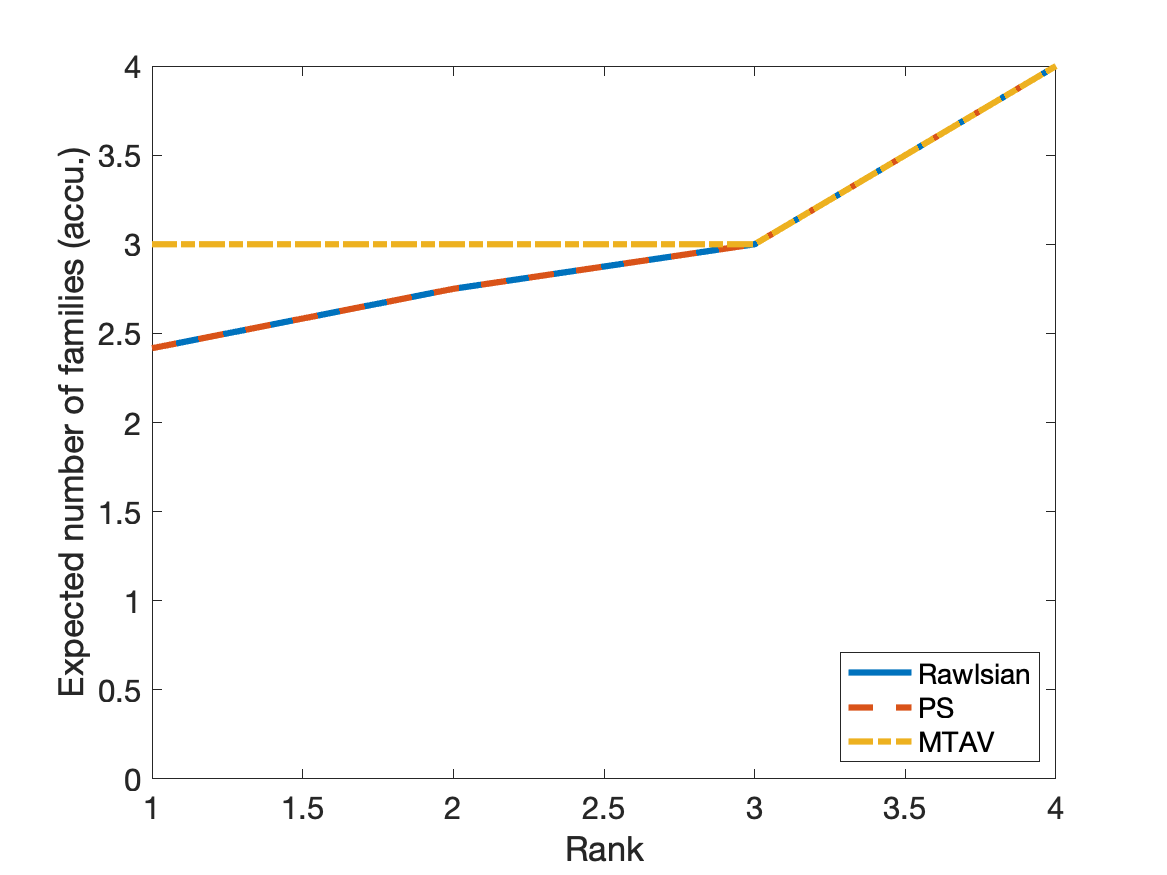}
  \caption{Cooperative $C_{10}$}
  \label{fig:sub2}
\end{subfigure}
\caption{}
\end{figure}

\begin{figure}[htp]
\centering
\begin{subfigure}{.5\textwidth}
  \centering
  \includegraphics[width=1.1\linewidth]{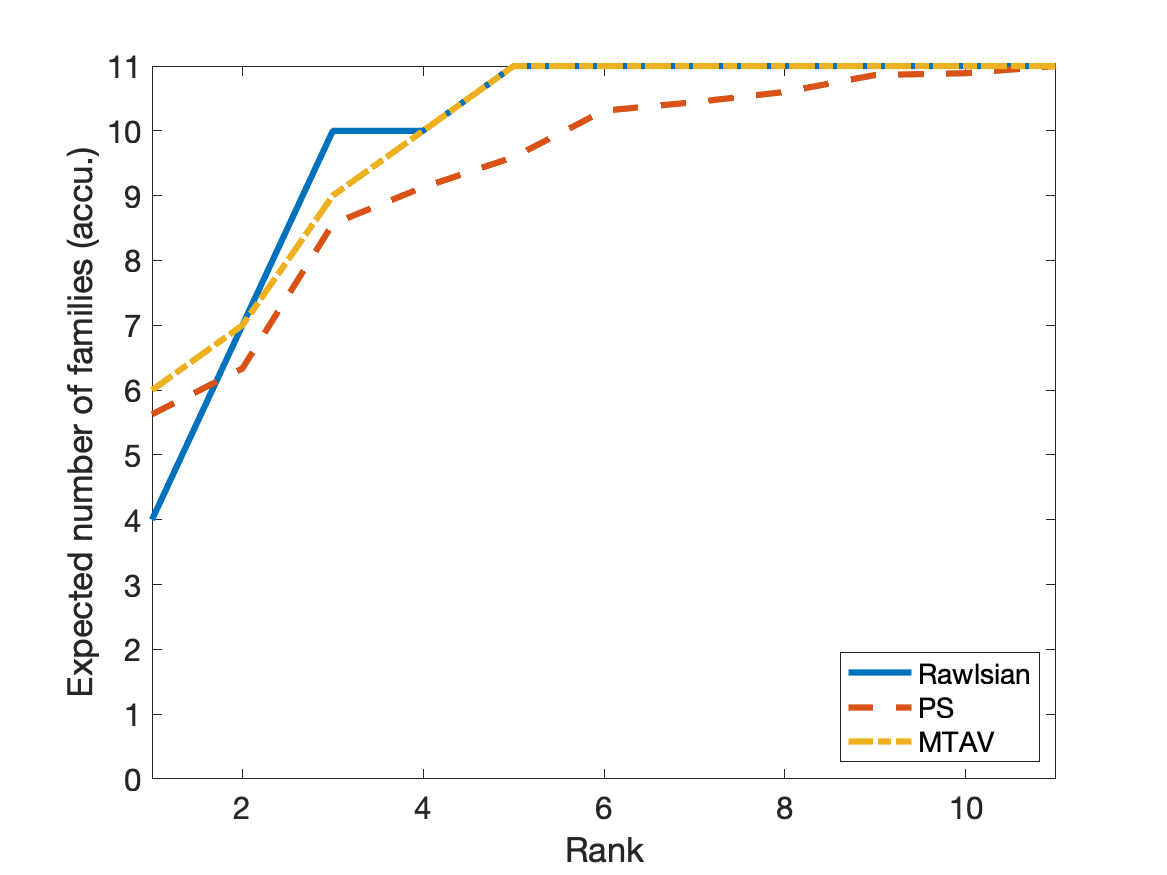}
  \caption{Cooperative $C_{11}$}
  \label{fig:dist}
\end{subfigure}%
\begin{subfigure}{.5\textwidth}
  \centering
  \includegraphics[width=1.1\linewidth]{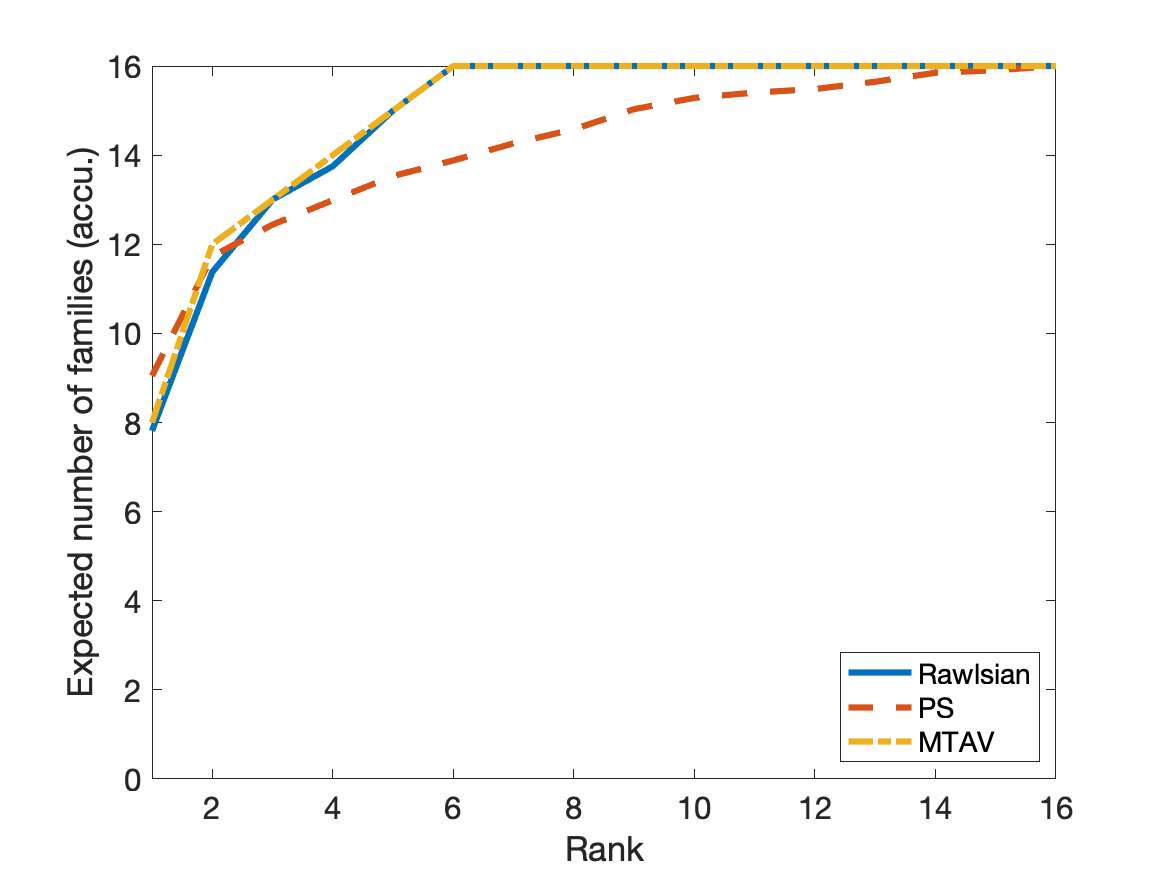}
  \caption{Cooperative $C_{12}$}
  \label{fig:sub2}
\end{subfigure}
\caption{}
\end{figure}

\begin{figure}[htp]
\centering
\begin{subfigure}{.5\textwidth}
  \centering
  \includegraphics[width=1.1\linewidth]{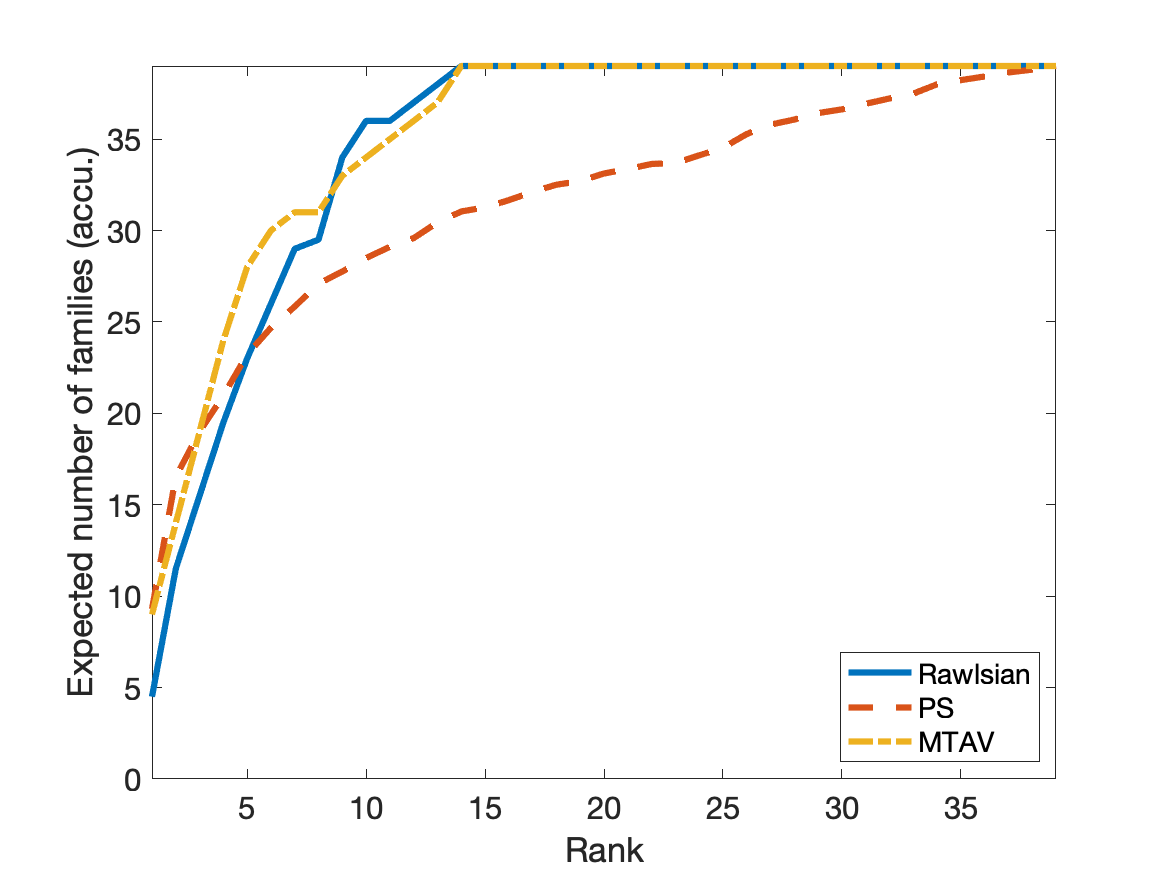}
  \caption{Cooperative $C_{13}$}
  \label{fig:dist}
\end{subfigure}%
\begin{subfigure}{.5\textwidth}
  \centering
  \includegraphics[width=1.1\linewidth]{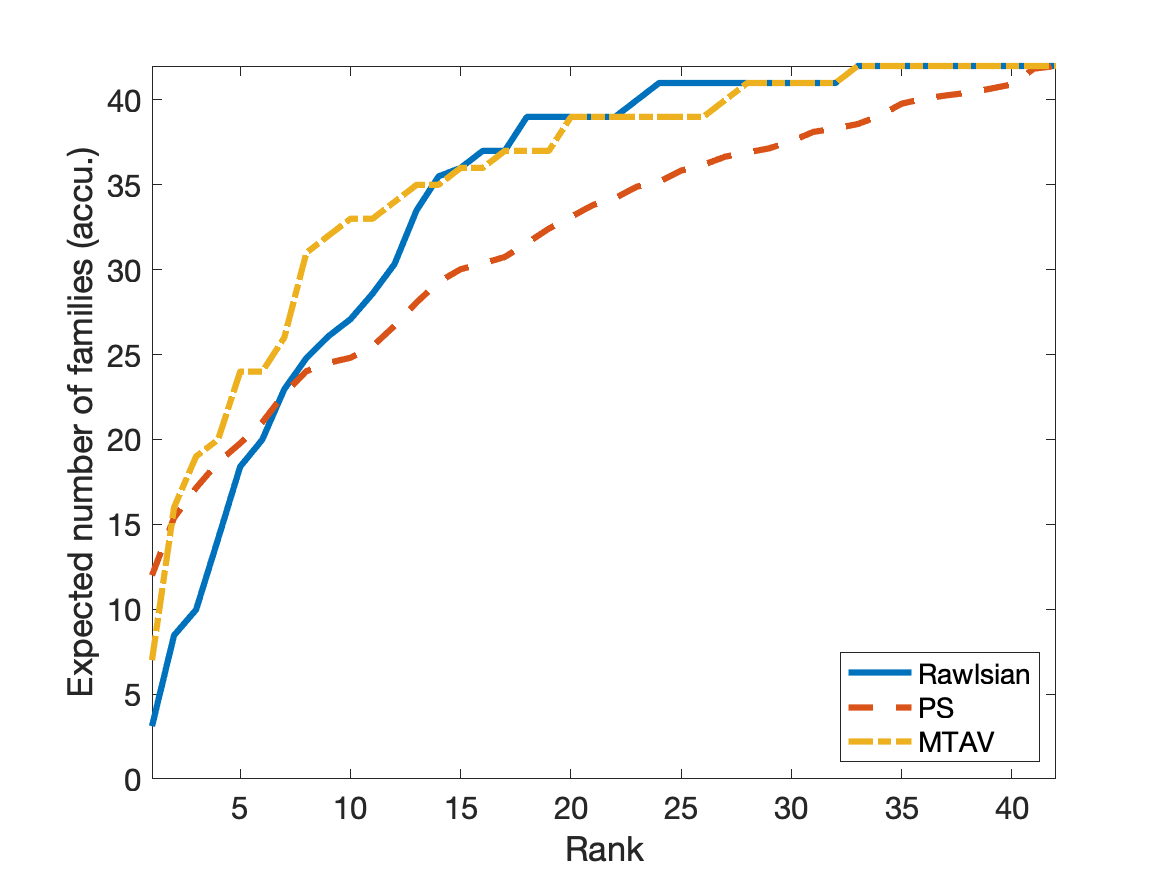}
  \caption{Cooperative $C_{14}$}
  \label{fig:sub2}
\end{subfigure}
\caption{}
\end{figure}

\begin{figure}[htp]
\centering
\begin{subfigure}{.5\textwidth}
  \centering
  \includegraphics[width=1.1\linewidth]{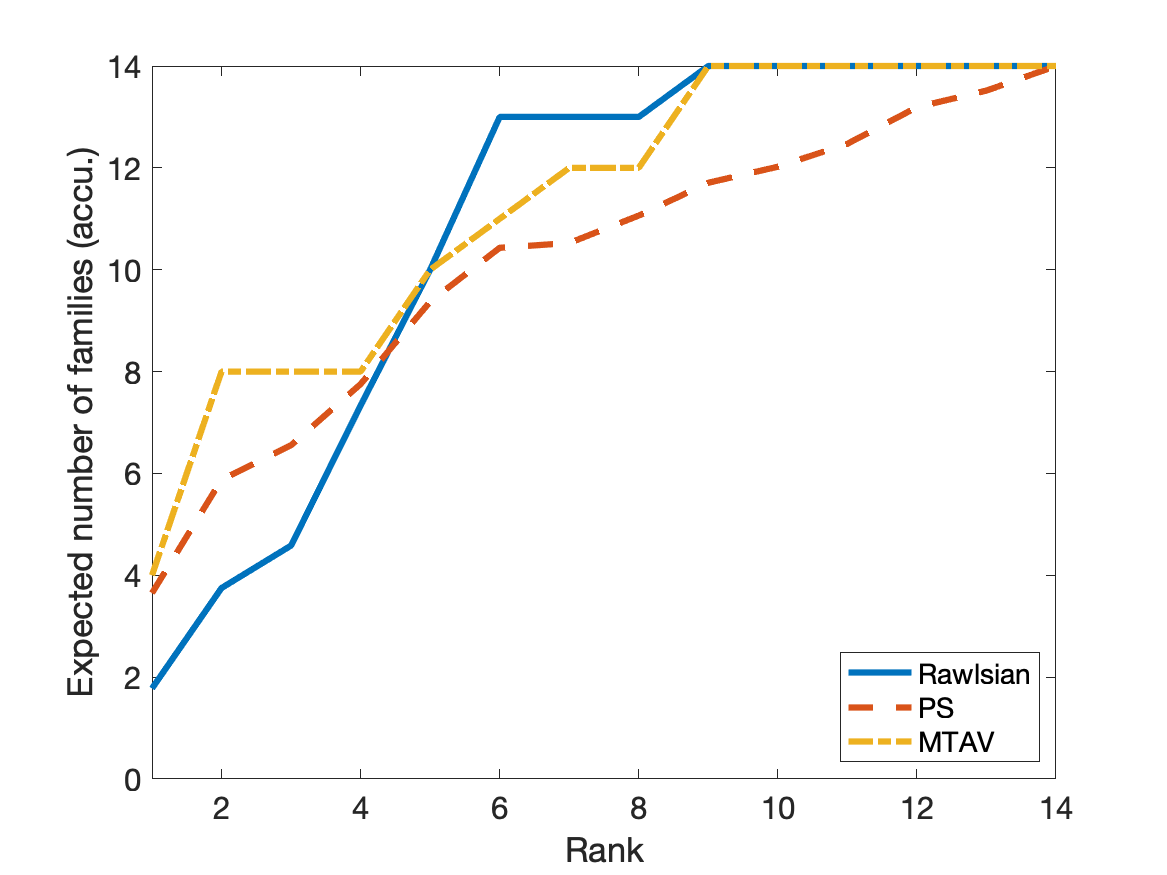}
  \caption{Cooperative $C_{15}$}
  \label{fig:dist}
\end{subfigure}%
\begin{subfigure}{.5\textwidth}
  \centering
  \includegraphics[width=1.1\linewidth]{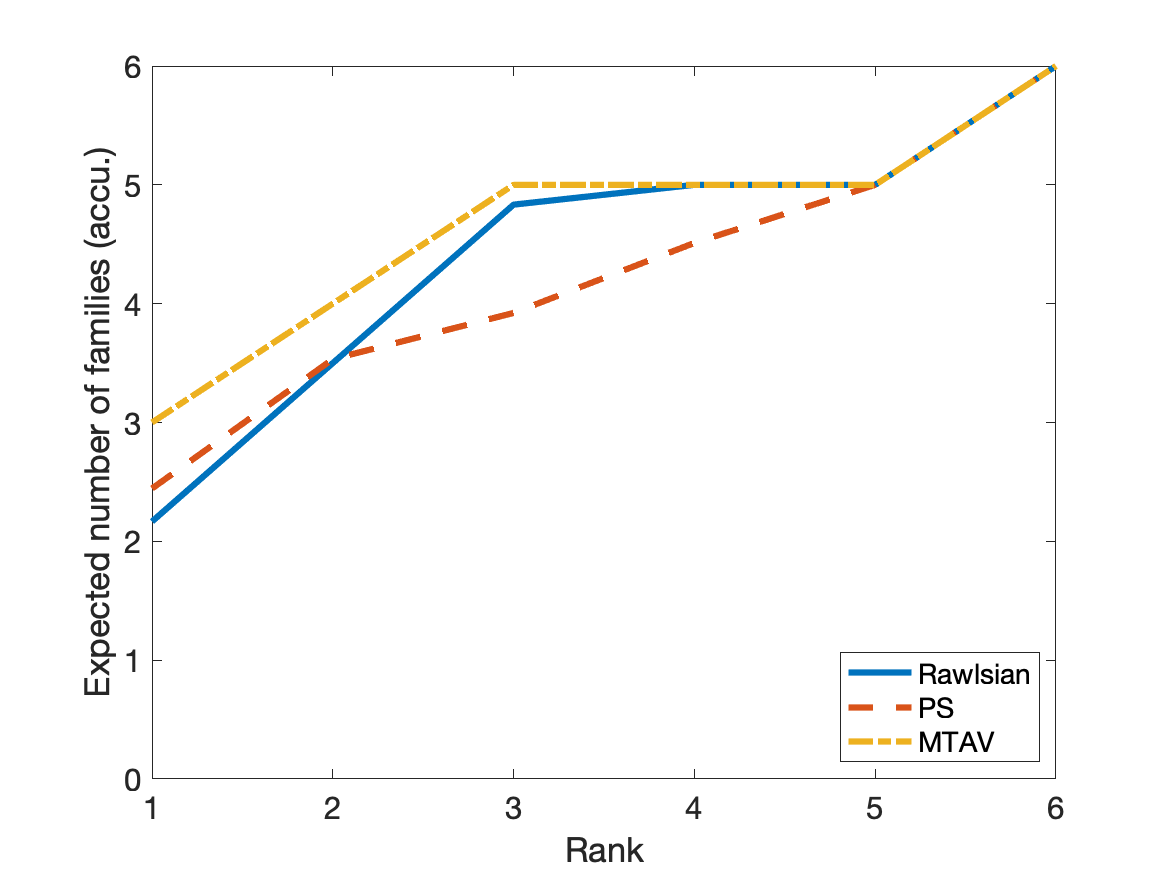}
  \caption{Cooperative $C_{16}$}
  \label{fig:sub2}
\end{subfigure}
\caption{}
\end{figure}

\begin{figure}[htp]
\centering
\begin{subfigure}{.5\textwidth}
  \centering
  \includegraphics[width=1.1\linewidth]{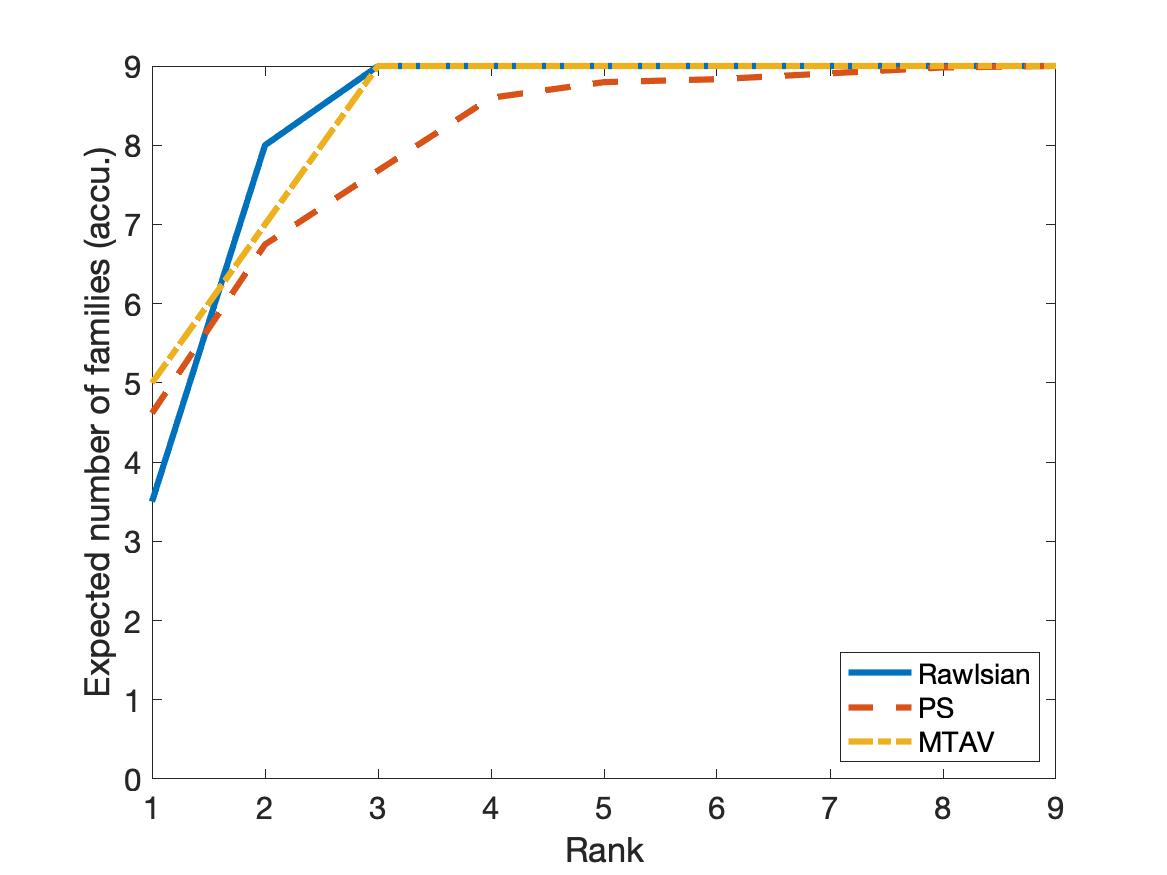}
  \caption{Cooperative $C_{17}$}
  \label{fig:dist}
\end{subfigure}%
\begin{subfigure}{.5\textwidth}
  \centering
  \includegraphics[width=1.1\linewidth]{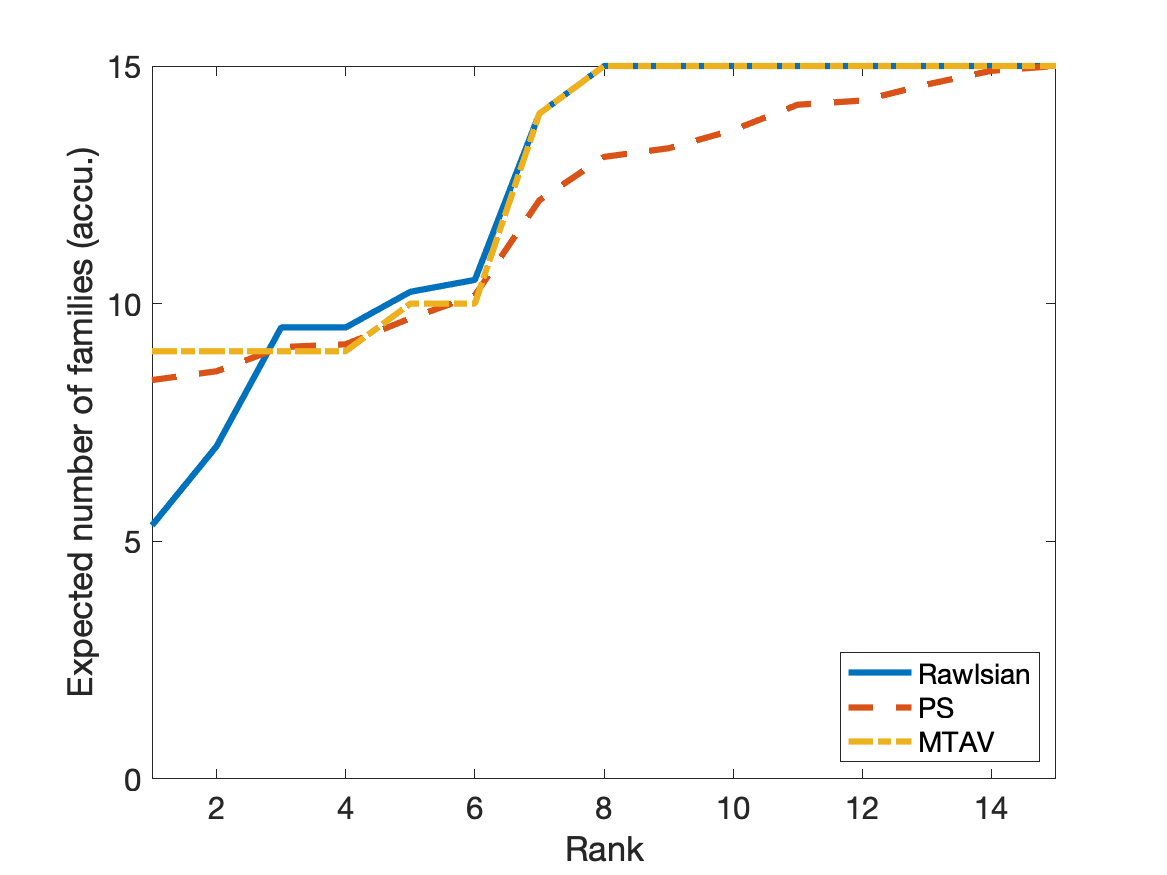}
  \caption{Cooperative $C_{18}$}
  \label{fig:sub2}
\end{subfigure}
\caption{}
\end{figure}

\begin{figure}[htp]
\centering
\begin{subfigure}{.5\textwidth}
  \centering
  \includegraphics[width=1.1\linewidth]{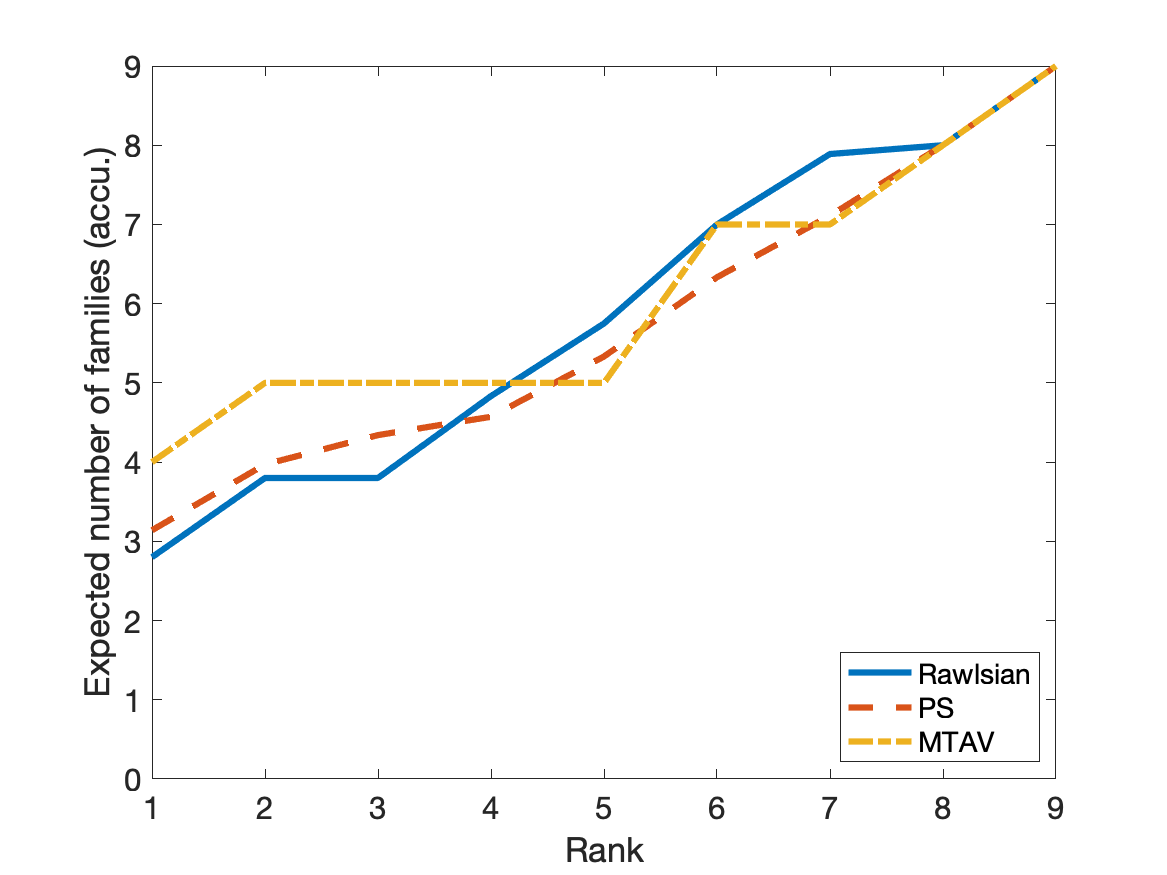}
  \caption{Cooperative $C_{19}$}
  \label{fig:dist}
\end{subfigure}%
\begin{subfigure}{.5\textwidth}
  \centering
  \includegraphics[width=1.1\linewidth]{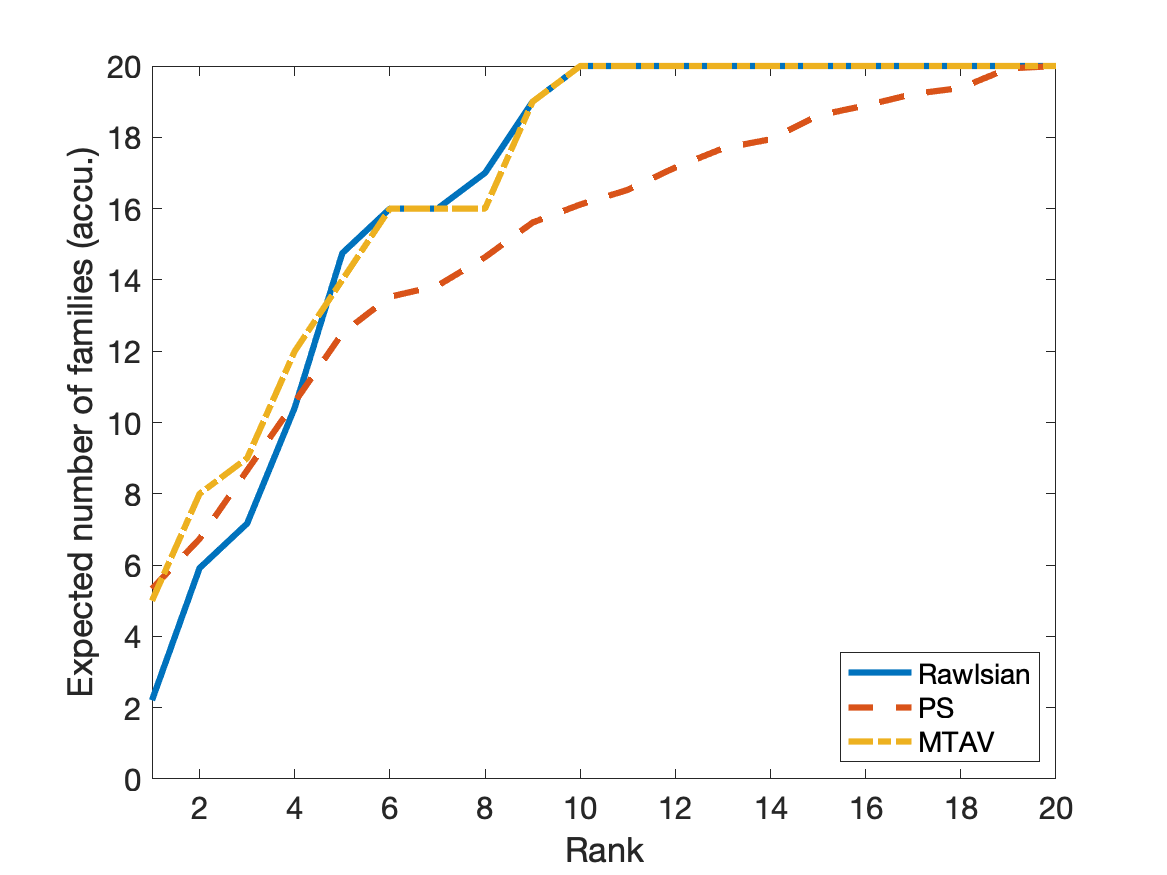}
  \caption{Cooperative $C_{20}$}
  \label{fig:sub2}
\end{subfigure}
\caption{}
\end{figure}

\begin{figure}[htp]
\centering
\begin{subfigure}{.5\textwidth}
  \centering
  \includegraphics[width=1.1\linewidth]{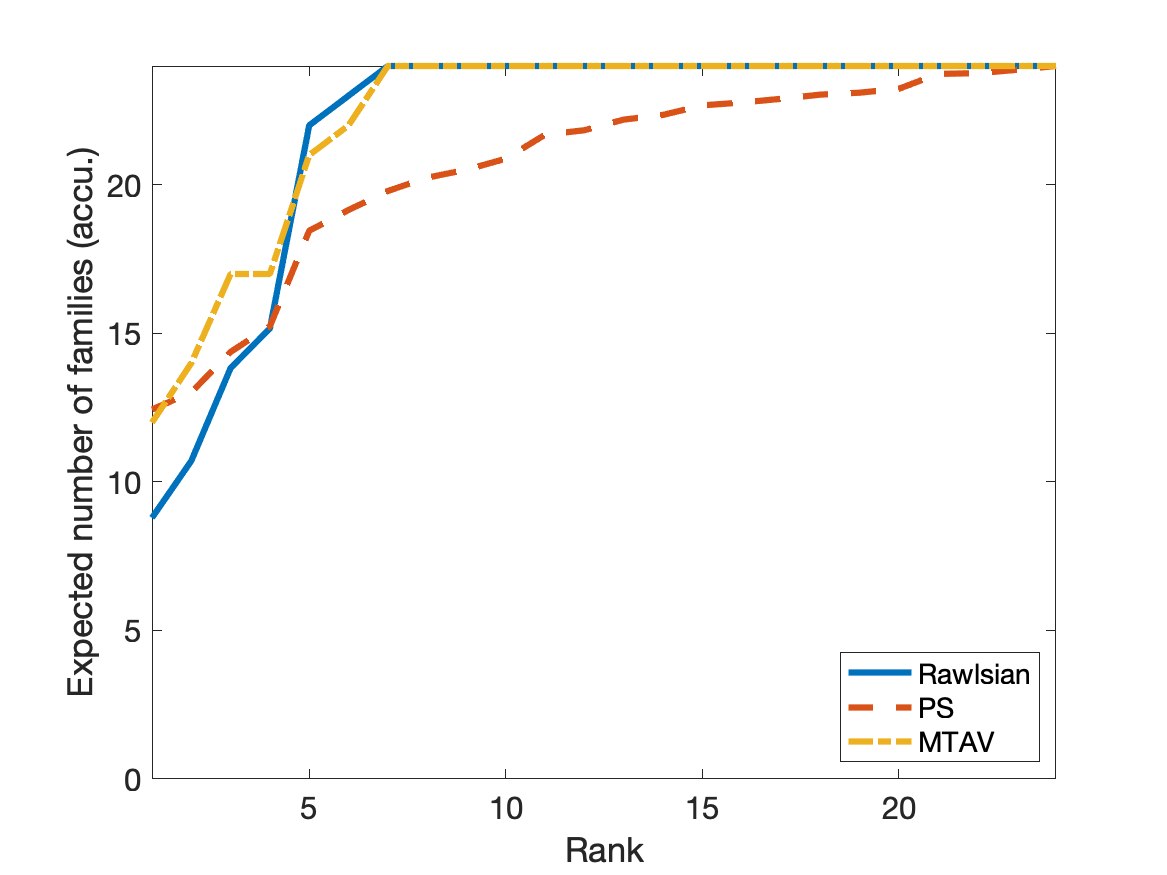}
  \caption{Cooperative $C_{21}$}
  \label{fig:dist}
\end{subfigure}%
\begin{subfigure}{.5\textwidth}
  \centering
  \includegraphics[width=1.1\linewidth]{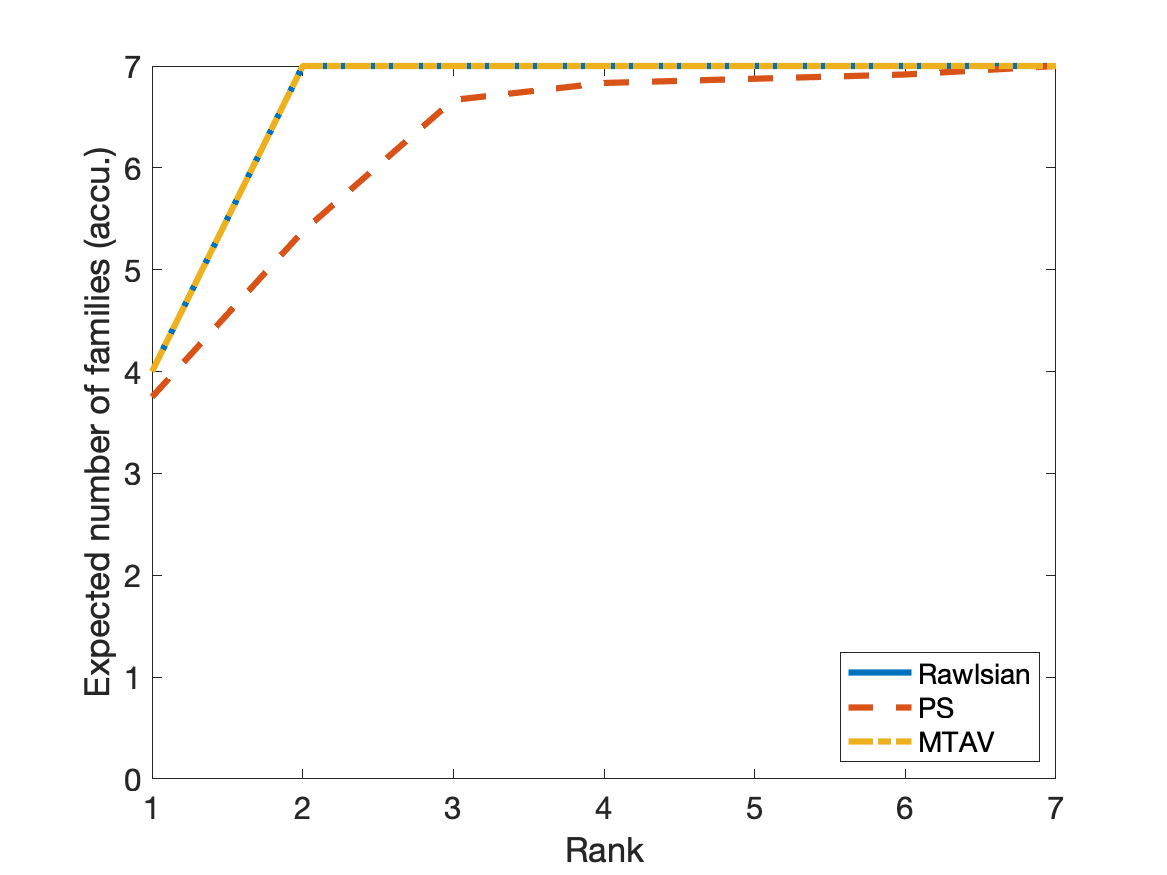}
  \caption{Cooperative $C_{22}$}
  \label{fig:sub2}
\end{subfigure}
\caption{}
\end{figure}

\begin{figure}[htp]
\centering
\begin{subfigure}{.5\textwidth}
  \centering
  \includegraphics[width=1.1\linewidth]{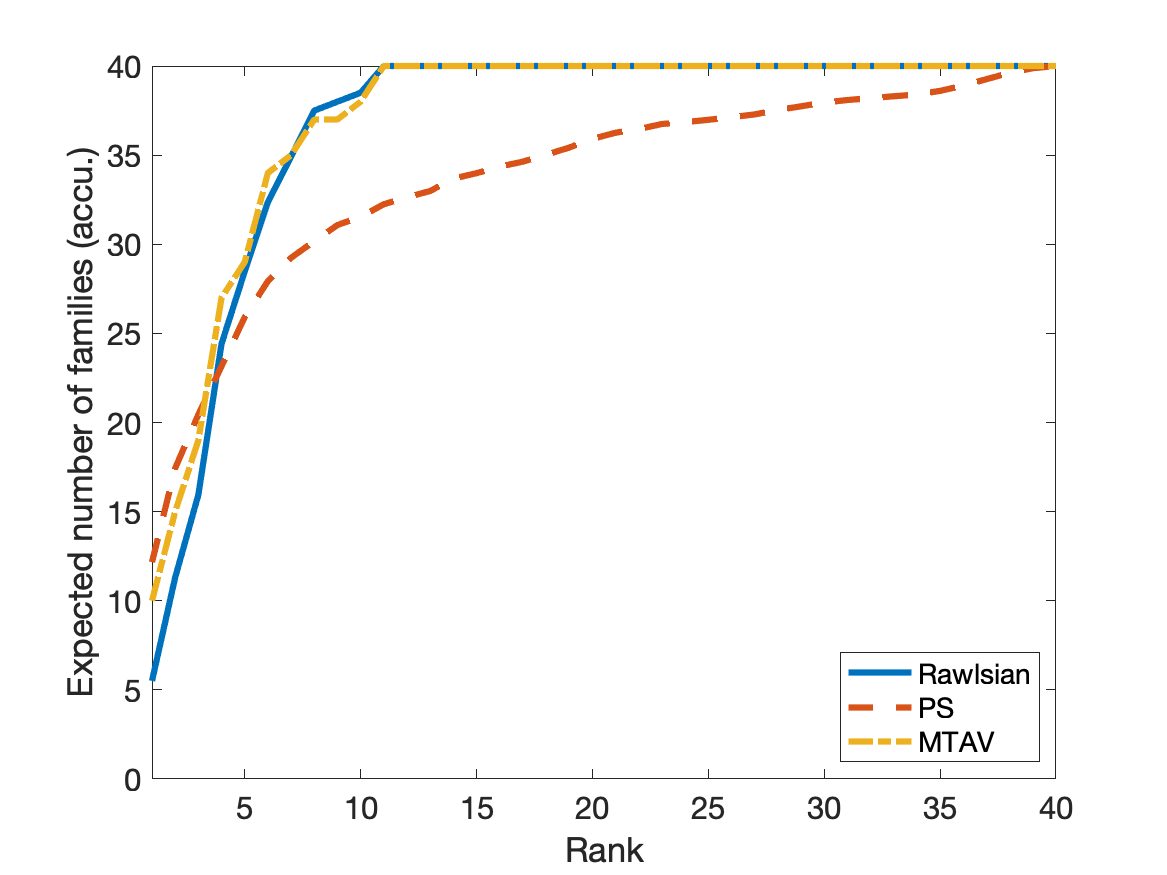}
  \caption{Cooperative $C_{23}$}
  \label{fig:dist}
\end{subfigure}%
\begin{subfigure}{.5\textwidth}
  \centering
  \includegraphics[width=1.1\linewidth]{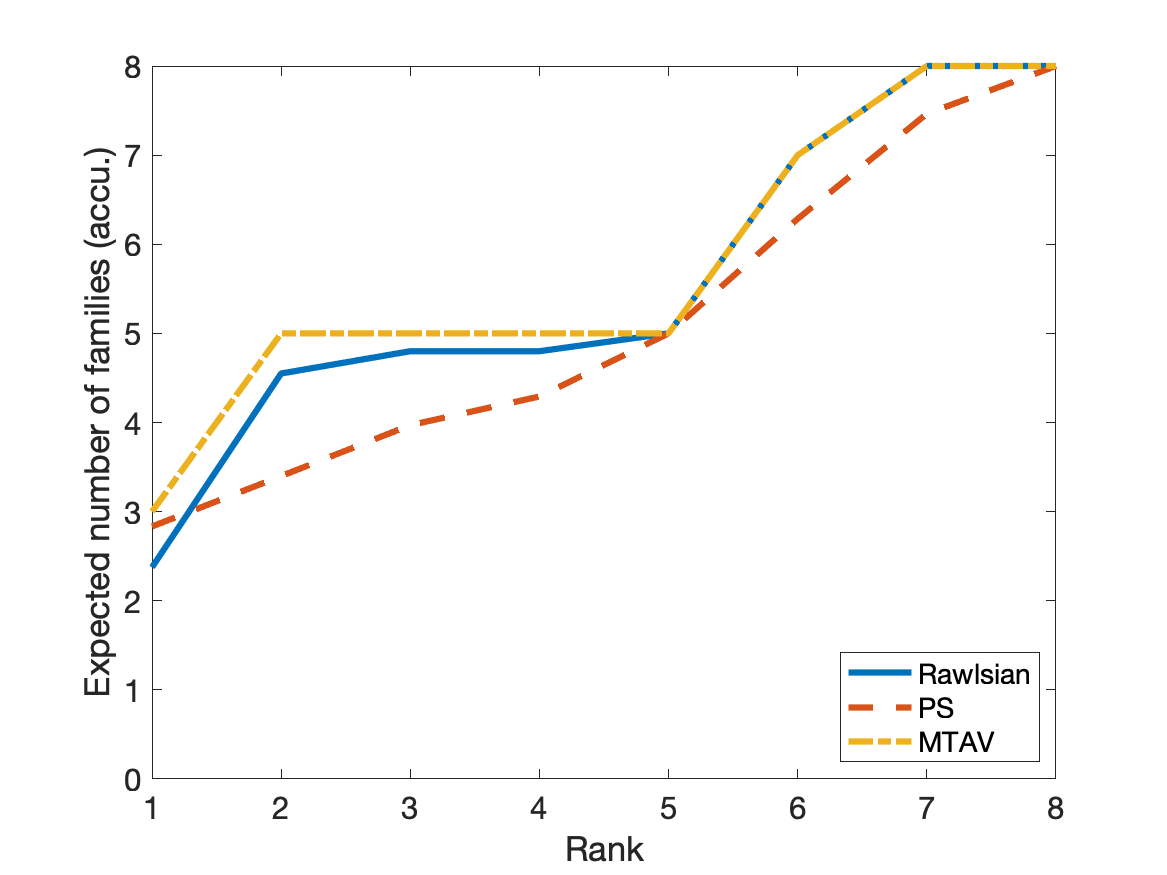}
  \caption{Cooperative $C_{24}$}
  \label{fig:sub2}
\end{subfigure}
\caption{}
\end{figure}
\newpage

\section{Proof of Proposition \ref{prop:LB_Rawls}} \label{proofs_large2}

We use the notation introduced in the proof of Proposition~\ref{large}. Fix \(\varepsilon\in(0,1)\), and for each \(n\), let
\[
k_n^-=\lfloor (1-\varepsilon)\log n\rfloor.
\]

\begin{lemma}
\[
\mathbb P(\widetilde X_{n,k_n^-}=1)\to 0.
\]
\end{lemma}

\begin{proof}
For this, it is enough to show that, with probability converging to one,
$\widetilde{\mathcal G}(n,k_n^-)$ has an isolated object. Fix an object $o$. The probability that a given agent does not rank $o$ among her top
$k_n^-$ objects is
\[
1-\frac{k_n^-}{n}.
\]
Since agents' preferences are independently drawn, the probability that no agent ranks
$o$ among her top $k_n^-$ objects is
\[
\Pr(o \text{ is isolated})=\left(1-\frac{k_n^-}{n}\right)^n.
\]

Let $Z_n$ denote the number of isolated objects in
$\widetilde{\mathcal G}(n,k_n^-)$.

We split the proof into three claims.

\begin{claim}\label{large2}
\[
\mathbb E[Z_n] \to+\infty.
\]

\end{claim}

\begin{proof}

For each object \(o\in O\), define the indicator variable
\[
I_o=
\begin{cases}
1 & \text{if object } o \text{ is isolated},\\
0 & \text{otherwise}.
\end{cases}
\]
Then the total number of isolated objects is
\[
Z_n=\sum_{o\in O} I_o.
\]
Taking expectations,
\[
\mathbb E[Z_n]
=
\mathbb E\left[\sum_{o\in O} I_o\right]
=
\sum_{o\in O}\mathbb E[I_o]
= 
\sum_{o \in O}\Pr(o \text{ is isolated}).
\]

Hence, 
\[
\mathbb E[Z_n]
=
n\left(1-\frac{k_n^-}{n}\right)^n.
\]

Since $k_n^-=(1-\varepsilon)\log n - \delta_n$, where $0 \leq \delta_n < 1$, we have
\[
\left(1-\frac{k_n^-}{n}\right)^n
\sim
e^{-k_n^-}
=
n^{-(1-\varepsilon)}e^{\delta_n}.
\]

Since $\delta_n \geq 0$, we know that $e^{\delta_n} \geq 1$. Then,
\[
\mathbb E[Z_n]\sim n^\varepsilon e^{\delta_n} \geq n^\varepsilon\to+\infty.
\]

\end{proof}

\begin{claim}
\[
\frac{\operatorname{Var}(Z_n)}{\mathbb E[Z_n]^2}\to 0.
\]
\end{claim}

\begin{proof}
Let
\[
p_n\equiv \Pr(o \text{ is isolated})
=
\left(1-\frac{k_n^-}{n}\right)^n.
\]
For two distinct objects \(o\neq o'\), the probability that a given agent ranks neither
\(o\) nor \(o'\) among her top \(k_n^-\) objects is
\[
\frac{\binom{n-2}{k_n^-}}{\binom{n}{k_n^-}}
=
\frac{(n-k_n^-)(n-k_n^- -1)}{n(n-1)}.
\]
Therefore, since agents' preferences are independently drawn,
\[
q_n
\equiv
\Pr(o \text{ and } o' \text{ are both isolated})
=
\left(
\frac{(n-k_n^-)(n-k_n^- -1)}{n(n-1)}
\right)^n.
\]

Note that the expression
\[
\mathbb E[Z_n(Z_n-1)]
\]
is the expected number of ordered pairs of isolated objects.

Then,
\[
\mathbb E[Z_n(Z_n-1)]
=
n(n-1)q_n.
\]

We now compute $\operatorname{Var}(Z_n)$,

\[
\operatorname{Var}(Z_n)
=
\mathbb E[Z_n^2]-\mathbb E[Z_n]^2
=
\mathbb E[Z_n(Z_n-1)]+\mathbb E[Z_n]-\mathbb E[Z_n]^2.
\]
Dividing by \(\mathbb E[Z_n]^2=(np_n)^2\), we obtain
\[
\frac{\operatorname{Var}(Z_n)}{\mathbb E[Z_n]^2}
=
\frac{(n-1)}{n}\frac{q_n}{p_n^2} +
\frac{1}{np_n}
-1.
\]

By Claim \ref{large2}, the second term converges to zero because \(np_n=\mathbb E[Z_n]\to+\infty\).

To finish the proof of the claim we show that
\[
\frac{q_n}{p_n^2}
=
\left(
\frac{n(n-k_n^- -1)}{(n-1)(n-k_n^-)}
\right)^n
\to 1.
\]

Indeed, note that
\[
(n-1)(n-k_n^-)=n(n-k_n^--1)+k_n^-.
\]
Therefore,
\[
\frac{n(n-k_n^--1)}{(n-1)(n-k_n^-)}
=
1-\frac{k_n^-}{(n-1)(n-k_n^-)}.
\]
Hence,
\[
\frac{q_n}{p_n^2}
=
\left(
1-\frac{k_n^-}{(n-1)(n-k_n^-)}
\right)^n.
\]

Since
\[
k_n^-=\lfloor(1-\varepsilon)\log n\rfloor,
\]
we have that \(k_n^-=O(\log n)\). Therefore,

\[
\frac{k_n^-}{(n-1)(n-k_n^-)}
=
\frac{O(\log n)}{O(n^2)}
=
O\left(\frac{\log n}{n^2}\right).
\]

Because $\frac{k_n^-}{(n-1)(n-k_n^-)} \to 0$ as $n \to \infty$, we can apply the standard asymptotic equivalence $(1-x)^n \sim e^{-n x}$. Substituting our bound into the exponent gives:
\[
n \frac{k_n^-}{(n-1)(n-k_n^-)} = n \cdot O\left(\frac{\log n}{n^2}\right) = O\left(\frac{\log n}{n}\right).
\]
Since $\frac{\log n}{n} \to 0$, it follows that 
\[
\frac{q_n}{p_n^2} = e^{-O\left(\frac{\log n}{n}\right)} \to e^0 =  1.
\]
which implies
\[
\frac{\operatorname{Var}(Z_n)}{\mathbb E[Z_n]^2}
=
\frac{1}{np_n}
+
\frac{(n-1)}{n}\frac{q_n}{p_n^2} 
-1 \rightarrow 0.
\]

\end{proof}

\begin{claim} \label{claimlarge3}
\[
\Pr(Z_n>0)\to 1.
\]
\end{claim}

\begin{proof}
Since \(Z_n\geq 0\) and \(\mathbb E[Z_n]>0\), if \(Z_n=0\), then
\[
\left|Z_n-\mathbb E[Z_n]\right|
=
\left|0-\mathbb E[Z_n]\right|
=
\mathbb E[Z_n].
\]
Therefore, whenever the event \(\{Z_n=0\}\) occurs, the event
\[
\left\{
\left|Z_n-\mathbb E[Z_n]\right|
\geq
\mathbb E[Z_n]
\right\}
\]
also occurs. Equivalently,
$
\{Z_n=0\}
\subseteq
\left\{
\left|Z_n-\mathbb E[Z_n]\right|
\geq
\mathbb E[Z_n]
\right\}.
$
Taking probabilities on both sides gives
\[
\Pr(Z_n=0)
\leq
\Pr\left(
\left|Z_n-\mathbb E[Z_n]\right|
\geq
\mathbb E[Z_n]
\right).
\]
By Chebyshev's inequality,
\[
\Pr\left(
\left|Z_n-\mathbb E[Z_n]\right|
\geq
\mathbb E[Z_n]
\right)
\leq
\frac{\operatorname{Var}(Z_n)}{\mathbb E[Z_n]^2}.
\]
Hence,
\[
\Pr(Z_n=0)
\leq
\frac{\operatorname{Var}(Z_n)}{\mathbb E[Z_n]^2}.
\]
Therefore, as
\[
\frac{\operatorname{Var}(Z_n)}{\mathbb E[Z_n]^2}\to 0,
\]
we obtain
\[
\Pr(Z_n=0)\to 0,
\]
and consequently
\[
\Pr(Z_n>0)\to 1.
\]

\end{proof}

We finish the proof of the lemma by noting that, by Claim \ref{claimlarge3}, the graph $\widetilde{\mathcal G}(n,k_n^-)$ has an isolated object with probability converging to one, and therefore no perfect matching. Thus,
\[
\mathbb P(\widetilde X_{n,k_n^-}=1)
=
\mathbb P\left(
\widetilde{\mathcal G}(n,k_n^-)\text{ has a perfect matching}
\right)
\to 0.
\]

\end{proof}

\begin{proof}[Proof of Proposition \ref{prop:LB_Rawls}]
Note that the event \(\{\widetilde X_{n,k}=1\}\) is nested in \(k\), for every \(k\leq k_n^-\). Indeed, if \(k\leq k_n^-\), then every edge of \(\widetilde{\mathcal G}(n,k)\) is also an edge of \(\widetilde{\mathcal G}(n,k_n^-)\). Hence, if \(\widetilde{\mathcal G}(n,k)\) has a perfect matching, then so does \(\widetilde{\mathcal G}(n,k_n^-)\).

Then
\[
\mathbb P(\widetilde X_{n,k}=1)
\leq
\mathbb P(\widetilde X_{n,k_n^-}=1).
\]

By the proof of Proposition \ref{large}, expression \ref{fact1} we have that:

\[
\mathbb{E}(r^{\mathrm{Rawls}}_{\max})
=
n-\sum_{k=1}^{n-1}\mathbb{P}(\widetilde X_{n,k}=1).
\]

Using this identity we obtain:

\[
\begin{aligned}
\mathbb{E}(r^{\mathrm{Rawls}}_{\max})
&=
1+\sum_{k=1}^{n-1}\left[1-\mathbb{P}(\widetilde X_{n,k}=1)\right]\\
&\geq
\sum_{k=1}^{k_n^-}\left[1-\mathbb{P}(\widetilde X_{n,k}=1)\right]\\
&\geq
k_n^-\left[1-\mathbb{P}(\widetilde X_{n,k_n^-}=1)\right].
\end{aligned}
\]
Because \(k_n^-\to\infty\) and \(\mathbb{P}(\widetilde X_{n,k_n^-}=1)\to 0\), it follows that
\[
\mathbb{E}(r^{\mathrm{Rawls}}_{\max})\to\infty.
\]
\end{proof}

\bibliographystyleapp{apalike} 
\bibliographyapp{rawls} 
\end{document}